\documentclass[11pt,english]{article}

\usepackage[table]{xcolor}
\usepackage{lineno}
\usepackage{amsmath,amssymb,amsthm}
\usepackage{multicol}
\usepackage[margin=1in]{geometry}
\usepackage{graphicx,color}
\usepackage{enumitem}
\usepackage{fullpage}
\usepackage[noblocks]{authblk}
\usepackage{tcolorbox}
\usepackage{babel}
\usepackage{wrapfig}
\usepackage{MnSymbol}
\usepackage{mdframed}
\usepackage{mathtools}
\usepackage{floatrow}
\usepackage{multirow}
\usepackage{tablefootnote}
\usepackage[leftcaption]{sidecap}
\usepackage{hhline}

\usepackage[ruled,linesnumbered,vlined]{algorithm2e}
	\usepackage{tablefootnote}
	\usepackage{thm-restate}
\usepackage{bbding}
\usepackage{pifont}
\usepackage{nicefrac}
\usepackage{undertilde}
\renewcommand{\arraystretch}{1.2}

\definecolor{Darkblue}{rgb}{0,0,0.4}
\definecolor{Brown}{cmyk}{0,0.61,1.,0.60}
\definecolor{Purple}{cmyk}{0.45,0.86,0,0}
\definecolor{Darkgreen}{rgb}{0.133,0.543,0.133}

\usepackage[colorlinks,linkcolor=Darkblue,filecolor=blue,citecolor=blue,urlcolor=Darkblue,pagebackref]{hyperref}
\usepackage[nameinlink]{cleveref}

\usepackage[colorinlistoftodos,prependcaption,textsize=tiny]{todonotes}

\newcommand{\red}[1]{{#1}}
\newcommand{\atodoAfter}[1]{}
\newcommand{\aidea}[1]{}

\newif\ifdraft 
\draftfalse

\newcommand{\namedref}[2]{\hyperref[#2]{#1~\ref*{#2}}}
\newcommand{\propref}[1]{\hyperref[#1]{property~(\ref*{#1})}}
\newcommand{\theoremref}[1]{\namedref{Thm.}{#1}}
\newcommand{\cororef}[1]{\namedref{Cor.}{#1}}

\newlength{\Oldarrayrulewidth}

\newcommand{\real}{\mathbb{R}}
\newcommand{\wsp}{\mathtt{Ws}}

\newtheorem{theorem}{Theorem}
\newtheorem{lemma}{Lemma}
\newtheorem{definition}{Definition}
\newtheorem{claim}{Claim}
\newtheorem{observation}{Observation}
\newtheorem{corollary}{Corollary}

\newtheorem{remark}{Remark}

\newcommand{\poly}{\mathrm{poly}}
\newcommand{\polylog}{\mathrm{polylog}}

\newcommand{\R}{\mathbb{R}}
\newcommand{\Z}{\mathbb{Z}}
\newcommand{\N}{\mathbb{N}}

\newcommand{\est}{{\rm est}}

\newcommand{\supp}{\mathrm{supp}}
\newcommand{\home}{\mbox{\bf home}}
\newcommand{\diam}{\mathrm{diam}}

\newcommand{\SPD}{\textsf{SPD}\xspace}

\newcommand{\SPDdepth}{\textsf{SPDdepth}\xspace}
\newcommand{\lca}{\mathsf{lca}}
\newcommand{\vsucc}{\mathrm{vs}}
\newcommand{\graycell}{\cellcolor[HTML]{EFEFEF}}

\newcommand{\ourresults}{\vspace{5pt}\hspace{-17pt}\textbf{Our contribution. }}

\newtheorem{fact}{Fact}
\definecolor{forestgreen}{rgb}{0.13, 0.55, 0.13}

\SetCommentSty{mycommfont}

\DeclareMathOperator*{\argmin}{arg~min}

\def\eps{\epsilon}

\DeclareMathAlphabet{\mathpzc}{OT1}{pzc}{m}{it}
\newcommand{\etal}{{\em et al. \xspace}}

\newlength{\dhatheight}

\newcommand {\ignore} [1] {}

\newcommand{\initOneLiners}{%
	\setlength{\itemsep}{0.2pt}
	\setlength{\parsep }{0.2pt}
	\setlength{\topsep }{0.2pt}
}

\newenvironment{OneLiners}[1][\ensuremath{\bullet}]
{\begin{list}
		{#1}
		{\initOneLiners}}
	{\end{list}}

\title{Labeled Nearest Neighbor Search and Metric Spanners\\ via Locality Sensitive Orderings}

\author{Arnold Filtser\thanks{	Email: \texttt{arnold.filtser@biu.ac.il}. This research was supported by the ISRAEL SCIENCE FOUNDATION (grant No. 1042/22).}}
\affil{Bar-Ilan University}
\date{}
\begin{document}
\maketitle
\begin{abstract}
Chan, Har-Peled, and Jones [SICOMP 2020] developed locality-sensitive orderings (LSO) for Euclidean space. A $(\tau,\rho)$-LSO is a collection $\Sigma$ of orderings such that for every $x,y\in\mathbb{R}^d$ there is an ordering $\sigma\in\Sigma$, where all the points between $x$ and $y$ w.r.t. $\sigma$ are in the $\rho$-neighborhood of either $x$ or $y$. 
In essence, LSO allow one to reduce problems to the $1$-dimensional line. 
Later, Filtser and Le [STOC 2022] developed LSO's for doubling metrics, general metric spaces, and minor free graphs.

For Euclidean and doubling spaces, the number of orderings in the LSO is exponential in the dimension, which made them mainly useful for the low dimensional regime.
In this paper, we develop new LSO's for Euclidean, $\ell_p$, and doubling spaces that allow us to trade larger stretch for a much smaller number of orderings.
We then use our new LSO's (as well as the previous ones) to construct path reporting low hop spanners, fault tolerant spanners, reliable spanners, and light spanners for different metric spaces.

While many nearest neighbor search (NNS) data structures were constructed for metric spaces with implicit distance representations (where the distance between two metric points can be computed using their names, e.g. Euclidean space), for other spaces almost nothing is known. 
In this paper we initiate the study of the labeled NNS problem, where one is allowed to artificially assign labels (short names) to metric points. We use LSO's to construct efficient labeled NNS data structures in this model. 

\end{abstract}

\newpage

\setcounter{secnumdepth}{5}
		\setcounter{tocdepth}{3} \tableofcontents

    \newpage
    \pagenumbering{arabic}

\section{Introduction}
\subsection{Locality Sensitive Ordering}
Chan, Har-Peled, and Jones \cite{CHJ20} recently introduce a new and powerful tool into the   algorithmist's toolkit, called \emph{locality sensitive ordering} (abbreviated LSO).
LSO provides an order over the points of a metric space $(X,d_X)$, this order being very useful, as it helps to store, sort, and search the data (among other manipulations).
\begin{restatable}[$(\tau,\rho)$-LSO]{definition}{defLSO}
	\label{def:LSOclassic}
	Given a metric space $(X,d_{X})$, we say that a collection $\Sigma$
	of orderings is a $(\tau,\rho)$-LSO if
	$\left|\Sigma\right|\le\tau$, and for every $x,y\in X$, there is
	a linear ordering $\sigma\in\Sigma$ such that (w.l.o.g.\footnote{That is either  $x\preceq_{\sigma}y$ or  $y\preceq_{\sigma}x$, and the guarantee holds w.r.t. all the points between $x$ and $y$ in the order $\sigma$.\label{foot:wlogOrder}}) $x\preceq_{\sigma}y$ and the points between $x$ and $y$ w.r.t. $\sigma$ could be partitioned
	into two consecutive intervals $I_{x},I_{y}$ where $I_{x}\subseteq B_{X}(x,\rho\cdot d_{X}(x,y))$	and $I_{y}\subseteq B_{X}(y,\rho\cdot d_{X}(x,y))$. $\rho$ is called the \emph{stretch} parameter. 
\end{restatable}
Morally speaking, given a problem, LSO can reduce it from a general and complicated space to a much simpler space: $1$-dimensional line.
Chan \etal \cite{CHJ20} constructed $\left(O_d(\eps^{-d})\cdot\log\frac1\eps,\eps\right)$-LSO for the $d$-dimensional Euclidean space. They used their LSO to design simple dynamic algorithms for approximate nearest neighbor search, bichromatic closest pair, MST, spanners, and fault-tolerant spanners.
Later, Buchin, Har{-}Peled, and Ol{\'{a}}h \cite{BHO19,BHO20} constructed reliable spanners using LSO, obtaining considerably superior results compared with previous techniques.

Filtser and Le \cite{FL22} generalized Chan \etal \cite{CHJ20} result to doubling spaces,\footnote{A metric space $(X, d)$ has doubling dimension $d$ if every ball of radius $2r$ can be 	covered by $2^{d}$ balls of radius $r$.\label{foot:doubling}} showing that every metric space with doubling dimension $d$ admits a $\left(\eps^{-O(d)},\eps\right)$-LSO.
Furthermore, they generalized the concept of LSO to other metric spaces, defining the two related notions of \emph{triangle-LSO} (which turn to be useful for general metric spaces), and \emph{left-sided LSO} (which turn to be useful for topologically restricted graphs).
Here, instead of presenting the left-sided LSO's of \cite{FL22}, we introduce the closely related notion of \emph{rooted-LSO}, which has some additional structure. All the results and constructions for left-sided LSO in \cite{FL22} hold for rooted LSO as well. We refer to \cite{FL22} for a comparison between the different notions, and to \Cref{fig:LSO} for an illustration.
\begin{definition}[$(\tau,\rho)$-Triangle-LSO]\label{def:TriangleLSO}
	Given a metric space $(X,d_{X})$, we say that a collection $\Sigma$
	of orderings is a $(\tau,\rho)$-triangle-LSO if
	$\left|\Sigma\right|\le\tau$, and for every $x,y\in X$, there is
	an ordering $\sigma\in\Sigma$ such that (w.l.o.g.$^{\ref{foot:wlogOrder}}$)  $x\prec_{\sigma}y$, and for every $a,b\in X$ such that $x\preceq_{\sigma}a\preceq_{\sigma}b\preceq_{\sigma}y$ it holds that $d_X(a,b)\le\rho\cdot d_X(x,y)$.
\end{definition}
\begin{definition}[$(\tau,\rho)$-rooted-LSO]\label{def:RootedLeftSided}
	Given a metric space $(X,d_{X})$, we say that a collection $\Sigma$
	of \emph{orderings} over subsets of $X$ is a $(\tau,\rho)$-rooted-LSO if the following hold:
	\begin{OneLiners}
		\item Each point $x\in X$ belongs to at most $\tau$ orderings in $\Sigma$.
		\item Each ordering $\sigma\in\Sigma$ is associated with a point
		$x_{\sigma}\in X$, which is the first in the order, and such that the ordering is w.r.t. distances from
		$x_{\sigma}$ (i.e. $y\prec_\sigma z~\Rightarrow~ d_X(x_\sigma,y)\le d_X(x_\sigma,z)$).
		\item For every pair of points $u,v$, there is some $\sigma\in\Sigma$ containing both $x,y$, and such that  
		$d_{G}(u,x_{\sigma})+d_{G}(x_{\sigma},v)\le\rho\cdot d_{G}(u,v)$.
	\end{OneLiners} 
\end{definition}

\begin{figure}[t]
	\centering
	\includegraphics[width=.84\textwidth]{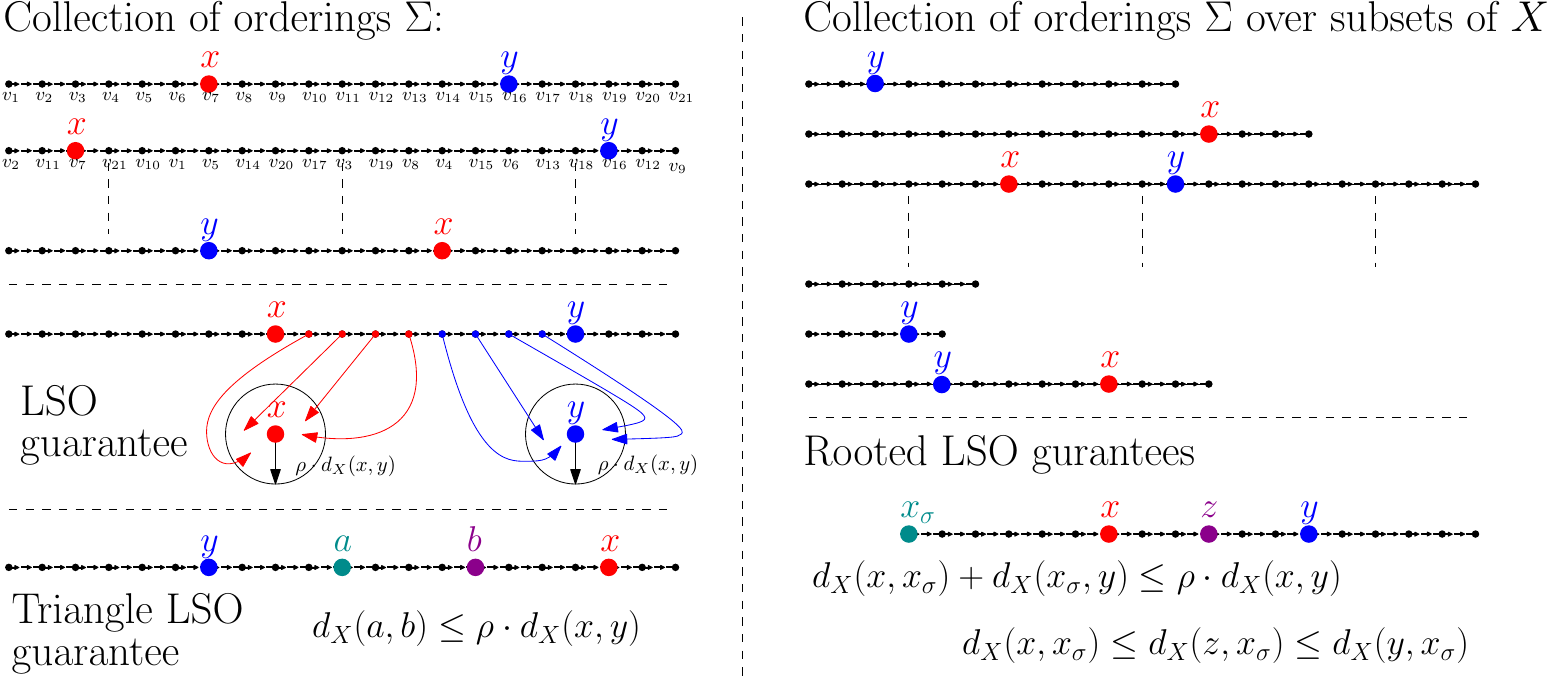}
	\caption{\footnotesize{Illustration of different types of LSO.
	}}
	\label{fig:LSO}
\end{figure}

Filtser and Le \cite{FL22} constructed triangle LSO for general metrics, and rooted LSO for the shortest path metrics of trees, treewidth graphs, planar graphs, and graph excluding a fixed minor. 
They used their LSO's to construct oblivious reliable spanners for the respective metric spaces, considerably improving previous constructions (that used different techniques).
All the known results on LSO's are summarized in \Cref{tab:LSO}.

\begin{table}[]
	\begin{tabular}{|l|l|l|l|l|}
		\hline
		LSO type                                             & Metric Space             & \# of orderings ($\tau$)                         & Stretch ($\rho$)         & Ref                                                  \\ \hline
		\multicolumn{1}{|c|}{\multirow{2}{*}{(Classic) LSO}} & Euclidean space $\mathbb{R}^d$ & $O_d(\epsilon^{-d})\cdot \log\frac{1}{\epsilon}$    \hspace{22pt} $^{(*)}$    & $\eps$           & \cite{CHJ20}                        \\ \cline{2-5} 
		\multicolumn{1}{|c|}{}                               & Doubling dimension $d$  & $\eps^{-O(d)}$                                   & $\eps$           & \cite{FL22}              \\ \hline
		
		\multirow{4}{*}{Triangle-LSO}                                        & General metric           & 
		$O(n^{\frac1k}\cdot\log n\cdot\frac{k^2}{\eps}\cdot\log\frac{k}{\eps})$
		& $2k+\eps$ & \cite{FL22} \\ \hhline{~|-|-|-|-|}
		&\graycell Euclidean space $\mathbb{R}^d$         &\graycell $e^{\frac{d}{2t^{2}}\cdot(1+\frac{2}{t^{2}})}\cdot\tilde{O}(\frac{d^{1.5}}{\eps\cdot t})$ &\graycell $(1+\eps)t$ &\graycell \theoremref{thm:LSOEuclideanLargeStretch} \\ \hhline{~|-|-|-|-|}
		&\graycell $\ell_p^d$ for $p\in[1,2]$         &\graycell $e^{O(\nicefrac{d}{t^{p}})}\cdot\tilde{O}(d)$ &\graycell $t$ &\graycell \theoremref{thm:LSOLp12} \\ \hhline{~|-|-|-|-|}
		&\graycell $\ell_p^d$ for $p\in[2,\infty]$         &\graycell $\tilde{O}(d)$ &\graycell $d^{1-\frac1p}$ &\graycell \cororef{cor:LpLSO} \\ \hhline{~|-|-|-|-|}
		&\graycell Doubling dimension $d$           &\graycell $2^{O(\nicefrac{d}{t})}\cdot d\cdot\log^{2}t$ &\graycell $t$ &\graycell \theoremref{thm:LSOdoublingLargeStretch} \\ \hline

		\multirow{3}{*}{Rooted LSO}                      & Tree                     & $\log n$                                         & $1$              & \cite{FL22}         \\ \cline{2-5} 
		& Treewidth $k$            & $k\cdot \log n$                                        & $1$              & \cite{FL22}         \\ \cline{2-5} 
		& Planar /  fixed minor free             & $O(\frac{1}{\eps}\cdot \log^2 n)$                         & $1+\eps$         & \cite{FL22}         \\ 
		\hline
	\end{tabular}
	\caption{\small{Summary of all known results, on all the different types of locality sensitive orderings (LSO). $k\in\N$, $t>1$, $\eps\in(0,1)$ is an arbitrarily small parameter.
	$^{(*)}$ $O_d$ hides an arbitrary function of $d$, the number of orderings in \cite{CHJ20} LSO is $O_d(\epsilon^{-d})\cdot \log\frac{1}{\epsilon}=2^{O(d)}\cdot d^{\frac{3}{2}d}\cdot\epsilon^{-d}\cdot\log\frac{1}{\epsilon}$.
}}
	\label{tab:LSO}
\end{table}
	Previously constructed LSO for the Euclidean space \cite{CHJ20}, as well as for metric spaces with doubling dimension $d$ \cite{FL22}, have exponential dependency on the dimension in their cardinality, a phenomena often referred to as ``the curse of dimensionality''. 
When the dimension is high, it can be a major obstacle. 
Indeed, the distances induced by $n$ point in an $O(\log n)$-dimensional Euclidean space induce a metric space which is much more structured than a general metric space. Therefore one might expect them to admit better LSO. However, using \cite{CHJ20} one can only construct $(n,\eps)$-LSO (note that every metric admits $(\lceil\frac n2\rceil,0)$-LSO \footnote{This follows from a theorem by Walecki \cite{Alspach08} who showed that the edges of the $K_n$ clique graph can be partitioned into $\lceil\frac n2\rceil$ Hamiltonian paths.}).

Every $n$ point metric space has doubling dimension at most $\log n$. Consider the case where the doubling dimension is somewhat large (e.g. $\sqrt{\log n}$) but not maximal. It is much more structured than general metric, however the only construction we have \cite{FL22} gives us $\eps^{-O(d)}$ orderings, which might be too large. If we insist that the number of orderings will be small, could we take advantage of the doubling structure to construct better LSO then for general metrics?

\ourresults
In this paper we construct new triangle-LSO for high dimensional spaces. We then present many applications for the newly constructed LSO's, as well as for the previously constructed LSO's.
Previous and new LSO construction are summarized in \Cref{tab:LSO}.

\begin{restatable}
	[]{theorem}{HighDimLSO}
	\label{thm:LSOEuclideanLargeStretch}
	\sloppy For every $t\in[4,2\sqrt{d}]$, $\delta\in(0,1]$, and $d\ge 1$, the $d$-dimensional Euclidean space $\mathbb{R}^{d}$ admits $\left(O\left(\frac{d^{1.5}}{\delta\cdot t}\cdot\log(\frac{2\sqrt{d}}{t})\cdot\log\frac{d}{\delta}\cdot e^{\frac{d}{2t^{2}}\cdot(1+\frac{2}{t^{2}})}\right),(1+\delta)t\right)$-triangle LSO.
\end{restatable}
For $t=\frac23\sqrt{d}$ and $\delta=\frac12$, we obtain $\left(O(d\log d),\sqrt{d}\right)$-triangle LSO. In particular, for every set of $n$ points in $\ell_2$, using the Johnson Lindenstrauss dimension reduction \cite{JL84}, for every fixed $t>1$, we can construct $\left(n^{\frac{1}{t^{2}}}\cdot\tilde{O}(\frac{\log^{1.5} n}{t}),O(t)\right)$-triangle LSO, or $\left(\tilde{O}(\log n),O(\sqrt{\log n})\right)$-triangle LSO, a quadratic improvement compared with general $n$-point metric spaces!

Interestingly, in \Cref{cor:LpLSO} we show that the $\left(O(d\log d),\sqrt{d}\right)$-triangle LSO $\Sigma$ for $\ell_2$, is in the same time also a $\left(O(d\log d),d^{\frac1p}\right)$-triangle LSO for $\ell_p$ where $p\in[1,2]$, and $\left(O(d\log d),d^{1-\frac1p}\right)$-triangle LSO for $\ell_p$ where $p\in[2,\infty]$.
For $p\in [1,2]$, we generalize \Cref{thm:LSOEuclideanLargeStretch} to $\ell_p$ spaces to get the entire \#ordering-stretch trade-off.
Finally, we generalize \Cref{thm:LSOEuclideanLargeStretch} to general metric spaces with doubling dimension $d$.
\begin{restatable}
	[]{theorem}{HighDimLSOLP}
	\label{thm:LSOLp12}
	\sloppy For every $p\in[1,2]$, $t\in[5,d^{\frac1p}]$ and $d\ge 1$, the $d$-dimensional $\ell_p$ space admits $\left( e^{O(\frac{d}{t^{p}})}\cdot\tilde{O}(d),t\right)$-triangle LSO.
\end{restatable}
\begin{restatable}
	[]{theorem}{DoublingLSO}
	\label{thm:LSOdoublingLargeStretch}
	Given a metric space $(X,d_X)$ with doubling dimension $d$, and parameter $t\in[\Omega(1),d]$, $X$ admits
	$\left(2^{O(\nicefrac{d}{t})}\cdot d\cdot\log^{2}t, t\right)$-triangle  LSO.
\end{restatable}
For $t=d$, we get $(\tilde{O}(d),d)$-triangle LSO, again much better then general metric spaces!

\subsection{Labeled Nearest Neighbor Search}

Nearest neighbor search (abbreviated NNS) is a classical and  fundamental task used in numerous domains
including machine learning, clustering, document retrieval, databases, statistics, data compression, database queries, computational biology, data mining, pattern recognition, and many others. 
In the NNS problem we are given a set $P$ of points in a metric space $(X,d_X)$. The goal is to construct a succinct  data structure that given a query point $q\in X$, quickly returns a point $p\in P$ closest to $q$ (i.e. $\argmin_{p\in P}d_X(p,q)$).
In order to keep the size of the data structure, and the query time small, usually approximation is allowed. In the $t$-approximate nearest neighbor problem (abbreviated $t$-NNS) the goal is to return a point $p$ at distance at most $t\cdot \min_{p\in P}d_X(p,q)$ from $q$.
The problem was extensively studied in $\ell_p$ spaces (see the survey \cite{AIR18}), and also in various norm spaces over $\R^d$ (see e.g. \cite{ANNRW17,ANRW21}).  
NNS data structures were also constructed beyond normed spaces. Some examples are Earth-Mover distance \cite{IT03}, Edit Distance \cite{OR07,AIR18}, and Fr\'echet distance \cite{Indyk02,DS17,EP18,FFK20}.
We observe that a crucial property shared by these examples, is that they have an ``implicit distance representation''.
That is, it is possible  to compute the distance between two points using only their names (e.g. the coordinates values in $\R^d$ used as names:  $d_{\R^d}\left((x_1,\dots,x_d),(y_1,\dots,y_d)\right)=\left\|(x_1,\dots,x_d)-(y_1,\dots,y_d)\right\|_2$). 

For general metric spaces, Krauthgamer and Lee \cite{KL05} introduced the \emph{black box model}. Here one is given access to an exact distance oracle \footnote{An exact distance oracle $D$ is a data structure that given two points $x,y$, returns $\est(x,y)=d_X(x,y)$. A distance oracle of stretch $t$ returns a value $\est(x,y)$ in $[d_X(x,y),t\cdot d_X(x,y)]$.\label{foot:DistanceOracle}} ${\rm DO}$ that answer distance queries in $t_{\rm DO}$ time. They showed that one can construct an efficient $(1+\eps)$-NNS (that is with polynomial space, and polylogarithmic query time), if and only if the doubling dimension of $X$ is at most $O(\log\log n)$.

\begin{wrapfigure}{r}{0.18\textwidth}
	\begin{center}
		\vspace{-20pt}
		\includegraphics[width=0.9\textwidth]{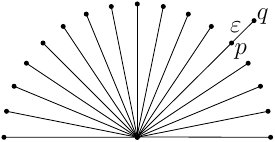}
		\vspace{-5pt}
	\end{center}
	\vspace{-15pt}
\end{wrapfigure}
Indeed, for metric spaces with large doubling dimension, distance queries provide very limited information. Consider for example the case where the input metric is the star graph (inducing uniform metric on the leaves, see illustration on the right), and the query point attached to one of the leaves with an edge of infinitesimal weight, one must query all the points before finding any finite approximation to the nearest neighbor.

An interesting case studied by Abraham, Chechik, Krauthgamer, and Wieder \cite{ACKW15} is that of planar graphs. 
Here we are given a huge weighted planar graph $G=(V,E,w)$ with $N$ vertices, and a subset of $n$ vertices $X\subseteq V$. The goal is to solve the $(1+\eps)$-NNS problem w.r.t. the shortest path metric $d_G$, input set $X$ and queries from $V$. 
Assuming access to an exact distance oracle$^{\ref{foot:DistanceOracle}}$ ${\rm DO}$ that answer distance queries in $t_{\rm DO}$ time, and given a planar graph $G$ of maximum degree $\Delta$, Abraham \etal \cite{ACKW15} constructed a $(1+\eps)$-NNS data structure for planar graph of size $n\cdot O(\eps^{-1} \cdot\log n\cdot\log N+\Delta\cdot\log^2n)$ and query time $O((\eps^{-1}\cdot\log\log n+t_{\rm DO})\cdot\log n\cdot\log N+\log n\cdot\Delta\cdot t_{\rm DO})$.

Linear dependence on the degree is a very limiting requirement, as planar graphs have  apriori unbounded degree. Moreover, exact distance computations (even in planar graphs) are time consuming, and if the graph is big enough could be infeasible. Exact distance oracle is a highly non-trivial assumption, it is an expensive data structure, \footnote{After a long line of work, the state of the art (by Long and Pettie \cite{LP21}) requires either super-linear space $N^{1+o(1)}$, or very large query time $N^{o(1)}$, both quite undesirable.} better to be avoided. 
One might hope to relax either the max degree assumption, or to use the much more reasonable and efficient data structure of approximate distance oracle  \cite{Thorup04,Kle02,LW21}.
Unfortunately, Abraham  \etal \cite{ACKW15} showed both assumptions to be necessary.
Specifically, the dependence on the degree is necessary, as every NNS data structure with space at most $O(\frac{N}{\Delta\log_\Delta n})$ 
must probe the distance oracle at least $\Omega(\Delta\log_\Delta n)$ times. Furthermore, they show that if one is only given access to a $(1+\eps)$-distance oracle, then there is a planar graph (in fact a tree) with maximum degree $O(\log n)$, aspect ratio $O(\frac{\log n}{\eps})$, $N\le n^2$, and the NNS data structure is forced to make $\Omega(n)$ queries to the distance oracle.

To conclude this discussion, exact distance oracle (assumed both by the black box model \cite{KL05} and \cite{ACKW15}) is an expensive data structure, which enables us to construct efficient NNS only under very limiting assumptions (small doubling dimension / constant maximum degree in planar graphs). 
On the other hand in many metric spaces with ``implicit distance representation'' efficient NNS were constructed. 
The crux is that the information stored in the name (e.g. coordinate values)  used to preform various manipulations on the data, in addition to distance computation.
What if in planar graphs, or even in completely general metric spaces, we could choose the names of the metric points, or alternatively assigning each point a short label,  
would it be possible to construct efficient NNS data structures?

To answer this question, we introduce the  \emph{labeled $t$-NNS} problem.
\begin{definition}[Labeled $t$-NNS] Consider an $N$-point metric space $(X,d_{X})$, where one can assign to each point $x\in X$ an arbitrary short label $l_{x}$. 
Given a subset $P\subseteq X$ of size $n$ (unknown in advance) together with their labels $\{l_x\}_{x\in P}$ (but without access to $(X,d_X)$ or any additional information) the goal is to construct a NNS search data structure as follows: given a query $q\in X$ together with its assigned label $\ell_q$, the data structure will return a $t$-approximate nearest neighbor $p\in P$: $d_X(p,q)\le t\cdot\min_{x\in P}d_X(x,q)$.
The parameters of study are: label size,  data structure size, query time, and approximation factor $t$.
	
We also consider the scenario where the set $P$ is changing dynamically: points are added and removed from $P$.
Here we are required to maintain a data structure for $P$, while minimizing the update time (as well as all the other parameters).\footnote{For example, suppose that we have a NNS data structure for a set $P$. Dynamic NNS, should be able to efficiently update the data structure to work w.r.t. a slightly updated set $P'=P\cup\{x\}\setminus\{y\}$ instead of $P$.}
\end{definition}
In the labeled NNS model we get to assign a short label (alternatively choose a name) for each point in a big metric space $(X,d_X)$.
These labels try to imitate the natural hint provided by the name of the points themselves in metric spaces with implicit distance representation.
The main object of study here is the trade-off between label size, and the approximation of the resulting NNS.
A trivial choice of label for each point $x$ will be simply to store distances to all other points. However the label size $\Omega(N)$ is infeasible.
A more sophisticated solution is the following: fix constants $k,t\in\N$, and embed all the points in $(X,d_X)$ into $d=\tilde{O}(N^{\frac1k})$-dimensional $\ell_\infty$ \cite{Mat96,ABN11}. That is we assign each point $x$ a vector $v_x\in\R^d$ such that $\forall x,y\in X$, $d_X(x,y)\le \|v_x-v_y\|_\infty\le(2k-1)\cdot d_X(x,y)$, and use the vectors as labels.
Given an $n$ point subset $P\subseteq X$ with its respective labels (vectors), use Indyk's NNS \cite{Indyk01infty} over $\{v_x\}_{x\in P}$ to construct a NNS data structure $\mathbb{D}_{\rm Ind}$ with approximation factor $O(\log_{1+\frac1t}\log d)=O(t\cdot\log\log N)$ w.r.t. the $\ell_\infty$ vectors, space $\tilde{O}(d\cdot n^{1+\frac1t})=\tilde{O}(N^{\frac1k}\cdot n^{1+\frac1t})$ and query time $\tilde{O}(n^{1+\frac1k})$. 
Given a query $q$, we will simply query $\mathbb{D}_{\rm Ind}$ on the vector $v_q$, and on answer $v_p$ will return $p$. Note that the query time and space are the same as above, while the approximation factor will be $O(k\cdot t \cdot \log\log N)$.

\ourresults
Our results for the labeled $t$-NNS are summarized in \Cref{tab:LabeledNNSFamilies}.
We begin by proving meta \Cref{thm:rootedToNNS}  showing that $(\tau,\rho)$-rooted LSO implies a labeled $\rho$-NNS with label size $O(\tau)$, space $O(n\cdot\tau)$, query time $O(\tau)$, and update time $O(\tau\cdot \log\log N)$. As a result we conclude efficient labeled $(1+\eps)$-NNS data structures for fixed minor free graphs (and planar), and exact labeled NNS for treewidth graphs.
Another interesting corollary is an efficient labeled NNS for metrics with small \emph{correlation dimension} (a generalization of doubling, see \cite{CG12}).

\begin{table}[]
	\scalebox{0.9}{	\begin{tabular}{|l|l|l|l|l|l|}
			\hline
			Family & stretch  & label & query time                &  update time & Ref\\ \hline
			
			Planar  / {\small Minor free} & $1+\eps$ & $O(\eps^{-1}\log^{2}N)$ & $O(\eps^{-1}\log^{2}N)$ & $\eps^{-1}\cdot\tilde{O}(\log^{2}N)$  &
			\cororef{cor:PlanarNNS}\\\hline
			
			Treewidth $k$ & $1$ & $O(k\log N)$  & $O(k\log N)$  & $k\cdot\tilde{O}(\log N)$  &
			\cororef{cor:TreeWidthNNS}\\\hline
			
			Correlation dim. $k$ & $1+\eps$ & $\tilde{O}_{k,\eps}(\sqrt{N})$ & $\tilde{O}_{k,\eps}(\sqrt{N})$ & $\tilde{O}_{k,\eps}(\sqrt{N})$  &
			\cororef{cor:CorrelationNNS}\\\hline
			
			Ultrametric & $1$ & $O(\log N)$ & $O(\log\log N)$ & $O(\log\log N)$ & \theoremref{lem:UltraDeterministic}\\\hline
			
			\multirow{2}{*}{General Metric} & $8(1+\eps)k$ & $O(\frac{k}{\eps}N^{\frac{1}{k}}\cdot\log N)$
			&  $O(\frac1\eps\cdot\log\log N)$  &      $O(\frac{k}{\eps}N^{\frac1k}\cdot\log\log N)$        & \theoremref{thm:NNSfromRamsey}\\\cline{2-6}

			&$t<2k+1$ & $\tilde{\Omega}(N^{\frac1k})$ & arbitrary & arbitrary&\theoremref{thm:NNS-LB}\\\hline

			Doubling dim. $d$ & $t$ & {\small $2^{O(\nicefrac{d}{t})}\cdot\tilde{O}(d)\cdot\log N$} & {\small $2^{O(\nicefrac{d}{t})}\cdot\tilde{O}(d)\cdot{\footnotesize\log\log N}$ }& {\small $2^{O(\nicefrac{d}{t})}\cdot\tilde{O}(d)\cdot{\footnotesize\log\log N}$ }
			
			& \cororef{cor:doublingNNS}\\\hline

	\end{tabular}}
	\caption{Labeled NNS data structures for different families. 
		The second to last line (\Cref{thm:NNS-LB}) is a lower bound.
		Space is measured in machine words.		
		The label size and query time always equal. The space in all the cases above equals $n$ times the label size.\label{tab:LabeledNNSFamilies}}
\end{table}

Next, we prove a meta \Cref{thm:triangleToNNS}, showing that $(\tau,\rho)$-triangle LSO implies a labeled $2\rho$-NNS with label size $O(\tau\cdot\log N)$, space $O(n\cdot\tau\cdot\log N)$, and query and update time $O(\tau\cdot \log\log N)$.
We conclude an efficient labeled NNS for graphs with large doubling dimension.
For the high-dimensional Euclidean space, approximate nearest neighbor search was extensively studied (see the survey \cite{AIR18}, and additional discussion in \Cref{subsec:Related}). 
However, for the case of doubling metrics, NNS never went beyond $1+\eps$ approximation. In particular, in all existing solutions the query time and space have exponential dependence on the dimension (see \Cref{subsec:Related}). Thus ours are the first results in this regime, removing ``the curse of dimensionality''.

As an additional corollary of the triangle LSO to labeled NNS meta theorem one can derive a NNS of for general metric spaces which considerably improved upon the labeled NNS based on \cite{Mat96}+\cite{Indyk01infty} discussed above. 
However, the query time turns out to be somewhat large.
We provide direct constructions for labeled NNS for general metrics, getting label size $\tilde{O}(\eps^{-1}\cdot N^{\frac1k})$, stretch $8(1+\eps)k$ and very small query time: $O(\eps^{-1}\cdot\log\log N)$ (\Cref{thm:NNSfromRamsey}).
We show that the standard information theoretic bound applies for the labeled NNS as well, specifically, for stretch $t<2k+1$, the label size must be $\tilde{\Omega}(n^{\frac1k})$ (regardless of query time, \Cref{thm:NNS-LB}).
Finally, we put special focus on the regime where the stretch is $O(\log N)$. We obtain labeled NNS scheme with very short label and small query time. Most notably, assuming polynomial aspect ratio, and allowing the bound on the label to be only in expectation, we can obtain $O(1)$ label size, and $O(\log\log N)$ query time.
See \Cref{tab:GeneralNNS} for a summary.

\subsection{Spanners}
Given a metric space $(X,d_X)$, a metric \emph{spanner} is a graph $H$ over $X$ points, such that that the shortest path metric $d_H$ in $H$, closely resembles the metric $d_X$.  
Formally, a \emph{$t$-spanner} for $X$ is a weighted graph $H(X,E,w)$ that has  $w(u,v) = d_X(u,v)$ for every edge $(u,v)\in E$ and $d_H(x,y) \leq t\cdot d_X(x,y)$ for every pair of points $x,y\in X$. \footnote{Frequently the literature is concerned with graph spanners, where given a graph $G=(V,E,w)$ the goal is to find a subgraph $H$ preserving distances. Here we study metric spanners, where there is no underlying graph.}
The classic parameter of study is the trade-off between stretch and sparsity (number of edges). 
Alth{\"{o}}fer \etal \cite{ADDJS93} showed that every $n$ point metric space admits a  $2k-1$ spanner with $O(n^{1+\frac1k})$ edges, while every set of $n$ points in $\R^d$, or more generally metric space of doubling dimension $d$, admits a $(1+\eps)$-spanner with $n\cdot\eps^{-O(d)}$ edges \cite{DHN93,GGN06}. 
We refer to the book \cite{NS07}, and the survey \cite{ABSHJKS20} for an overview.

\paragraph*{Path Reporting Low Hop Spanners}
Recently, Kahalon, Le, Milenkovic, and Solomon \cite{KLMS22} studied \emph{path reporting low-hop spanners}.
While a $t$-spanner guarantees that a ``short'' path exists between every two points, such a path might be very long, and finding it is a time consuming operation.
A path reporting $t$-spanner, is a spanner accompanied with a data structure that given a query pair $\{x,y\}$, efficiently retrieves a path between $x$ and $y$ (of total weight $\leq t\cdot d_X(x,y)$).
A path $P$ with $h$ edges is called an $h$-hop path.
$H$ is an $h$-hop $t$-spanner of $X$ if for every $x,y\in X$, there is an $h$-hop path $P$ from $x$ to $y$ in $H$, such that $w(P)\leq t\cdot d_X(x,y)$.
Clearly, the time required to report a path is at least as large as the number of edges along the path, thus we wish to minimize the number of hops.

Low number of hops is a highly desirable property in network design, as each transmission causes delays, which are non-negligible when the number of transmissions is large \cite{AT11,BF18}. 
Low hop networks are also known to be more reliable \cite{BF18,WA88,RAJ12},
and used in electricity and telecommunications \cite{BF18}, and many other (practical) network design problems \cite{BCM99,BA92,GPSV03,GM03,PS03}.
Hop-constrained network approximation is often used in parallel computing \cite{Cohen00,ASZ20}, as the number of hops governs the number of required parallel rounds (e.g. in Dijkstra). 

Kahalon \etal \cite{KLMS22} constructed path reporting low-hop spanners for many spaces, such as path reporting $2$-hop $O(k)$-spanners with $O(n^{1+\frac1k}\cdot k\cdot \log n)$ edges, and $O(1)$ query time for general metrics, and path reporting $2$-hop $(1+\eps)$-spanners with $O(\frac{n}{\eps^2}\cdot \log^{3}n)$ edges and $O(\eps^{-2}\cdot \log^2 n)$ query time for planar graphs.
They showed a plethora of applications for their spanners: compact routing schemes, fault tolerant routing, spanner sparsification, approximate shortest path trees (SPT), minimum weight trees (MST), and online MST verification.

\begin{table}[t]
	\begin{tabular}{|c|c|c|c|c|}
		\hline 
		Metric family & stretch & sparsity & query time & Ref\tabularnewline
		\hline 
		\hline 
		
		\multirow{-0.5}{*}{\begin{tabular}[c]{@{}c@{}}General \\ Metric\end{tabular}} & $O(k)$ & $O\left(n^{1+\frac{1}{k}}\cdot k\cdot\log n\right)$ & $O(1)$&\cite{KLMS22} \tabularnewline
		\hhline{~|-|-|-|-|}
		&\graycell $2k-1$ & \graycell $O(n^{1+\frac{1}{k}}\cdot k)$  &\graycell $O(k)$ &\graycell\theoremref{thm:TZ_GeneralPathReportingSpanner}, \cite{TZ05}\tabularnewline\hhline{~|-|-|-|-|}

		&\graycell $(1+\eps)(4k-2)$ &\graycell $O(n^{1+\frac{1}{k}}\cdot\eps^{-1}\cdot k\cdot\log\Phi)$ &\graycell $O(\eps^{-1}\cdot\log2k)$ &\graycell \theoremref{thm:GeneralSmallQueryPathReportingSpanner} \tabularnewline\hline

		\multirow{-1}{*}{\begin{tabular}[c]{@{}c@{}}Doubling \\ Dimension  $d$\end{tabular}} & $1+\epsilon$ & $\epsilon^{-O(d)}\cdot n\cdot\log n$ & $\epsilon^{-O(d)}$&\cite{KLMS22}\tabularnewline
		\hhline{~|-|-|-|-|} 
		& \graycell$t$ &\graycell $2^{-O(\nicefrac dt)}\cdot \tilde{O}(n)$ &\graycell $2^{-O(\nicefrac dt)}\cdot d\cdot \log^2 t$ &\graycell \cororef{cor:LowHopSpannerFromLSO}\tabularnewline
		\hline

		\multirow{-1}{*}{\begin{tabular}[c]{@{}c@{}}Euclidean \\ $\R^d$\end{tabular}}  &\graycell $1+\epsilon$ &\graycell $O_d(\eps^{-d})\cdot\log\frac1\eps\cdot n\cdot\log n$ &\graycell $O_d(1)$&\graycell\cororef{cor:LowHopSpannerEuclidean}\tabularnewline
		\hhline{~|-|-|-|-|} 
		&\graycell $(1+\eps)t$ &\graycell $\tilde{O}(\frac{d^{1.5}}{\eps\cdot t})\cdot e^{\frac{2d}{t^{2}}\cdot(1+\frac{8}{t^{2}})}\cdot n\log n$ &\graycell $\tilde{O}(\frac{d^{1.5}}{\eps\cdot t})\cdot e^{\frac{2d}{t^{2}}\cdot(1+\frac{8}{t^{2}})}$&\graycell\cororef{cor:LowHopSpannerFromLSO}\tabularnewline
		\hline 
		
		$\ell_p^d$, $p\in[1,2]$&\graycell $t$ & \graycell$\tilde{O}(d)\cdot e^{O(\frac{d}{t^{p}})}\cdot n\log n$ &\graycell $\tilde{O}(d)\cdot e^{O(\frac{d}{t^{p}})}$&\graycell\cororef{cor:LowHopSpannerFromLSO}\tabularnewline
		\hline 
		
		$\ell_p^d$, $p\in[2,\infty]$&\graycell $2\cdot d^{1-\frac1p}$ & \graycell$\tilde{O}(d)\cdot n\log n$ &\graycell $\tilde{O}(d)$&\graycell\cororef{cor:LowHopSpannerFromLSO}\tabularnewline
		\hline 
		
		Tree & $1$ & $O\left(n\cdot\log n\right)$ & $O(1)$&\cite{KLMS22}\tabularnewline
		\hline 
		
		\multirow{-1}{*}{\begin{tabular}[c]{@{}c@{}}Fixed\\ Minor Free\end{tabular}}
		& $1+\epsilon$ & $O\left(n\cdot\epsilon^{-2}\cdot\log^{3}n\right)$ & $O(\eps^{-2}\cdot\log^{2}n)$&\cite{KLMS22}\\\hhline{~|-|-|-|-|}
		&\graycell $1+\epsilon$ &\graycell $O(n\cdot\eps^{-1}\cdot\log^2 n)$ &\graycell $O(\eps^{-1}\cdot\log n)$&\graycell\cororef{cor:LowHopMinor}\\\hline
		Planar &\graycell $1+\epsilon$ & \graycell$O(n\cdot\eps^{-1}\cdot\log^2 n)$ &\graycell $O(\eps^{-1})$&\graycell\cororef{cor:LowHopPlanar}\tabularnewline
		\hline 
		Treewidth $k$&\graycell $1$ & \graycell$O(n\cdot k\cdot\log n)$ &\graycell $O(k)$&\graycell\cororef{cor:LowHopTreewidth}\tabularnewline
		\hline 		
		
	\end{tabular}
	\caption{Summary of old and new results on path reporting low hop spanners. The spanners are for $n$ point metrics, and all report paths with hop bound $2$. Here $\epsilon\in(0,1)$, $k,d\ge1$ are integers. The space required for the path reporting data structure is asymptotically equal to the sparsity of the spanner in all the cases other than Euclidean space \Cref{cor:LowHopSpannerEuclidean}, where there is an additional additive factor of $O_d(\eps^{-2d})\log\frac1\eps$.  \label{tab:LowHopSpanners}}
\end{table}

\ourresults Kahalon \etal \cite{KLMS22} first constructed path reporting low hop spanners for trees, and then reduced each type of metric to the case of trees.
We observe that it is actually enough to reduce to the even simpler case of paths, and obtain a host of such spanners using LSO's. 
We then manually improve some of the resulting spanners, most notably we create path reporting $2$-hop $(1+\eps)$-spanner for planar graph with $O(\frac n\eps\log^2n)$ edges and $O(\frac1\eps)$-query time, and a path reporting $2$-hop $(1+\eps)$-spanner for points in $d$-dimensional Euclidean space with $O_d(\eps^{-d})\cdot \log\frac1\eps\cdot n\log n$ edges and $O_d(1)$-query time.
See \Cref{tab:LowHopSpanners} for a summary of old and new results.

\paragraph*{Fault tolerant spanners}
Levcopoulos, Narasimhan, and Smid \cite{LNS02} introduced the notion of a fault-tolerant spanner. A graph $H=(X,E_H,w)$ is an $f$-\emph{vertex-fault-tolerant} $t$-spanner of a metric space $(X,d_X)$, if for every set $F\subset X$ of at most $f$ vertices, it holds that $\forall u,v\notin F$, $d_{H\setminus F}(u,v)\le t\cdot d_{X}(u,v)$.
For general metrics, after a long line of work \cite{CLPR10,DK11,BDPW18,BP19Spanners,DR20,BDR21,Parter22}, it was shown that every $n$-vertex graph admits an efficiently constructible $f$-vertex-fault-tolerant $(2k-1)$-spanner with $O(f^{1-1/k}\cdot n^{1+1/k})$ edges, which is optimal assuming the Erd\"{o}s' Girth Conjecture~\cite{Erdos64}.
For $n$-points in $d$ dimensional Euclidean space, or more generally in a space of doubling dimension $d$, $f$-\emph{vertex fault tolerant} $(1+\eps)$-spanner were constructed with $\eps^{-O(d)}\cdot f\cdot n$ edges \cite{LNS02,Lukovszki99,Solomon14}.

Kahalon \etal \cite{KLMS22} initiated the study of low-hop fault tolerant spanners (previous constructions had $\Omega(\log n)$ hops). 
An $h$-hop $f$-fault tolerant $t$-spanner $H$ of a metric $(X,d_x)$ is a graph over $X$  such that for every set $F\subseteq X$ of at most $f$ vertices, for every $x,y\notin F$, the spanner without $F$: $H[X\setminus F]$ contains an $h$-hop path between $x$ to $y$ of weight at most $t\cdot d_X(x,y)$.
The advantages of such a spanner are straightforward, we refer to \cite{KLMS22} for a discussion.
Kahalon \etal constructed a $2$-hop $f$-fault tolerant spanner for doubling spaces with $n\cdot f^2\cdot\eps^{-O(d)}\cdot\log n$ edges. Note that a linear dependence on $f$ is necessary (as if a point has degree $\le f$ in $H$, we can delete all it's neighbors and get distortion $\infty$). It is natural to ask whether it is possible to construct such a spanner with only a linear dependence, and not quadratic as in   \cite{KLMS22}.

\ourresults One can easily construct $f$-fault tolerant $1$-spanner for the path graph with  $O(nf)$ edges. We observe that using $O(nf\log n)$ edges, it is possible to obtain  $f$-fault tolerant $2$-hop $1$-spanner for the path graph (note that $O(n\log n)$ edges are necessary for every $2$-hop spanner \cite{AS87,LMS22}).
Using the various old and new LSO's, we obtain a host of $f$-fault tolerant $2$-hop spanners for various metric spaces. Most notably, for metrics with doubling dimension $d$, we obtain an $f$-fault tolerant $2$-hop $(1+\eps)$-spanner with $\eps^{-O(d)}\cdot f\cdot n\cdot \log n$ edges, getting the desired linear dependence on $f$.
See \Cref{tab:FaultTolerantSpanners} for a summary of results.

\begin{restatable}[]{table}{TabFaultTolerant}
	\begin{tabular}{|l|l|l|l|}
		\hline
		Family                                                                    & Stretch & Edges & Ref \\ \hline
		\multirow{3}{*}{\begin{tabular}[c]{@{}l@{}}Doubling dimension $d$\end{tabular}} & $1+\eps$        & $\eps^{-O(d)}\cdot f^2\cdot n\cdot\log n$       &  \cite{KLMS22}   \\ \cline{2-4} 
		&  $1+\eps$       &  $\eps^{-O(d)}\cdot f\cdot n\cdot\log n$     & \cororef{cor:FaultTolerantSpannerFromLSO} \\ \cline{2-4} 
		&   $t$      &  $2^{-O(\nicefrac dt)}\cdot f\cdot \tilde{O}(n)$     &    \cororef{cor:FaultTolerantSpannerFromLSO} \\ \hline
		General Metric                                                            & $4k+\eps$        & $\tilde{O}(n^{1+\frac1k}\cdot f\cdot\eps^{-1})$      &    \cororef{cor:FaultTolerantSpannerFromLSO} \\ \hline
		\multirow{2}{*}{Euclidean $\R^d$}           & $1+\eps$        &   $O_d(\eps^{-d})\log\frac1\eps\cdot f\cdot n\cdot\log n$    &    \cororef{cor:FaultTolerantSpannerFromLSO} \\ \cline{2-4} 
		&  $(1+\eps)k$       &  $e^{\frac{2d}{k^{2}}\cdot(1+\frac{8}{k^{2}})}\cdot\tilde{O}(\frac{d^{1.5}}{\eps\cdot k})\cdot f \cdot n\cdot \log n$     &    \cororef{cor:FaultTolerantSpannerFromLSO} \\ \hline
		$\ell_p^d$, $p\in[1,2]$                                                                &   $k$      &  $e^{O(\frac{d}{k^{p}})}\cdot\tilde{O}(d)\cdot f \cdot n\cdot \log n$     &    \cororef{cor:FaultTolerantSpannerFromLSO} \\ \hline
		
		$\ell_p^d$, $p\in[2,\infty]$                                                                &   $2\cdot d^{1-\frac1p}$      &  $\tilde{O}(d)\cdot f \cdot n\cdot \log n$     &    \cororef{cor:FaultTolerantSpannerFromLSO} \\ \hline

		Treewidth k                                                              &   $2$      &  $O(n\cdot k\cdot f\cdot \log n)$     &    \cororef{cor:FaultTolerantSpannerFromLSO} \\ \hline
		Fixed Minor Free                                                          &  $2+\eps$       &  $O(\frac{n}{\eps}\cdot f\cdot \log^2 n)$     &    \cororef{cor:FaultTolerantSpannerFromLSO} \\ \hline
	\end{tabular}
	\caption{Summary of old and new results on $2$-hop $f$-fault tolerant spanners. The spanners are for $n$ point metrics, and all report paths with hop bound $2$. Here $\epsilon\in(0,1)$, $k,d\ge1$ are integers.  \label{tab:FaultTolerantSpanners}}
\end{restatable}

\paragraph*{Reliable spanners}
A major limitation of fault tolerant spanners is that the number of failures must be determined in advance. In particular, such spanners cannot withstand a massive failure.
One can imagine a scenario where a significant portion (even 90\%) of a network fails and ceases to function (due to, e.g., close-down during a pandemic), it is important that the remaining parts of the network (or at least most of it) will remain highly connected and functioning.  To this end, Bose \etal \cite{BDMS13} introduced the notion of a \emph{reliable spanner}. 
A $\nu$-reliable spanner is a graph such that for every failure set $B\subseteq X$, the residual spanner $H\setminus B$ is a $t$-spanner for $X\setminus B^+$, where $B^+\supseteq B$ is a superset of cardinality at most $(1+\nu)\cdot|B|$.
An oblivious $\nu$-reliable $t$-spanner is a distribution $\mathcal{D}$ over spanners, such  that for every failure set $B$, $H\setminus B$ is a $t$-spanner for $X\setminus B_H^+$, where the superset $B_H^+$ depends on both $B$ and the sampled spanner $H$. The guarantee is that the cardinality of $B^+_H$ is bounded by $(1+\nu)\cdot|B|$ in expectation, see \Cref{def:reliableSpanner}.

$\nu$-Reliable spanners were constructed for $d$ dimensional Euclidean and doubling spaces with $n\cdot\eps^{-O(d)}\cdot\tilde{O}(\log n)$ edges \cite{BHO19,BHO20,FL22} by a reduction from (classic) LSO's.
Oblivious reliable spanners were constructed also for planar, minor free, treewidth graphs, and general metrics \cite{FL22} by reductions from triangle, and rooted LSO's (as well as from sparse covers \cite{HMO21}).

\ourresults Our newly constructed triangle LSO's for high dimensional Euclidean, $\ell_p$ spaces, and doubling spaces, directly imply reliable spanners for these spaces, obtaining the first results without exponential dependence on the dimension. See \Cref{tab:reliableSpanners} for a summary.

\begin{restatable}[t]{table}{TabReliable}
	\begin{tabular}{|c|l|l|l|l|}
		\hline
		Family                                                                                    & stretch               & guarantee     & size                                                                                                                                 & ref                        \\ \hline
		\multirow{3}{*}{\begin{tabular}[c]{@{}l@{}}Euclidean\\ $(\R^d,\|\cdot\|_2)$\end{tabular}}& $1+\eps$              & Deterministic & $n\cdot\tilde{O}_d(\epsilon^{-7d})\nu^{-6}\cdot\tilde{O}(\log n)$
		& \cite{BHO19}                       \\ \cline{2-5} 
		& $1+\eps$              & Oblivious     & $n\cdot\tilde{O}_d(\epsilon^{-2d})\cdot\tilde{O}(\nu^{-1}(\log\log n)^{2})$ & \cite{BHO20}          \\ \hhline{~|-|-|-|-|}
		
		&\graycell $(1+\eps)t$              & \graycell Oblivious     &\graycell $\nu^{-1}\cdot e^{\frac{4d}{t^{2}}\cdot(1+\frac{8}{t^{2}})}\cdot\tilde{O}(n\cdot \frac{d^{3}}{\eps^{2}\cdot t^{2}})$ &\graycell \cororef{cor:RelaibleSpanner}        \\ \hline
		
		$\ell_p^d$ for $p\in[1,2]$ &\graycell  $t$              & \graycell Oblivious     &\graycell $\nu^{-1}\cdot e^{O(\frac{d}{t^{p}})}\cdot\tilde{O}(n\cdot d^{2})$ &\graycell \cororef{cor:RelaibleSpanner}        \\ \hline
		
		$\ell_p^d$ for $p\in[2,\infty]$ &\graycell  $2\cdot d^{1-\frac1p}$              & \graycell Oblivious     &\graycell $\nu^{-1}\cdot\tilde{O}\left(n\cdot d^{2}\right)$ &\graycell \cororef{cor:RelaibleSpanner}        \\ \hline

		\multirow{3}{*}{\begin{tabular}[c]{@{}l@{}}Doubling\\ dimension $d$\end{tabular}}         & $1+\eps$              & Deterministic & $n\cdot\epsilon^{-O(d)}\nu^{-6}\cdot\tilde{O}(\log n)$
		&    \cite{FL22}                        \\ \cline{2-5} 
		& $1+\eps$              & Oblivious     & $n\cdot\epsilon^{-O(d)}\nu^{-1}\log\nu^{-1}\cdot\tilde{O}(\log\log n)^{2}$
		&    \cite{FL22}                       \\ \hhline{~|-|-|-|-|}
		
		&\graycell $t$              &\graycell Oblivious     &\graycell $\tilde{O}(n)\cdot\nu^{-1}\cdot2^{O(\nicefrac{d}{t})}$                                               
		&\graycell    \cororef{cor:RelaibleSpanner}                   \\ \hline
		
		General metric& $8t+\eps$ & Oblivious     & $\tilde{O}(n^{1+\nicefrac{1}{t}}\cdot\eps^{-2})\cdot\nu^{-1}$ & \cite{FL22} \\ \hline
		
		Tree& $2$                   & Oblivious     & $n\cdot O(\nu^{-1}\log^{3}n)$                                                                                                  & \cite{FL22}       \\ \hline
		Treewidth $k$                                                                            & $2$                   & Oblivious     & $n\cdot O(\nu^{-1}k^{2}\log^{3}n)$                                                                                             & \cite{FL22}      \\ \hline
		Planar/Minor-free    & $2+\eps$                                                                                          & Oblivious     & $n\cdot O(\nu^{-1}\eps^{-2}\log^{5}n)$                                                                          & \cite{FL22}        \\ \hline
	\end{tabular}
	\caption{\small{Summary of previous and new constructions of $\nu$-reliable spanners.
	}}
	\label{tab:reliableSpanners}
\end{restatable}

\paragraph*{Light spanners}
An extensively studied parameter is the \emph{lightness} of a spanner, defined as the ratio $w(H)/w(MST(X))$, where $w(H)$ resp.~$w(MST(X))$ is the total weight of edges
in $H$ resp.~a minimum spanning tree (MST) of $X$. 
Obtaining spanners with small lightness (and thus total weight) is
motivated by applications where edge weights denote e.g.~establishing cost.
The best possible total weight that can be achieved in order to ensure finite stretch is the weight of an MST, thus making the definition of lightness very natural.

Obtaining light spanners for general graphs has been the subject of an active line of
work \cite{ADDJS93,CDNS95,ENS15,BFN19,CW18,FS20}, where the state of the are by Le and Solomon \cite{LS21Unified2} who obtained $(1+\eps)(2k-1)$ spanner with lightness $O(\eps^{-1}\cdot n^{\frac1k})$.
Light spanners were also studied extensively in Euclidean spaces (see the book \cite{NS07}), doubling spaces \cite{Gottlieb15,FS20,BLW19}, planar and minor free graphs \cite{Klein06,Klein08,BLW17,Le20,CFKL20}, and high dimensional Euclidean and doubling spaces \cite{HIS13,FN22,LS21Unified2}.

\ourresults Recently Le and Solomon \cite{LS21Unified2} obtain a general framework for constructing light spanners from spanner oracles. 
We construct new spanner oracles using LSO's. As a result we derive new light spanners, that improve the state of the art for high dimensional spaces (and match the state of the art for low dimensional doubling spaces). See \Cref{tab:resultsLightSpanners} for a summary of results.

\begin{restatable}[t]{table}{TabLight}
	\def\arraystretch{1.15}
	\begin{tabular}{|c|l|l|l|}
		\hline
		\multicolumn{1}{|l|}{Metric space}                      & \multicolumn{1}{c|}{\textbf{Stretch}}                 & \multicolumn{1}{c|}{\textbf{Lightness}} & \multicolumn{1}{c|}{\textbf{ref}}   \\ \hline
		\multirow{4}{*}{Euclidean space}      	& $O(t)$                                                & $O(n^{\frac{1}{t^2}}\cdot\log n\cdot t)$               & \cite{LS21Unified2} \\ \cline{2-4}
		
		& $O(t)$                                                & $ O(e^{\frac{d}{t^2}}\cdot\log^2 n\cdot t)$               & \cite{FN22} \\ \hhline{~|-|-|-|}
		
		& \graycell$(1+\eps)2t$                                                & \graycell$e^{\frac{d}{2t^{2}}\cdot(1+\frac{2}{t^{2}})}\cdot\tilde{O}(\frac{d^{1.5}}{\eps^{2}})\cdot\log n$               & \graycell\cororef{cor:highDimEuclidean} \\ \hhline{~|-|-|-|}
		& \graycell$(1+\eps)4t$                                                &\graycell $e^{\frac{d}{2t^{2}}\cdot(1+\frac{2}{t^{2}})}\cdot\tilde{O}(\frac{d^{1.5}}{\eps^{2}})\cdot\log^* n $               &\graycell \cororef{cor:highDimEuclidean} \\ \hline	
		
		\multirow{4}{*}{Doubling dimension $d$} & $O(t)$                                                & $O(2^{\frac{d}{t}}\cdot t\cdot\log^2n)$    &\cite{FN22} \\ \hhline{~|-|-|-|}
		
		&\graycell$O(t)$                                                &\graycell $2^{O(\nicefrac{d}{t})}\cdot d\cdot\log^{2}t\cdot\log^*n$    &\graycell\cororef{cor:LightSpannerDoubling}\\ \cline{2-4}
		
		& $d$                       & $O(d\cdot\log^2n)$         & \cite{FN22}\\ \hhline{~|-|-|-|}

		&\graycell $d$                       & \graycell$O(d\cdot\log^2d\cdot\log^*n)$         &\graycell\cororef{cor:LightSpannerDoubling} \\ \hline

		\multirow{2}{*}{$\ell^d_p$ for $p\in[1,2]$} & $t$                                                & $O(\frac{t^{1+p}}{\log^{2}t}\cdot n^{O(\frac{\log^{2}t}{t^{p}})}\cdot\log n)$    &\cite{FN22} \\ \hhline{~|-|-|-|}
		
		&\graycell$t$                                                &\graycell $e^{O(\frac{d}{t^{p}})}\cdot\tilde{O}(d\cdot t)\cdot\log^{*}n$    &\graycell\cororef{cor:LightSpannerLp12}\\ \hline
		
		$\ell^d_p$ for $p\in[2,\infty]$&\graycell$4\cdot d^{1-\frac1p}$                                                &\graycell $\tilde{O}(d^{2-\frac{1}{p}})\cdot\log^{*}n$    &\graycell\cororef{cor:LightSpannerLp2infty}\\ \hline
	\end{tabular}
	\caption{Summary of previews and new results of light spanners for high dimensional metric spaces. Interestingly, for $p\in[1,2]$ \cite{FL22} obtain lightness $O(\frac{t^{1+p}}{\log^{2}t}\cdot n^{O(\frac{\log^{2}t}{t^{p}})}\cdot\log n)$ regardless of dimension, which is superior to our \Cref{cor:LightSpannerLp12} for $d\gg\log n$.
	}\label{tab:resultsLightSpanners}
\end{restatable}

\subsection{Technical ideas}\label{subsec:techIdeas}

\paragraph*{Triangle LSO for high dimensional Euclidean space}
Our construction is very natural: partition the space randomly in every distance scale $\xi_i$ (for some large $\xi$) into clusters of diameter $\xi_i$, such that close-by points are likely to be clustered together. In the created ordering $\sigma$, points in each cluster will be ordered consecutively and recursively.
In particular, the ordering $\sigma$ will correspond to a laminar partition obtained by the clustering in all possible scales.
For a pair of points $x,y\in\R^d$ to be satisfied in the resulting ordering $\sigma$, they have to be clustered together in all the distance scales $\xi^i\ge t\cdot\|x-y\|_2$.

Our space partition in each scale is done using ball carving (ala \cite{AI08}): pick a uniformly random series of centers $z_1,z_2,\dots,$. Each points is assigned to the cluster of the first center at distance at most $R=\frac12\cdot\xi^i$. 
We show that a finite random seed of size $d^{O(d)}$ is enough to sample such a clustering (in all possible distance scales, simultaneously).
The probability that two points $x,y$ are clustered together is then equal to the ratio between the volumes of intersection and union of balls:
$\Pr[x,y\mbox{ clustered together}]=\frac{\text{Vol}_{d}\left(B(x,R)\cap B(y,R)\right)}{\text{Vol}_{d}\left(B(x,R)\cup B(y,R)\right)}\ge\Omega(\frac{1}{\sqrt{d}})\cdot\left(1-(\frac{\|x-y\|_{2}}{R})^{2}\right)^{\nicefrac{d}{2}}$.
We bound this ratio for the case $\|x-y\|_2\le \frac{R}{\sqrt{d}}$ using a lemma from \cite{CCGGP98} (see \Cref{fact:diffClustersSmall}).
For the general case, we prove that the ratio between these volumes is at least $\Omega(\frac{R}{\sqrt{d}\cdot\|p-q\|_{2}})\cdot(1-(\frac{\|p-q\|_{2}}{R})^{2})^{\frac{d}{2}}$
 (\Cref{fact:diffClustersLarge}), slightly improving a similar fact from \cite{AndoniThesis}, by a $\frac{R}{\|p-q\|_{2}}$ factor. This ratio eventually governs our success probability (when replacing $\nicefrac{R}{\|p-q\|_2}$ by twice the stretch $2t$). 
 The improved analysis of the volumes ratio is significant for the $O(\sqrt{d})$-stretch regime, improving the number of orderings to $\tilde{O}(d)$ (compered with $\tilde{O}(d^{1.5})$ orderings if we were using \cite{AndoniThesis}).
 
To generalize this construction to $\ell_p$ spaces, we use the exact same construction, replacing $\ell_2$ balls with $\ell_p$ balls. The volume ratio lemma from \cite{CCGG98} for close-by points is replaced by a crude observation (\Cref{obs:diffClustersSmallLp}) without any significant consequences to the resulting number of orderings. For the general case, we directly analyze the ratio of volumes for $\ell_p$-balls (our computation is similar to \cite{Nguyen13}). The rest of the analysis is the same. 

\paragraph*{Triangle LSO for doubling spaces}
Ultrametrics (\Cref{def:HST}) are trees with additional structure, where each ultrametric admits a $(1,1)$-triangle LSO.
$(\tau,\rho)$-ultrametric cover (\Cref{def:UltrametricCover}) of a metric space $(X,d_X)$ is a collection $\mathcal{U}$ of $\tau$ ultrametrics such that every pair $x,y\in X$ is well approximated by the ultrametrics: $d_X(x,y)\le\min_{U\in\mathcal{U}}d_{U}(x,y)\le \rho\cdot d_X(x,y)$. 
Filtser and Le \cite{FL22} showed that $(\tau,\rho)$-ultrametric cover implies $(\tau,\rho)$-triangle LSO. 
We construct $\left(2^{O(\nicefrac dt)}\cdot d\cdot\log^2 t,t\right)$-ultrametric cover for spaces with doubling dimension $d$, which implies \Cref{thm:LSOdoublingLargeStretch}.

Our starting point for constructing the ultrametric cover is Filtser's \cite{Fil19padded} padded partition cover (\Cref{def:PaddedPartitionCover}), which is a collection of $\approx 2^{O(\nicefrac dt)}$ space partitions where all clusters are of diameter at most $\Delta$, and every ball of radius $\frac{\Delta}{t}$ is fully contained in a single cluster in one of the partitions. 
We take a single partition from each distance scales, where the gap between the distance scales is somewhat large: $O(\frac{t}{\eps})$. Initially these partitions are unrelated, and we ``force'' them to be laminar (see \Cref{fig:Laminar}), while keeping the padding property. Each such laminar partition induces an ultrametric, and their union is the desired ultrametric cover.

\paragraph*{Labeled NNS}
Morally, given a $(\tau,\rho)$ LSO (or triangle LSO), the NNS label of every point is simply its position in each ordering. Given a query $q$, we simply find its successor and predecessor in each one of the orderings, one of them is guaranteed to be an approximate nearest neighbor (abbreviated ANN).
We can find the successor and predecessor in each ordering in $O(\log\log N)$ time using Y-fast trie \cite{Wil83}, it only remains to choose one of the $2\tau$ candidates to be the ANN.
To solve this problem we again deploy the LSO structure, and construct a $2$-hop $1$-spanner for the implicit path graph induced by each ordering. Specifically, each point will be associated with $O(\log N)$ edges (the name and weight of which will be added to the NNS label), where given two points $x\prec_\sigma y$, in $O(1)$ time we will be able to find a point $z$ such that  $x\preceq_\sigma z\preceq_\sigma y$ and $x$ and $y$ stored $\{x,z\}, \{y,z\}$ respectively (see \Cref{thm:2hopPath}).
Then $d_X(x,z)+d_X(z,y)$ will provide us the desired estimate of $d_X(x,y)$, which will be used to choose the ANN.

The case of rooted LSO is simpler- the label of each point $z$ will consist of its position in all the orderings $\sigma$ it belongs to, and the distance to the first point $x_\sigma$ (w.r.t. $d_X$).
Given a query $q$, for each ordering $\sigma$ containing $q$, 
the leftmost point  $y_\sigma\in P$ in the ordering will be a candidate ANN. We will estimate the distance from $q$ to $y_\sigma$ by $d_X(q,x_\sigma)+d_X(x_\sigma,y_\sigma)$, and return the point with minimum estimation. 

For general metrics, the number of orderings is polynomial, $N^{\frac1k}$ which results in similar NNS label size, and query time (following the approach above). While the NNS label essentially cannot be improved (\Cref{thm:NNS-LB}), we can significantly reduce the query time. 
Our solution is to use Ramsey trees \cite{BLMN05,MN07,Bar21,ACEFN20}, which are a collection of embeddings into ultrametrics $\mathcal{U}$ such that each point $x$ has a single home ultrametric $U_{x}\in\mathcal{U}$ which well approximate all the distances to $x$. We thus reduce the labeled NNS problem to ultrametrics, where it can be efficiently solved.
For the case of approximation factor $O(\log N)$ the required number of ultrametrics is $O(\log N)$, which leads us to label size $O(\log^2 N)$.
To reduce it even farther, we use the novel clan embedding \cite{FL21}, where instead of embedding the space $X$ into a collection of ultrametrics, we embed it into a single ultrametric (but where each point might have several copies). 
This allows us to reduce the label size to $O(\log N)$ (in expectation), and with one additional easing assumption (either polynomial aspect ratio or small failure probability) to even $O(1)$ label size, see \Cref{tab:GeneralNNS}.

\paragraph*{Path reporting low hop spanners}
A $(\tau,\rho)$-tree cover is similar to ultrametric cover discussed above, where the ultrametrics are replaced by trees.
Kahalon \etal \cite{KLMS22} first constructed path reporting low hop spanner for a tree metric,
and then for each metric space of interest, they considered it's tree cover, and constructed a path reporting low hop spanner for each tree in the cover. The spanner for the global metric is obtained by taking the union of all these spanners constructed for the trees in the cover. 
To report a queried distance, they simply computed the paths in all the trees, and returned the shortest observed path.

Thus Kahalon \etal idea is to reduce the problem to the fairly simple case of tree metrics. 
We reduce each metric space into the even simpler case of paths using LSO.
Given an LSO (or triangle LSO) we simply construct a path reporting $2$-hop path for each path associated with an ordering, and similarly to \cite{KLMS22}, check all the path spanners and return the shortest observed path. 
The resulting query time has linear dependence on the number of orderings.
The case of rooted LSO is simpler, where it is enough to add a single edge per ordering, to the leftmost point in the ordering.

Next we present some improvements to the query. First, for the case of Euclidean space (low dimensional), we observe that given two points $x,y$, the ordering satisfying them could be computed in $O_d(1)$ time, implying that we don't need to check all the orderings, and return a $2$ hop path in $O_d(1)$ time.
Next, for the case of planar graphs, using the structure of cycle separators (which are used to construct the rooted LSO), in $O(1)$ time one can narrow the number of potential orderings to $O(\eps^{-1})$, implying $O(\eps^{-1})$ query time.
For general graphs we observe that the celebrated Thorup Zwick distance oracle \cite{TZ05} can be used to produce a path reporting $2$-hop $(2k-1)$-spanner with $O(n^{1+\frac1k}\cdot k)$ edges and $O(k)$ query time.
Finally, we use sparse covers \cite{AP90} to obtain an exponential improvement in the query time, while incurring a factor $2$ increase in the stretch.

\paragraph*{Fault tolerant spanners}
The $2$-hop $f$-fault tolerant spanner for doubling metrics by Kahalon \etal \cite{KLMS22} is based on a quite sophisticated tool of robust tree cover.
We have a superior, and an extremely simple construction.
First we observe that the path graph has a $2$-hop $f$-fault tolerant $1$-spanner with $O(nf\log n)$ edges. Indeed, add edges from all the vertices to the middle $f+1$ vertices, delete the middle vertices and recurse on each side.
We then apply this construction on each of the path graphs induced by the LSO (or triangle LSO) to obtain our results.
The case of  rooted LSO is even simpler: for every path it is enough to add all the edges to the first $f+1$ points.

\subsection{Related Work}\label{subsec:Related}
Most of the literature on NNS studies the $(r,c)$-Near Neighbor problem, where given a query $q$, the goal is to find a neighbor $p$ at distance at most $t\cdot r$ from $q$, or declare that there are no points at distance at most $r$ from $q$.
Har-Peled, Indyk, and Motwani \cite{HIM12} gave a general reduction from  the $(r,c)$-Near Neighbor problem to $(1+\eps)\cdot c$-NNS with a blow up of $O(\log n)$ in the query time, $O(\eps^{-1}\cdot\log n)$ in space, and $O(\log n)$ in the failure probability. 
In this paper we do not go though this reduction, and solve the actual problem directly. In particular, for general metric spaces we can get query time much better that the additional $O(\log n)$ inflicted by the reduction.

Nearest neighbor search for doubling dimension $d$ spaces is well studied \cite{KL04,HM06,CG06NNS,HK13}. One can find a $(1+\eps)$-approximated nearest neighbor with space $2^{O(d)}\cdot n$, and query time $2^{O(d)}\log n+\eps^{-O(d)}$ both being exponentially dependent on the dimension. 
Note that one can retrieve similar performance in the labeled NNS model using the approach in this paper. No previous work attempt to solve the problem for  moderately large doubling dimension. 

For the high-dimensional Euclidean space, approximate nearest neighbor search was extensively studied (see the survey \cite{AIR18}), where the state of the art is by Andoni and Razenshteyn \cite{AR15} who solve the $r$-near neighbor problem with approximation $c$, space $n^{1+\rho}+O(nd)$ and query time $n^{\rho}+dn^{o(1)}$, where $\rho=\frac{1}{2c^2-1}$.
Using the triangle-LSO's presented in this paper, we can obtain a similar result (though with worse constants).
Specifically, first use the JL transform \cite{JL84} to reduce the dimension to $O(\log n)$, then use the (hard-coded) $(\tilde{O}(e^{\frac{d}{c^2}}),c)\approx (\tilde{O}(n^{\frac{1}{c^2}}),O(c))$-triangle LSO, and proceed as in \Cref{thm:triangleToNNS} (given a query, check it's neighbors in all orderings and take the closest one).  The result will be an $O(c)$ approximation with $\tilde{O}(n^{1+\frac{1}{c^2}})$ space and $\tilde{O}(n^{\frac{1}{c^2}})$ query time.

Sparse (and light) spanners were constructed efficiently in different computational models such as  standard RAM  \cite{ES16,EN19,ADFSW22,LS22}, LOCAL \cite{KPX08}, CONGEST \cite{EFN20}, streaming \cite{Elkin11,Baswana08,BS07,KW14,AGM12Spanners,FKN21}, massive parallel computation (MPC) \cite{BDGMN21} and dynamic graph  algorithms~\cite{BKS12, BFH19}, and online \cite{BFKT22}.
Path reporting distance oracles were studied by Elkin and Pettie \cite{EP15} (see also \cite{ENS15}), who constructed a distance oracle of space $O(n^{1+\frac1k})$, and that given a query pair, in $O(n^\eps)$ time returns a path with stretch $O_\eps(k)$.

An extensively studied related notion to low hop metric spanners is \emph{hop-set}. Given a graph $G=(V,E,w)$, a hop-set is a set of edges $H'$ that when added to $G$, every pair of vertices have a low-hop path.
Formally, an  $(t,h)$-hop-set is a set $H'$ of edges such that $\forall u,v\in V$, $d_G(u,v)\le d_{G\cup H'}(u,v)$, and $G\cup H'$ contains an $h$-hop path $P$ between $u,v$ of weight $w(P)\le  t\cdot d_G(u,v)$.
Elkin and Neiman \cite{EN19hop} constructed $\left(1+\eps,\red{\left(\frac{\log k}{\eps}\right)^{\log k-2}}\right)$ hop-sets with $O_{\eps,k}(n^{1+\frac1k})$ edges.
We refer to  the survey \cite{EN20} for further details.
Another related notion is  low-hop emulator \cite{ASZ20}, where the $h$-hop distance is required to respect the triangle inequality.
 
 \section{Preliminaries}
 All logarithms (unless explicitly stated otherwise) are in base $2$.
 We use $\tilde{O}$ notation to hide poly-logarithmic factors. That is $\tilde{O}(s)=O(s)\cdot\log^{O(1)}(s)$.
 Similarly, we use $O_d$ notation to hide arbitrary factors in $d$. That is $O_d(s)=O(s)\cdot f(d)$ for some arbitrary function $f:\N\rightarrow\N$.
 
 In the data structures constructed in this paper, the space is counted in machine words and not bits (unless explicitly stated otherwise). Machine word is generally thought of as being $O(\log N)$ bits, and can store number, identifier, pointer etc.
 In addition, we consider storing a given label as taking $O(1)$ time (as it is just storing a pointer to the place where it is actually stored).
 
 Let $(X,d_X)$ be a metric space.  The aspect ratio (sometimes called spread)  $\Phi = \frac{\max_{x,y\in X}d_X(x,y)}{\min_{x\not=y \in X}d_X(x,y)}$ is the ratio between maximum and minimum distances.
 This paper studies metric spaces, however we also study graphs. Note that a metric space can be viewed as a complete undirected graph with edge weights respecting the triangle inequality.
 Given a weighted graph $G=(V,E,w)$ where the weight function $w:E\rightarrow\R_{\ge0}$ gets positive values, the shortest path metric $d_G$ is a metric space where the distance between two vertices $u,v$ equals to the minimal weight of a path from $u$ to $v$.
 Frequently the literature is concerned with graph spanners, where given a graph $G=(V,E,w)$ the goal is to find a subgraph $H$ preserving distances. Here we study metric spanners, where there is no underlying graph. For example, when we discuss spanners for planar graphs (or other topologically restricted graph families), we assume that the underlying metric is the shortest path metric of a planar graph $G$, but we are allowed to take any edge, regardless of whether it is in $G$. 
 Interesting to note, for the case of the lightness parameter, if the metric comes from a shortest path metric of a graph $G$, given a metric spanner $H$ for $d_G$, we can always turn $H$ into a graph spanner $H'$ (a subgraph of $G$) of the same lightness as $H$ (this is by replacing each edge $e$ in $H$ by the shortest path between its endpoints in $G$).
 Thus our result in \Cref{subsec:light} hold for graph spanners as well.

 A metric $(X,d_X)$ has doubling dimension $d$ if every ball of radius $r$ can be covered by at most $2^{d}$ balls of radius $r/2$. 
 Doubling dimension is a generalization of the concept of Euclidean dimension to general metric spaces. In particular, the $d$-dimensional Euclidean space has doubling dimension $\Theta(d)$.
 A key property of doubling metrics (see, e.g., \cite{GKL03}) is the \emph{packing property} that grantees that for a subset of points  $S \subseteq X$ with minimum interpoint distance $r$ that is contained in a ball of radius $R$, it holds that  
 $|S| = \left(\frac{2R}{r}\right)^{O(d)}$.

The path graph $P_{n}$ contains $n$ vertices
$v_{1},v_{2},\dots,v_{n}$ and there is (unweighted) edge between
$v_{i}$ and $v_{j}$ iff $|i-j|=1$. A path $v_{i_{1}},v_{i_{2}},\dots v_{i_{s}}$
is monotone iff for every $j$, $i_{j}<i_{j+1}.$ Note that if a spanner
$H$ contains a monotone path between $v_{i},v_{j}$ then $d_{H}(v_{i},v_{j})=d_{P_{n}}(v_{i},v_{j})=|i-j|$.
The following result is usually attributed to Alon and Schieber \cite{AS87} (see also \cite{BTS94,Solomon13}), we provide a proof as we will use properties not always explicitly stated.
\begin{theorem}[\cite{AS87}]\label{thm:2hopPath}
	The path graph $P_n=(v_1,v_2,\dots,v_n)$ admits a $2$-hop $1$-spanner $H$ with $O(n\cdot\log n)$ edges. Moreover, there is a path reporting data structure with constant query time.
	Specifically, each vertex $v_i$ is responsible for a subset of edges $E_i$ incident on $v_i$ where $|E_i|=O(\log n)$. Given a pair $i<j$, in $O(1)$ time, we can find an $l\in[i,j]$ such that $\{v_i,v_l\}\in E_i$ and $\{v_j,v_l\}\in E_j$.
\end{theorem}
\begin{proof}
	For simplicity we will assume that $n=2^\delta$ for some integer $\delta$. Later we can simply construct the spanner for a power of $2$ in $[n,2n)$.
	Assume by induction that for every $\delta'<\delta$, a $2$-hop $1$-spanner with $n'\log n'+1=2^{\delta'}\cdot\delta'+1$ edges and the data structure above, exist for the path graph $P_{n'}=P_{2^{\delta'}}$.
	The base case with $2$ vertices and $\delta'=0$ clearly holds.
	Next consider $n=2^\delta$. Add edges from all the vertices to $v_{2^{\delta-1}}$ (and for every vertex $v_i$ add $\{v_i,v_{2^{\delta-1}}\}$ to $E_i$). Next use the induction hypothesis on the graph induced by $(v_1,v_2,\dots,v_{2^{\delta-1}})$, and the graph induced by $(v_{2^{\delta-1}+1},\dots,v_{2^{\delta}})$. By the induction hypothesis, the total number of edges is bounded by 
	\[
	2^{\delta}-1+2\cdot\left(2^{\delta-1}\cdot(\delta-1)+1\right)=2^{\delta}\cdot\delta+1=n\cdot\log n+1~.
	\]
	As the depth of the recursion is $\log n$, clearly for every $i$, $|E_i|\le\log n$.
	To report a $2$ hop path between $v_a$ and $v_b$, we simply need to find the maximal $i$ such that $a\le c\cdot2^i\le b-1$ for some integer $c$. The path will be $(v_a,v_{c\cdot2^i},v_b)$, where $\{v_a,v_{c\cdot2^i}\}\in E_a$ and $\{v_b,v_{c\cdot2^i}\}\in E_b$. To find the number $c\cdot2^i$, we can simply look on the binary representation of $a=\alpha_1\alpha_2\dots\alpha_s$ and $b-1=\beta_1\beta_2\dots\beta_s$, and let $j$ be the most significant bit where they differ, then $c\cdot2^i=\beta_1\dots\beta_j0\dots0$. This operation take $O(1)$ time in the standard RAM computational model.
\end{proof}

\section{LSO's for high dimensional Euclidean space}\label{sec:LSOForEuclidean}

This section is devoted to proving \Cref{thm:LSOEuclideanLargeStretch}, restated bellow for convenience:
\HighDimLSO*
In \Cref{sec:highDimLSOconstruction,subsec:BallIntersection,subsec:RdLSOproof} we describe the construction of our LSO, and prove it's properties. In \Cref{sec:ell-p-LSO} we generalize the construction for $\ell_p$ spaces for $p\in[1,2]$.

\subsection{The construction of the LSO}\label{sec:highDimLSOconstruction}

Fix a stretch parameter $\red{t\ge1}$.
Given a distance scale $w>0$ we define the
following random clustering procedure (introduced by Andoni and Indyk \cite{AI08}), where in addition there will be an order over the clusters. 
The point set $4w\cdot\mathbb{Z}^{d}=\left\{\left(4w\cdot \alpha_1,\dots,4w\cdot \alpha_d\right)\in\R^d~\mid~\alpha_1,\dots,\alpha_d\in\Z\right\}$
is the $d$-dimensional axis parallel lattice with side length $4w$. $G^{d,w}=\cup_{x\in4w\cdot\mathbb{Z}^{d}}B(x,w)$
is the union of all the balls with radius $w$ and center at a lattice point
$4w\cdot\mathbb{Z}^{d}$. 
Given a point $v\in\mathbb{R}^{d}$,
$v+G^{d,w}$ is simply $G^{d,w}$ shifted by $v$, we will denote
this set by $G_{v}^{d,w}$. 
We create clusters $\mathcal{C}$ of $\mathbb{R}^{d}$ as
follows: 
let $v_{1},v_{2},\dots,v_{l}\in[0,4\cdot w)^{d}$ be random points chosen independently (i.i.d) and uniformly, until $\bigcup_{i=1}^l G_{v_{i}}^{d,w}=\mathbb{R}^{d}$.
By a standard volume argument (see also Lemma 3.2.2. in \cite{AndoniThesis}),
$l=2^{O(d\log d)}$ with high probability. For every $v_{i}$, and
$u\in\mathbb{Z}^{d}$ there will be a cluster $C_{v_{i},u}$ containing
all the points $p$ in the ball $B\left(v_{i}+4w\cdot u,w\right)$
with center $v_{i}+4w\cdot u$ and radius $w$, such that there is
no $v_{i'}$ with $i'<i$ and $u'\in\mathbb{Z}^{d}$ such that $p\in B\left(v_{i'}+4w\cdot u',w\right)$.
The set of created clusters is denoted $\mathcal{C}_{w}$. 
See \Cref{fig:AI08clustering} for illustration.
We define
an order $\preceq_{\mathcal{C}^{w}}$ over the clusters where $C_{v_{i},u}\preceq_{\mathcal{C}^{w}}C_{v_{i'},u'}$
if either $i<i'$, or $i=i'$ and $u\preceq_{\text{lex}}u'$ in the
lexicographical order (defined on vectors in $\mathbb{Z}^d$).

\begin{SCfigure}[][t]\caption{\it 
		Illustration of the clustering process for $w=\frac14$ and $d=2$.
		First the point $v_1\in[0,4\cdot w)^d=[0,1)^2$ is chosen, and we add the clusters $G^{2,\frac14}_{v_1}$ (colored in green).
		The algorithm continues to create clusters in these manner until the entire space is clustered.\\~\\~\\~\\
		\label{fig:AI08clustering}}
	\includegraphics[width=.75\textwidth]{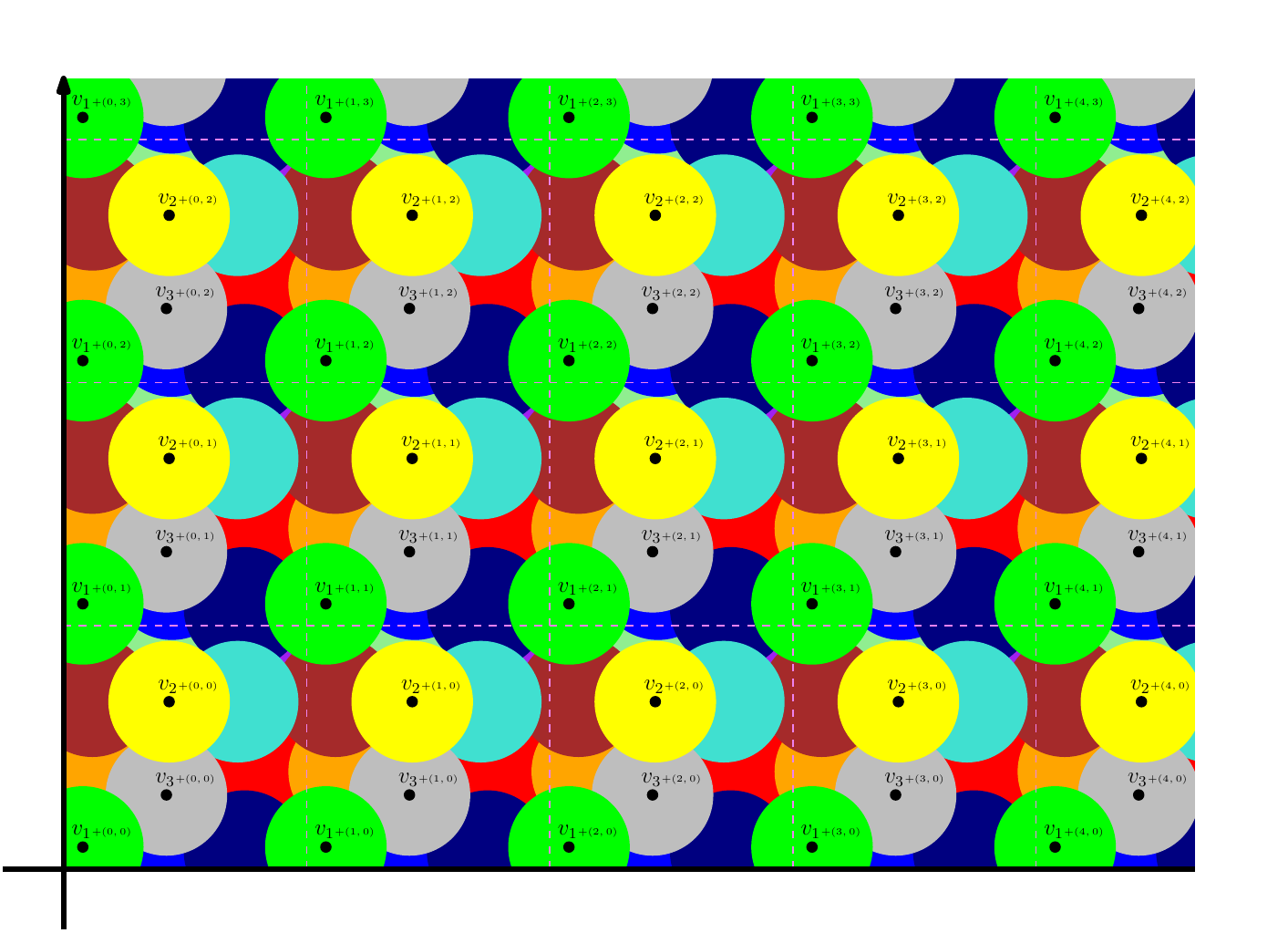}	
\end{SCfigure}


Next, we define a random order $\preceq$ over the entire space $\mathbb{R}^{d}$.
For every $i\in\mathbb{Z}$, let $w_{i}=\xi^{i}$ for $\red{\xi=\frac{12\sqrt{d}}{t}}$ (which will denote the gap between two consecutive scales), and let $\mathcal{C}^{w_{i}}$
be a clustering drawn randomly as above. Given two different points $p,q$,
note that they must belong to different clusters for small enough $i$.
Let $i_{p,q}$ be the maximal 
index such that $p$ and $q$ belong
to different clusters $C_{p},C_{q}\in\mathcal{C}^{w_{i_{p,q}}}$.
Then $p\preceq q$ iff $C_{p}\preceq_{\mathcal{C}^{w_{i_{p,q}}}}C_{q}$,
and $q\preceq p$ otherwise. 
\footnote{If the clusterings $\mathcal{C}^{w_{i}}$ are drawn independently for each scale $i$, such a maximal index $i_{p,q}$ exist with probability $1$. However, we will construct clusters for different scales with dependencies, and ensure that every pair will have such an index.}
The random order $\preceq$ is now defined. 

\paragraph*{Computational aspects}
Our LSO is defined w.r.t the entire space $\R^d$. Formally, given a pair of points, one is required to check infinite number of scales to find the first place they are separated. In all the applications in this paper we will be dealing with finite $n$-point subsets. Here it will be enough to create $\log \Phi$ scales of clusters ($\Phi$ is the aspect ratio), and there is no need to define the LSO w.r.t. the entire $\R^d$. Thus computing the order between a pair of points could be done efficiently (though exponential in $d$ time).

That being said, there is a small tweak to our construction that enable us to define an ordering on the entire space without the need to check infinite number of scales.
Sample a radius $r\in[\frac12,1]$ uniformly at random. For scale $w_i$, add all the points in the ball $B(\vec{0},r\cdot w_i)$ into a special first cluster. The rest of the points will be clustered as described above, and the random ordering defined w.r.t. clustering as above. One can show that the probability that a pair of points $p,q$ clustered together is only slightly changed. On the other hand, for every scale $i$ such that $\xi^i\ge2\cdot\max\{\|p\|_2,\|q\|_2\}$, $p$ and $q$ will necessarily be clustered together, and thus it will be enough to check only finite number of scales (which can be only smaller than the representation length of $p,q$ as binary strings).

\subsection{The volume of ball intersection}\label{subsec:BallIntersection}
Consider a pair of points $p,q\in \R^d$ and the random partition of the space $\mathcal{C}^{w_i}$. 
The points $p$ and $q$ will belong to the same cluster if and only if the first center chosen in $B(p,w_i)\cup B(q,w_i)$ (that is the first ball clustering either $p$ or $q$) will be in the intersection $B(p,w_i)\cap B(q,w_i)$ (that is clustering them both). As the distribution of centers in $B(p,w_i)\cup B(q,w_i)$ is uniform, the probability of $p$ and $q$ to be clustered together is equal to the ratio $\frac{\text{Vol}_{d}\left(B(p,w_{i})\cap B(q,w_{i})\right)}{\text{Vol}_{d}\left(B(p,w_{i})\cup B(q,w_{i})\right)}$.

Charikar \etal \cite{CCGGP98} analyzed this ratio (see Theorem 3.1.), as they essentially used the same partition created here. They proved the following fact:
\begin{fact}[\cite{CCGGP98}]\label{fact:diffClustersSmall}
	For every pair of points $p,q\in\R^d$, the probability that $p$ and $q$ belong to the same clusters in $\mathcal{C}^{w_i}$ is at least $\frac{\text{Vol}_{d}\left(B(p,w_{i})\cap B(q,w_{i})\right)}{\text{Vol}_{d}\left(B(p,w_{i})\cup B(q,w_{i})\right)}\ge1-\sqrt{d}\cdot\frac{\|p-q\|_2}{w_i}$.
\end{fact}

\begin{wrapfigure}{r}{0.25\textwidth}
	\begin{center}
		\vspace{-20pt}
		\includegraphics[width=0.9\textwidth]{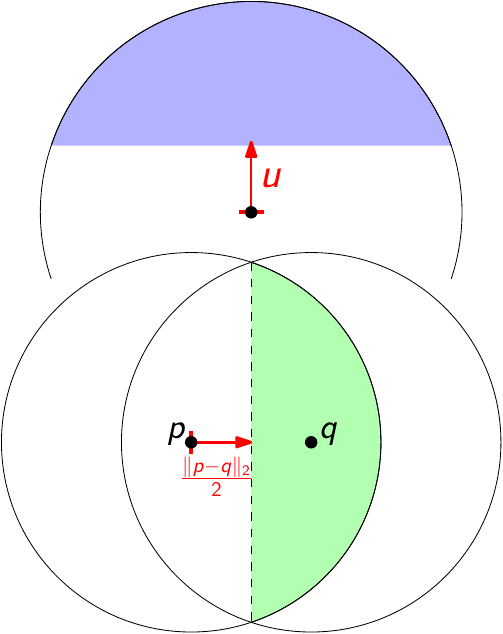}
		\vspace{-5pt}
	\end{center}
	\vspace{-15pt}
\end{wrapfigure}
\Cref{fact:diffClustersSmall} is only meaningful in the regime where $\|p-q\|_2\le\frac{w_i}{\sqrt{d}}$. For farther $p,q$ we will follow an analysis similar to Andoni \cite{AndoniThesis} (see Fact 2.3.1). 
We provide here a direct analysis to save on a second degree factor $\sqrt{d}$, that becomes significant in the important regime $t=\sqrt{d}$.
A hyperspherical cup of radius $r$ at distance $u$ from the origin is the set of all the points in the ball $B(\vec{0},r)$ with the value of the first coordinate being at least $u$. See the blue area in the illustration on the right. 
The volume of the intersection of two balls of radius $r$, is equivalent to
the volume of the union of two hyperspherical cups of radius $r$ at distance $\frac{\|p-q\|_2}{2}$ from the origin (see illustration on the right).
We provide a proof for the following technical claim in \Cref{appendix:Missing}.

\begin{restatable}{claim}{HypersphericalCup}
	\label{clm:hypersphericalCup} Denote by $V_{d}$ the volume of a unit ball in $\mathbb{R}^{d}$.
	The volume of a hyperspherical cup of radius $r$ at distance $\alpha\cdot r$
	(for $\alpha\in(\frac{1}{2\sqrt{d}},\frac{1}{2})$) from the origin
	is $\Omega(1)\cdot\frac{1}{\alpha\cdot\sqrt{d}}\cdot(1-\alpha^{2})^{\frac{d+1}{2}}\cdot V_{d}\cdot r^{d}$.
\end{restatable}

The following fact immediately follows:
\begin{fact}\label{fact:diffClustersLarge}
	For every pair of points $p,q\in\R^d$ at distance $\|p-q\|_2\in[\frac{w_i}{2\sqrt{d}},w_i]$, the probability that $p$ and $q$ belong to the same clusters in $\mathcal{C}^{w_i}$ is bounded by
	\[
	\frac{\text{Vol}_{d}\left(B(p,w_{i})\cap B(q,w_{i})\right)}{\text{Vol}_{d}\left(B(p,w_{i})\cup B(q,w_{i})\right)}=\Omega(\frac{1}{\sqrt{d}}\cdot\frac{w_i}{\|p-q\|_2})\cdot\left(1-(\frac{\|p-q\|_{2}}{2w_{i}})^{2}\right)^{\frac{d}{2}}~.
	\]
\end{fact}

\subsection{Proof of \Cref{thm:LSOEuclideanLargeStretch}}\label{subsec:RdLSOproof}

We will assume here that $d$ is large enough. Note that for constant $d$, $O(1)^d$ is a constant, and hence \Cref{thm:LSOEuclideanLargeStretch} follows from \cite{CHJ20}.
Further, we will assume that $t<\sqrt{d}$ , as otherwise, one can plug in $t=\sqrt{d}$ in \Cref{thm:LSOEuclideanLargeStretch}, and get a better stretch, with the same asymptotic number of orderings.
Fix $t\in[2,\sqrt{d}]$. Pick $\red{\gamma=c_{\gamma}\cdot\frac{d}{t^{2}}\in\N}$,
and 
$\red{\eps=\frac{1}{c_{\epsilon}}\cdot\min\left\{ \frac{1}{d},\frac{\delta }{t}\right\}=\Omega\left(\frac{\delta}{d}\right)}$
 for large enough constants $c_{\gamma},c_{\epsilon}$, to be determined later.
We begin by analyzing the success probability of a single pair in a random ordering. 

A natural way to construct the LSO will be to use fresh random bits to choose cluster centers in each scale. However, it is unclear how to succeed with all possible point pairs using such an approach (and it is also computationally inferior).
Instead, we will reuse the same choices (scaled) in many different weight scales.
Specifically, we will create well defined orderings with finite random seeds, by choosing random centers for the scales $0,1,\dots,\gamma-1$. Then, for a scale $i$ such that $i\mod\gamma=j$, we will ``reuse'' the centers created for the $j$'th scale.
More formally, in order to sample an ordering $\preceq$, for every $j\in\{0,1,\dots,\gamma-1\}$, sample centers $v^j_1,\dots,v^j_{\ell_j}\in[0,4\cdot w_j^d)$ for $w_j=\xi^j$ (where $\xi=\frac{28\sqrt{d+1}}{t}$).
Then for every distance scale $i\in\Z$, such that $i=j+k\cdot\gamma$ for some $k\in\Z$, we will simply use the centers $\xi^{k\cdot\gamma}\cdot v^j_1,\dots,\xi^{k\cdot\gamma}\cdot v^j_{\ell_j}\in[0,4\cdot w_i)^d$. The ordering $\preceq$ is now well defined.

Sample an ordering $\preceq$ as above.
Given two points $p,q$ at distance $\|p-q\|_{2}\le\frac{1}{t}$,
we say that the pair $(p,q)$ is \emph{$i$-successful} (w.r.t.  $\preceq$) if in the clustering $\mathcal{\mathcal{C}}^{w_{i}}$, the balls $B(p,\epsilon),B(q,\epsilon)$ contained in the same cluster (rather than only $p,q$). 
We say that $(p,q)$ is \emph{very-successful} (w.r.t.  $\preceq$), if $(p,q)$ is $i$-successful for every $i\ge 0$.
Note that the radius $\eps$ does not scale with $i$ here.

Fix $i\ge0$, following the discussion in \Cref{sec:highDimLSOconstruction}, and the triangle inequality, $(p,q)$ is $i$-successful if and only if the first cluster center chosen from $B(p,\xi^{i}+\eps)\cup B(q,\xi^{i}+\eps)$ is in $B(p,\xi^{i}-\eps)\cap B(q,\xi^{i}-\eps)$.
Denote by $\Psi_0$ the event that $(p,q)$ is $0$-successful. The probability that $(p,q)$ is $0$-successful is thus
\begin{align}
	\Pr\left[\Psi_{0}\right] & =\frac{\text{Vol}_{d}\left(B(p,1-\epsilon)\cap B(q,1-\epsilon)\right)}{\text{Vol}_{d}\left(B(p,1+\epsilon)\cup B(q,1+\epsilon)\right)}\nonumber\\
	& =\frac{\text{Vol}_{d}\left(B(p,1)\cup B(q,1)\right)}{\text{Vol}_{d}\left(B(p,1+\epsilon)\cup B(q,1+\epsilon)\right)}\cdot\frac{\text{Vol}_{d}\left(B(p,1)\cap B(q,1)\right)}{\text{Vol}_{d}\left(B(p,1)\cup B(q,1)\right)}\cdot\frac{\text{Vol}_{d}\left(B(p,1-\epsilon)\cap B(q,1-\epsilon)\right)}{\text{Vol}_{d}\left(B(p,1)\cap B(q,1)\right)}\nonumber\\
	& \stackrel{(*)}{=}(1+\epsilon)^{-d}\cdot\Omega(\frac{1}{\|p-q\|\cdot\sqrt{d}})\cdot\left(1-(\frac{\|p-q\|_{2}}{2})^{2}\right)^{\frac{d}{2}}\cdot(1-\epsilon)^{d}\nonumber\\
	& \stackrel{(**)}{=}\Omega(\frac{t}{\sqrt{d}})\cdot\left(1-\frac{1}{4t^{2}}\right)^{\frac{d}{2}}\ge\Omega(\frac{t}{\sqrt{d}})\cdot e^{-\frac{d}{2}\cdot\frac{1}{4t^{2}}\cdot(1+\frac{1}{2t^{2}})}~,\label{eq:Psi0}
\end{align}
Here the inequality $^{(*)}$ follows by \Cref{fact:diffClustersLarge},
and the observations:
\begin{itemize}
	\item $\frac{\text{Vol}_{d}\left(B(p,1)\cup B(q,1)\right)}{\text{Vol}_{d}\left(B(p,1+\epsilon)\cup B(q,1+\epsilon)\right)}\ge\frac{\text{Vol}_{d}\left(B(p,1)\cup B(q,1)\right)}{\text{Vol}_{d}\left(B(p,1+\epsilon)\cup B(q',1+\epsilon)\right)}=(1+\epsilon)^{-d}$,
	here $q'$ is a point at distance $(1+\epsilon)\cdot\|p-q\|_{2}$
	from $p$. The inequality holds as taking farther centers will increase
	the union, while the equality holds by symmetry.
	\item $\frac{\text{Vol}_{d}\left(B(p,1-\epsilon)\cap B(q,1-\epsilon)\right)}{\text{Vol}_{d}\left(B(p,1)\cap B(q,1)\right)}\ge\frac{\text{Vol}_{d}\left(B(p,1-\epsilon)\cap B(q,1-\epsilon)\right)}{\text{Vol}_{d}\left(B(p,1)\cap B(q'',1)\right)}=(1-\epsilon)^{d}$,
	here $q''$ is a point at distance $(1-\epsilon)\cdot\|p-q\|_{2}$
	from $p$. The inequality holds as taking closer centers will increase
	the intersection, while the equality holds by symmetry.
\end{itemize}
Equality $^{(**)}$ follows by our choice of $\eps=O(\frac1d)$ (which implies $(1+2\epsilon)^{-d},(1-2\epsilon)^{d}\le e^{-2\epsilon\cdot d}=\Omega(1)$) and $\|p-q\|_{2}\le\frac{1}{t}$. The last equality follows by the fact $1-x\ge e^{-\frac{x}{1-x}}\ge e^{-x\cdot(1+2x)}$.
%

For $i\ge1$,
following the same logic as above, and using \Cref{fact:diffClustersSmall}, the probability that $(p,q)$ is $i$-successful is 

\begin{align}
	\frac{\text{Vol}_{d}\left(B(p,\xi^{i}-\epsilon)\cap B(q,\xi^{i}-\epsilon)\right)}{\text{Vol}_{d}\left(B(p,\xi^{i}+\epsilon)\cup B(q,\xi^{i}+\epsilon)\right)} & \ge(1+\frac{\epsilon}{\xi^{i}})^{-d}\cdot(1-\sqrt{d}\cdot\frac{\|p-q\|_{2}}{\xi^{i}})\cdot(1-\frac{\epsilon}{\xi^{i}})^{d}>e^{-\left(\frac{2\sqrt{d}}{t\cdot\xi^{i}}+\frac{3\epsilon d}{\xi^{i}}\right)}>e^{-\frac{3\sqrt{d}}{t\cdot\xi^{i}}}~,\label{eq:i-succ}
\end{align}
where the second inequality follows by the facts $1-x\ge e^{-2x}$
and $1+x\le e^{x}$, and the last inequality hold for large enough
$c_{\eps}$ (as $\eps\le\frac{1}{c_{\eps}\cdot d}$).
For $i\in[1,\gamma-1]$, all these probabilities are independent.
Denote by $\Psi_1$ the event that $(p,q)$ is $i$-successful in all $i\in[1,\gamma-1]$. Then 
\begin{align}
\Pr\left[\Psi_{1}\right]\overset{(\ref{eq:i-succ})}{\ge}\Pi_{i=1}^{\gamma-1}e^{-\frac{3\sqrt{d}}{t\cdot\xi^{i}}}=e^{-\sum_{i=1}^{\gamma-1}\frac{3\sqrt{d}}{t\cdot\xi^{i}}}\ge e^{-\frac{6\sqrt{d}}{t\cdot\xi}}=e^{-\frac{1}{2}}\ge\frac{1}{2}~.\label{eq:Psi1}
\end{align}
Denote by $\Psi_2$ the event that $(p,q)$ is $i$-successful for all $i\ge\gamma$. By union bound (and $e^{-x}>1-x$)
\begin{align}
\Pr\left[\overline{\Psi_{2}}\right]\le\sum_{i\ge\gamma}1-e^{-\frac{3\sqrt{d}}{t\cdot\xi^{i}}}\le\sum_{i\ge\gamma}\frac{3\sqrt{d}}{t\cdot\xi^{i}}\le\frac{6\sqrt{d}}{t\cdot\xi^{\gamma}}\le\frac{1}{4}\cdot\Pr\left[\Psi_{0}\right]~,\label{eq:Psi2}
\end{align}
for large enough constant $c_{\gamma}$ (recall that $\xi=\frac{12\sqrt{d}}{t}$ and $\gamma=c_{\gamma}\cdot\frac{d}{t^{2}}$).
Finally, as $\Psi_0$ and $\Psi_1$ are independent, and using union bound for $\Psi_2$, we conclude 
\begin{align}
	\Pr\left[p,q\text{ are very successful}\right] & =\Pr\left[\Psi_{0}\wedge\Psi_{1}\wedge\Psi_{2}\right]\ge\Pr\left[\Psi_{0}\wedge\Psi_{1}\right]-\Pr\left[\overline{\Psi_{2}}\right]\nonumber\\
	& \ge\frac{1}{2}\cdot\Pr\left[\Psi_{0}\right]-\frac{1}{4}\Pr\left[\Psi_{0}\right]=\frac{1}{4}\cdot\Pr\left[\Psi_{0}\right]=\Omega(\frac{t}{\sqrt{d}})\cdot e^{-\frac{d}{8t^{2}}\cdot(1+\frac{1}{2t^{2}})}~.\label{eq:VerySucc}
\end{align}

Denote this probability by $p_\vsucc$.
Our analysis up to this point holds for $p,q$ such that $\|p-q\|_2\le\frac1t$, and $(p,q)$ is said to be very successful if it is $i$-successful for $i\ge 0$.
In general, for $p,q$ such that $\frac{w_{j-1}}{t}=\frac{\xi^{j-1}}{t}<\|p-q\|_2\le\frac{\xi^{j}}{t}=\frac{w_{j}}{t}$, we say that the pair $(p,q)$ is $i$-successful (w.r.t.  $\preceq$) if in the clustering $\mathcal{\mathcal{C}}^{w_{i}}$, the balls $B(p,\epsilon\cdot w_j),B(q,\epsilon\cdot w_j)$ belong to the same cluster. 
The pair $(p,q)$ is very successful if it is $i$-successful for $i\ge j$ (note that $j$ is fixed for each pair).
By symmetry, for every pair $(p,q)$, it is very successful with probability at least
$\Pr\left[p,q\text{ are very successful}\right] \ge p_\vsucc$.

For every $i\in\{0,\dots,\gamma-1\}$, consider the box $[0,8\cdot w_i)^{d}$, and let $N_i\subseteq[0,8\cdot w_i)^{d}$ be an $\eps_i=\epsilon\cdot w_i$-net: a set of points such that every two are at distance at least
$\epsilon_i$, and every point in $[0,8\cdot w_i)^{d}$ has a net point at distance at most $\epsilon_i$.
By volume argument, 
as all the balls of radius $\frac{\eps_i}{2}$ around net points are disjoint, and contained in the box $[-w_i,9\cdot w_i]^d$, we have that
\begin{align*}
	|N_{i}| & \le\frac{\text{Vol}\left([-w_{i},9\cdot w_{i}]^{d}\right)}{\text{Vol}\left(B(\vec{0},\frac{\eps_{i}}{2})\right)}=\frac{(10\cdot w_{i})^{d}}{\nicefrac{\pi^{\frac{d}{2}}(\frac{\epsilon\cdot w_{i}}{2})^{d}}{\Gamma(\frac{d}{2}+1)}}\\
	& =\epsilon^{-d}\cdot2^{O(d\log d)}\le\Omega(\frac{d}{\delta})^{d}\cdot2^{O(d\log d)}=2^{O(d\log\frac{d}{\delta})}~.
\end{align*}
Denote by $N=\bigcup_{i=0}^{\gamma-1}N_i$ the union of all net points. Then $|N|=\sum_{i=0}^{\gamma-1}|N_i|=O(\frac{d}{t^2})\cdot 2^{O(d\log \frac{d}{\delta})}=2^{O(d\log \frac{d}{\delta})}$.
We sample i.i.d. 
\[
m\coloneqq p_{\vsucc}^{-1}\cdot2\cdot\ln|N|=O\left(\frac{d^{1.5}}{t}\cdot\log\frac{d}{\delta}\cdot e^{\frac{d}{8t^{2}}\cdot(1+\frac{1}{2t^{2}})}\right)
\]
 orderings $\preceq_{1},\preceq_{2},\dots,\preceq_{m}$.
Given $i\in\{0,\dots,\gamma-1\}$ and two net points $p,q\in N_i$, the probability that $(p,q)$ is not very successful in any of the orderings is bounded by $\left(1-p_\vsucc\right)^{m}\le e^{-2\ln|N|}=|N|^{-2}$, while the number of such triples is smaller than $\frac{|N|^{2}}{2}$. 
Hence by union bound, we can assume that we picked a set of orderings such that for every $i\in\{0,\dots,\gamma-1\}$, and $p,q\in N_i$, $(p,q)$ is very successful in some ordering.

Next, we argue that the created set of orderings is an $(m,\frac{2\xi\cdot t}{1-2\eps t})$-LSO for the entire space $\R^d$ (we will later improve  the stretch guarantee).
Consider a pair $p,q\in\R^d$.
First assume that $p,q\in [1,7]^d$ and $\|p-q\|_{2}\in\left(\frac{1}{\xi\cdot t}-\frac{2\eps}{\xi},\frac{1}{t}-2\eps\right]=\left(\frac{1-2\eps t}{\xi\cdot t},\frac{1-2\eps t}{t}\right]$.
Then there are net points $p',q'\in N_0$ such that $\|p-p'\|_2,\|q-q'\|_2\le\eps$, and in particular $\|q'-p'\|_2\le \|p-q\|_2+2\eps\le \frac{1}{t}$.
Let $\preceq$ be some ordering where $p',q'$ are very successful. Hence the balls $B(p',\eps),B(q',\eps)$ are clustered together in all the hierarchical scales for $i\ge 0$. These clusters contain $p$ and $q$ as well. In particular, all the points between $p$ to $q$ in the ordering $\preceq$ belong to a single ball of radius $1$. It follows that the distance between each such pair is at most $2\le\frac{2\xi\cdot t}{1-2\eps t}\cdot\|p-q\|_{2}$. We conclude that $(p,q)$ is satisfied by the LSO.

Next consider the case where
$p,q\in[w_i,7\cdot w_i]^d$ and $\|p-q\|_{2}\in\left((\frac{1}{\xi\cdot t}-\frac{2\eps}{\xi})\cdot w_{i},(\frac{1}{t}-2\eps)\cdot w_{i}\right]$.
If $i\in\{0,\dots,\gamma-1\}$, then the same argument as above shows that $(p,q)$ is satisfied by the LSO (because there are close very successful net points in $N_i$).
For general $i$, let $j,k\in \Z$ such that $i=j+k\cdot\gamma$.
By symmetry, and the fact that we reused the same centers every $\gamma$ scales, $(p,q)$ is satisfied by the LSO if and only if $(\xi^{-k\gamma}\cdot p,\xi^{-k\gamma}\cdot q)$ is satisfied. But this follow from the previous case.

Finally, consider the case where $\|p-q\|_{2}\in\left((\frac{1}{\xi\cdot t}-\frac{2\eps}{\xi})\cdot w_{i},(\frac{1}{t}-2\eps)\cdot w_{i}\right]$ (but not necessarily $p,q\in[w^i,7\cdot w^i]^d$).
There is a vector $x\in4w_i\cdot\Z^d$ such that $p+x,q+x\in [2^i,7\cdot 2^i]^d$.
By the second case, there is an ordering $\preceq$ which satisfies $(p+x,q+x)$. By symmetry, the same ordering satisfies $(p,q)$ as well.

Next, we improve the stretch guarantee to $(1+\delta)2t$.
Set $\delta'=\frac\delta3$. For every $s\in{0,1,\dots,\lfloor\log_{1+\delta'}\xi\rfloor}$, construct the same LSO as above, with the only change that for every $i$ we will use a slightly shifted scale $w^s_i=\xi^i\cdot(1+\delta')^s$ (instead of $w_i=\xi^i$). Call the resulting LSO $\mathcal{L}_s$.
Consider a pair $(p,q)$, and let $i,s$ be the unique indices such that
\[
(\frac{1}{t}-2\eps)\cdot\xi^{i}\cdot(1+\delta')^{s-1}<\|p-q\|_{2}\le(\frac{1}{t}-2\eps)\cdot\xi^{i}\cdot(1+\delta')^{s}~.
\]
In the LSO $\mathcal{L}_s$, there is an ordering $\preceq$, where the pair $(p,q)$ is $i'$-successful for every $i'\ge i$. In particular, by the analysis above, for every $i'\ge i$, both $p,q$ are contained in a cluster of $\mathcal{C}^s_{i'}$. Further, in $\preceq$, all the points between $p$ to $q$ (and including $p,q$) are contained in a ball of radius $\xi^i\cdot (1+\delta')^s$, and thus are at distance at most $2\cdot\xi^{i}\cdot(1+\delta')^{s}\le\frac{1+\delta'}{1-2\eps t}\cdot2t\cdot\|p-q\|_{2}$.
As this holds for every pair $(p,q)$, it follows that the union  
$\cup_{s=0}^{\lfloor\log_{1+\delta'}2\rfloor}\mathcal{L}_s$ has $\frac{1+\delta'}{1-2\eps t}\cdot 2t\le\frac{1+\nicefrac{\delta}{3}}{1-\nicefrac{\delta}{3}}\cdot 2t\le(1+\delta)\cdot 2t$ stretch guarantee.
The total number of orderings is $\lfloor\log_{1+\delta'}\xi\rfloor\cdot m=O\left(\frac{d^{1.5}}{\delta\cdot t}\cdot\log(\frac{2\sqrt{d}}{t})\cdot\log\frac{d}{\delta}\cdot e^{\frac{d}{8t^{2}}\cdot(1+\frac{1}{2t^{2}})}\right)$. 
\Cref{thm:LSOEuclideanLargeStretch} now follows by replacing $t$ with $\frac t2$.

\subsection{Generalization for $\ell_p$ spaces for $p\in[1,2]$.}\label{sec:ell-p-LSO}
in this subsection we generalize \Cref{thm:LSOEuclideanLargeStretch} to general $\ell_p$ spaces. We begin with a simple corollary for the maximal stretch regime. Afterwards, we obtain the full \#orderings-stretch tradeoff for $p\in[1,2]$ 
\begin{corollary}\label{cor:LpLSO}
	The is a collection of $O(d\cdot\log d)$ orderings $\Sigma$ of all the points in $\R^d$, such that for every $p\in[2,\infty]$,  $\Sigma$ is an $\left(O\left(d\cdot\log d\right),d^{1-\frac{1}{p}}\right)$-triangle LSO for $d$-dimensional $\ell_p$ space, and for every $p\in[1,2]$, $\Sigma$ is an $\left(O\left(d\cdot\log d\right),d^{\frac{1}{p}}\right)$-triangle LSO for  $d$-dimensional $\ell_p$ space.
\end{corollary}
\begin{proof}
	\sloppy Apply \Cref{thm:LSOEuclideanLargeStretch} with parameters $\delta=\frac12$ and $t=\frac23\sqrt{d}$ to obtain a $\left(O\left(d\cdot\log d\right),\sqrt{d}\right)$-triangle LSO $\Sigma$ for the $d$-dimensional Euclidean space. We argue that $\Sigma$ fulfills the properties stated in \Cref{cor:LpLSO}.
	The following folklore fact will be crucial to our proof:	
	\begin{fact}\label{fact:norms}
		For $1\le q\le p$, and every vector $x\in\R^d$, it holds that $\|x\|_p\le \|x\|_q\le d^{\frac1q-\frac1p}\|x\|_p$.
	\end{fact}
	Consider a pair of points $x,y\in \R^d$.  There is
	an ordering $\sigma\in\Sigma$ such that w.l.o.g. $x\prec_{\sigma}y$, and for every $a,b\in \R^d$ such that $x\preceq_{\sigma}a\preceq_{\sigma}b\preceq_{\sigma}y$ it holds that $\|a-b\|_2\le d^{\frac12}\cdot \|x-y\|_2$.
	Then for every $p\in[2,\infty]$ and  $a,b\in \R^d$ such that $x\preceq_{\sigma}a\preceq_{\sigma}b\preceq_{\sigma}y$, by \Cref{fact:norms} it holds that 
	\[
	\|a-b\|_{p}\le\|a-b\|_{2}\le d^{\frac{1}{2}}\cdot\|x-y\|_{2}\le d^{1-\frac{1}{p}}\cdot\|x-y\|_{p}~.
	\]
	Similarly, for $p\in[1,2]$, and  $a,b\in \R^d$ such that $x\preceq_{\sigma}a\preceq_{\sigma}b\preceq_{\sigma}y$ it holds that
	\[
	\|a-b\|_{p}\le d^{\frac{1}{p}-\frac{1}{2}}\cdot\|a-b\|_{2}\le d^{\frac{1}{p}}\cdot\|x-y\|_{2}\le d^{\frac{1}{p}}\cdot\|x-y\|_{p}~.
	\]
	The corollary now follows.
\end{proof}

Next we prove \Cref{thm:LSOLp12}, restated bellow for convenience.
\HighDimLSOLP*
\begin{proof}
	
Denote by $B_p(x,r)=\{y\mid\|x-y\|_p\le r\}$ the $\ell_p$ ball around $x$ of radius $r$.
We will repeat the same algorithm from \Cref{sec:highDimLSOconstruction} with small changes. First, instead of using standard Euclidean balls we will use $\ell_p$ balls. Second, the gap between consecutive scales will be $\xi_p=\frac{36d}{t}$ instead of $\xi=\frac{12\sqrt{d}}{t}$ (this introduces additional $\poly(d)$ factor to the number of orderings).
Thus we will obtain random clustering $\mathcal{C}^{w_i}_p$ in each scale $w_i$, which will induce an LSO. We begin with analyzing the probability that nearby points will be clustered together. 
This boils down to the volume of the intersection between two $\ell_p$ balls. The proof of the following lemma follows the lines of a similar  proof in \cite{Nguyen13}, and is differed to \Cref{appendix:Missing}.
\begin{restatable}{lemma}{BallIntersectionLpLarge}
	\label{lem:diffClustersLargeLp} 	For $p\in[1,2]$, and every pair of points $x,y\in\R^d$ at distance $\|x-y\|_p=\frac1t\in[d^{-\frac1p},\frac13]\cdot w_i$, the probability that $x$ and $y$ belong to the same clusters in $\mathcal{C}^{w_i}_p$ is bounded by
	\[
	\frac{\text{Vol}_{d}\left(B_{p}(x,w_{i})\cap B_{p}(y,w_{i})\right)}{\text{Vol}_{d}\left(B_{p}(x,w_{i})\cup B_{p}(y,w_{i})\right)}=e^{-\Omega(\frac{d}{t^{p}})}
	\]
\end{restatable}
The following lemma is quite straightforward, and far from being tight. However, the loss will not be significant in our analysis. The proof is differed to \Cref{appendix:Missing}.
\begin{restatable}{observation}{BallIntersectionLpSmall}
	\label{obs:diffClustersSmallLp} 
	For $p\in[1,2]$, and every pair of points $x,y\in\R^d$ at distance $\|x-y\|_p\le \frac{w_i}{4d}$, the probability that $x$ and $y$ belong to the same clusters in $\mathcal{C}^{w_i}_p$ is bounded by
	\[
	\frac{\text{Vol}_{d}\left(B_{p}(x,w_{i})\cap B_{p}(y,w_{i})\right)}{\text{Vol}_{d}\left(B_{p}(x,w_{i})\cup B_{p}(y,w_{i})\right)}\ge1-4d\cdot\frac{\|x-y\|_{p}}{w_{i}}~.
	\]
\end{restatable}

The rest of the analysis follows the exact same lines as the proof of \Cref{thm:LSOEuclideanLargeStretch}, while replacing the usage of \Cref{fact:diffClustersSmall} and \Cref{fact:diffClustersLarge}  by \Cref{obs:diffClustersSmallLp} and \Cref{lem:diffClustersLargeLp}  respectively, and using a larger gap $\xi_p=\frac{36d}{t}$ between consecutive scales in a single LSO. 
The other parameters, $\red{\gamma=c_{\gamma}\cdot\frac{d}{t^{p}}\in\N}$,
and 
$\red{\eps=\frac{1}{c_{\epsilon}}\cdot\min\left\{ \frac{1}{d},\frac{\delta }{t}\right\}=\Omega\left(\frac{\delta}{d}\right)}$ remain the same, while we fix $\delta=\frac14$.

Specifically, consider a pair $x,y$ at distance $\|x-y\|_p\le\frac1t$. The definition of $i$-successful, very successful, and $\Psi_0, \Psi_1, \Psi_2$ remains the same. Following \Cref{eq:Psi0} and using \Cref{lem:diffClustersLargeLp}, 
\[
\Pr\left[\Psi_{0}\right]=(1+\epsilon)^{-d}\cdot e^{-\Omega(\frac{d}{t^{p}})}\cdot(1-\epsilon)^{d}=e^{-\Omega(\frac{d}{t^{p}})}~.
\]
Following \Cref{eq:i-succ} and using \Cref{obs:diffClustersSmallLp}, the probability that $(x,y)$ is  $i$-successful is 
\[
\frac{\text{Vol}_{d}\left(B_{p}(x,\xi_{p}^{i}-\epsilon)\cap B_{p}(y,\xi_{p}^{i}-\epsilon)\right)}{\text{Vol}_{d}\left(B_{p}(x,\xi_{p}^{i}+\epsilon)\cup B_{p}(y,\xi_{p}^{i}+\epsilon)\right)}\ge(1+\frac{\epsilon}{\xi_{p}^{i}})^{-d}\cdot(1-4d\cdot\frac{\|p-q\|_{p}}{\xi_{p}^{i}})\cdot(1-\frac{\epsilon}{\xi_{p}^{i}})^{d}>e^{-\left(\frac{8d}{t\cdot\xi_{p}^{i}}+\frac{3\epsilon d}{\xi_{p}^{i}}\right)}>e^{-\frac{9d}{t\cdot\xi_{p}^{i}}}~.
\]
Hence by independence (following \Cref{eq:Psi1}), and by union bound (following \Cref{eq:Psi2}) we have 
\begin{align*}
	\Pr\left[\Psi_{1}\right] & \ge\Pi_{i=1}^{\gamma-1}e^{-\frac{9d}{t\cdot\xi_{p}^{i}}}=e^{-\sum_{i=1}^{\gamma-1}\frac{9d}{t\cdot\xi_{p}^{i}}}\ge e^{-\frac{18d}{t\cdot\xi_{p}}}=e^{-\frac{1}{2}}\ge\frac{1}{2}~.\\
	\Pr\left[\overline{\Psi_{2}}\right] & \le\sum_{i\ge\gamma}1-e^{-\frac{9d}{t\cdot\xi_{p}^{i}}}\le\sum_{i\ge\gamma}\frac{9d}{t\cdot\xi_{p}^{i}}\le\frac{18d}{t\cdot\xi_{p}^{\gamma}}\le\frac{1}{4}\cdot\Pr\left[\Psi_{0}\right]~.
\end{align*}

Following \Cref{eq:VerySucc}, the probability that $(x,y)$ is very successful is at least $p_\vsucc\ge\frac{1}{4}\cdot\Pr\left[\Psi_{0}\right]=e^{-\Omega(\frac{d}{t^{p}})}$.
We pick $N_i\subseteq[0,8\cdot w_i)^{d}$ an $\eps_i=\epsilon\cdot w_i$-net (w.r.t. $\|\cdot\|_p$) of size $2^{O(d\log \frac{d}{\delta})}$, and denote $N=\bigcup_{i=0}^{\gamma-1}N_i$. 
Sampling $m_{p}\coloneqq p_{\vsucc}^{-1}\cdot2\cdot\ln|N|=d\cdot\log d\cdot e^{O(\frac{d}{t^{p}})}$ i.i.d. orderings we can ensure that every pair  of points $x,y\in N_i$ is very successful in some ordering. Following the same arguments as in the proof of \Cref{thm:LSOEuclideanLargeStretch} (triangle inequality), we conclude that this ordering is an $(m,\frac{2\xi_p\cdot t}{1-2\eps t})$-LSO for the entire space $(\R^d,\|\cdot\|_p)$. By adding $\log \xi_p=O(\log\frac{2d}{t})$ different shifts of this construction, we obtain an $\left(m_{p}\cdot\log\xi_{p},(1+\delta)\cdot2t\right)=\left(d\cdot\log d\cdot\log\frac{2d}{t}\cdot e^{O(\frac{d}{t^{p}})},(1+\delta)\cdot2t\right)
$-triangle LSO.
\Cref{thm:LSOLp12} now follows by replacing $t$ with $\frac{2t}{5}$.

\end{proof}

\section{LSO for high dimensional doubling spaces}\label{sec:LSOdoubling}
This section is devoted to proving the following theorem (restated for convenience):
\DoublingLSO*

We begin by introducing the notion of \emph{padded partition cover scheme}.
This notion is closely related to padded decompositions and sparse covers (see e.g. \cite{Fil19padded} and references therein).%
\begin{definition}[Padded Partition Cover Scheme]\label{def:PaddedPartitionCover}
	A partition $\mathcal{P}$ of a metric space $(X,d_X)$ is $\Delta$-bounded if the diameter of every cluster $C\in\mathcal{P}$ is
	at most $\Delta$. 
	A collection of partitions $\mathcal{P}_{1},\dots,\mathcal{P}_{s}$
	is $(\tau,\rho,\Delta)$-padded partition cover if (a) $s\le \tau$, (b) every partition $\mathcal{P}_{i}$ is $\Delta$-bounded, and (c) for every point $x$, there is a cluster $C$ in one of the partitions $\mathcal{P}_{i}$ such that $B(x,\frac{\Delta}{\rho})\subseteq C$.\\
	A space $(X,d_{X})$ admits a $(\tau,\rho)$-\emph{padded partition cover scheme} if for every $\Delta$, it admits a  $(\tau,\rho,\Delta)$-padded  partition cover.
\end{definition}

Filtser \cite{Fil19padded} implicitly constructed padded partition cover scheme for doubling metrics. Explicitly \cite{Fil19padded} only argued that the construction is sparse cover,  \footnote{$(\tau,\rho,\Delta)$-	Sparse cover is a collection $\mathcal{C}$ of non-disjoint clusters of diameter at most $\Delta$, such that every vertex belongs to at most $\tau$ clusters, and every ball of radius $\nicefrac{\Delta}{\rho}$ is contained in some cluster. Note that the union of all the partitions in a $(\tau,\rho,\Delta)$-padded partition cover is a $(\tau,\rho,\Delta)$-cover (but not necessarily vice versa).} however it is actually a padded partition cover scheme.
\begin{theorem}[\cite{Fil19padded}]\label{thm:DdimCover}
	Let $(X,d_X)$ be a metric space with doubling dimension $d$ and parameter $t=\Omega(1)$. Then $G$ admits a $\left(t,O(2^{O(\nicefrac{d}{t})}\cdot d\cdot\log t)\right)$-padded partition cover scheme. 		
\end{theorem}

An \emph{ultrametric} $\left(X,d\right)$ is a metric space satisfying a strong form of the triangle inequality, that is, for all $x,y,z\in X$,
$d(x,z)\le\max\left\{ d(x,y),d(y,z)\right\}$. A related notion is \emph{hierarchical well-separated tree} (HST).
\begin{definition}[HST]\label{def:HST}
	A metric $(X,d_X)$ is an hierarchical well-separated tree (HST) if there exists a bijection $\varphi$ from $X$ to leaves of a rooted tree $T$ in which:
	\begin{enumerate}[noitemsep]
		\item Each node $v\in T$ is associated with a label $\Gamma_{v}$ such that $\Gamma_{v} = 0$ if $v$ is a leaf and $\Gamma_{v}\geq \Gamma_{u}$ if $v$ is an internal node and $u$ is any child of $v$.
		\item $d_X(x,y) = \Gamma_{\lca(\varphi(x),\varphi(y))}$ where $\lca(u,v)$ is the least common ancestor of any two given nodes $u,v$ in $T$. 
	\end{enumerate}
\end{definition}
It is well known that any ultrametric an HST are equivalent definitions (see \cite{BLMN05}).

\paragraph*{Ultrametric cover.~}  
Filtser and Hung \cite{FL22} introduced the notion of ultrametric cover.
Consider a  metric space $(X,d_X)$, a distance measure $d_Y$ is said to be dominating if $\forall x,y\in X$, $d_X(x,y)\le d_Y(x,y)$. 
\begin{definition}[Ultrametric Cover]\label{def:UltrametricCover}
	A \emph{$(\tau,\rho)$-ultrametric cover} for a space $(X,d)$ is a collection of at most $\tau$ dominating ultrametrics $\mathcal{U} = \{(U_i,d_{U_i})\}_{i=1}^{\tau}$ over $X$, such that for every $x,y\in X$, $ d_X(x,y)\le \min_{1\le i\le\tau}d_{U_i}(x,y)\le \rho\cdot d_X(x,y)$.
\end{definition} 

In \cite{FL22} it was shown that given an ultrametric cover, one can construct a triangle-LSO.
\begin{restatable}{lemma}{CoverToTriangleLSO}[\cite{FL22}]
	\label{thm:CoverToTriangleLSO} If a metric $(X,d_X)$ admits a $(\tau,\rho)$-ultrametric cover $\mathcal{U}$, then it has a $\left(\tau, \rho\right)$-triangle-LSO.
\end{restatable}

Thus the main technical part of this section is the construction of ultrametric covers for doubling spaces. 
We will prove a reduction from padded partition cover to ultrametric cover. 
In essence, ultrametric is simply a hierarchical partition.
Thus, this reduction takes unrelated partitions in all possible distance scales, and combines them into hierarchical/laminar partition. Reductions similar in spirit were constructed in the context of the Steiner point removal problem \cite{Fil20scattering}, stochastic Steiner point removal \cite{EGKRTT14}, universal Steiner tree \cite{JLNRS05}, as well as for ultrametric covers \cite{FL22}, and others.
We follow here a bottom-up approach, where the ratio between consecutive scales in a single hierarchical partition (a.k.a. ultrametric) is $O(\frac\rho\eps)$. When constructing the next level in the  hierarchical partition, we take partitions from a padded partition cover of the current scale, and slightly ``round'' them around the ``borders'' so that no previously created cluster will be divided (see \Cref{fig:Laminar}). The argument is that due to the large ratio between consecutive scales, the effects of this rounding are marginal. The proof appears bellow
\begin{lemma}\label{lem:PartitionCoverSchemeToUltrametric}
	Suppose that a metric space $(X,d_{X})$ admits a $(\tau,\rho)$-padded
	partition cover scheme. Then for every $\epsilon\in(0,\frac{1}{2})$, $(X,d_{X})$  admits an  $\left(O(\frac{\tau}{\epsilon}\log\frac{\rho}{\epsilon}),\rho(1+\epsilon)\right)$-ultrametric cover.
\end{lemma}
By plugging \Cref{thm:DdimCover} (\cite{Fil19padded})  into \Cref{lem:PartitionCoverSchemeToUltrametric} (with $\eps=\frac14$ and $t'=\frac{t}{1+\eps}$) we obtain that every metric space with doubling dimension $d$ admits an $\left(2^{O(\nicefrac{d}{t})}\cdot d\cdot\log^{2}t,t\right)$-ultrametric cover. Now, \Cref{thm:LSOdoublingLargeStretch} follows by \Cref{thm:CoverToTriangleLSO} (\cite{FL22}).
\begin{proof}[Proof of \Cref{lem:PartitionCoverSchemeToUltrametric}]
	Assume w.l.o.g. that the minimal pairwise distance in $X$ is $\red{1}$, while the maximal  pairwise distance is $\red{\Phi}$.	
	Fix $\red{\epsilon}\in(0,\frac{1}{4})$ and $\red{c}\ge 1$ to be determined later.
	For $i\ge0$, set
	$\red{\Delta_{i}=c\cdot(\frac{4\rho}{\epsilon})^{i}}$, and let
	$\red{\mathbb{P}_i=\{\mathcal{P}_{1}^{i},\dots,\mathcal{P}_{\tau}^{i}\}}$ be a $(\tau,\rho,\Delta_{i})$-padded partition cover (we assume that $\mathbb{P}_i$ has exactly $\tau$ partitions; we can enforce this assumption by duplicating partitions if necessary). 
	Fix some $j$, let $\mathcal{P}^{-1}_{j}$ be the partition where each vertex is a singleton, and consider $\{\mathcal{P}^{i}_{j}\}_{i\ge -1}$. We will inductively define a new set of partitions, enforcing it to be a laminar system. The basic idea is to produce a tree of partitions where the lower level is a refinement of the higher level, and we do so by grouping a cluster at a lower level to one of the clusters at a  higher level separating it. 
	
	The lowest level $\mathcal{P}^{-1}_{j}$ where each set in the partition is a singleton, stays as-is. 
	Inductively, for any $i\geq 0$, after constructing $\tilde{\mathcal{P}}_{j}^{i-1}$ from $\mathcal{P}_{j}^{i-1}$, we will construct $\tilde{\mathcal{P}}_{j}^{i}$
	from $\mathcal{P}_{j}^{i}$ using $\tilde{\mathcal{P}}_{j}^{i-1}$.
	Let $\mathcal{P}_{j}^{i}=\left\{ C_{1},\dots,C_{\phi}\right\}$ be all sets in the partition $\mathcal{P}_{j}^{i}$. For each $q\in[1,\phi]$, let $\red{Y_{q}}=X\setminus\cup_{a<q}\tilde{C}_{a}$ be the set of unclustered points (w.r.t. level $i$). 
	Let	$\red{C'_{q}}=C_{q}\cap Y_{q}$, and let $\red{S_{C'_{q}}}=\left\{ C\in\tilde{\mathcal{P}}_{j}^{i-1}\mid C\cap C'_{q}\ne\emptyset\right\} $ be  the set of new level-$(i-1)$ clusters intersecting $C'_{q}$.
	We set $\red{\tilde{C}_{q}}=\cup S_{C'_{q}}$ to be the union of these clusters and continue iteratively.
	See \Cref{fig:Laminar} for illustration.
	Clearly, $\tilde{\mathcal{P}}_{j}^{i-1}$
	is a refinement of $\tilde{\mathcal{P}}_{j}^{i}$. We conclude that
	$\left\{ \tilde{\mathcal{P}}_{j}^{i}\right\} _{i\ge-1}$ is a laminar hierarchical set of partitions that refine each other.
	
	\begin{figure}[t]
		\centering
		\includegraphics[width=1\textwidth]{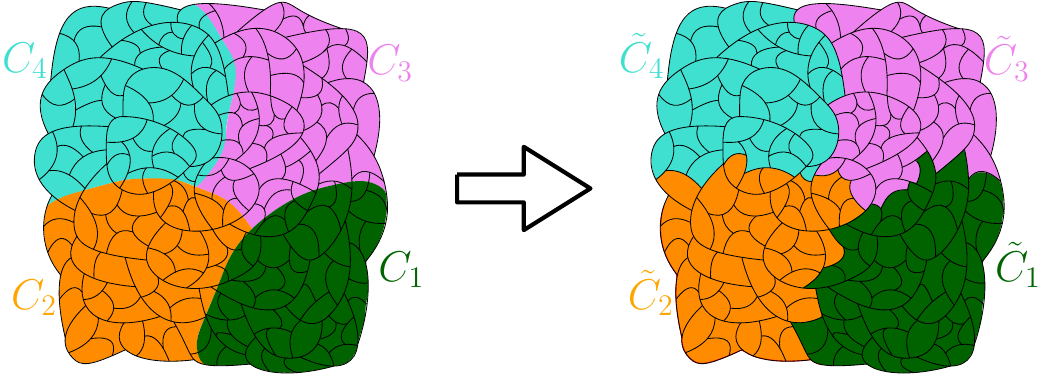}
		\caption{\footnotesize{Illustration of the construction of the partition $\tilde{\mathcal{P}}_{j}^{i}$ given  $\mathcal{P}_{j}^{i}$ and $\tilde{\mathcal{P}}_{j}^{i-1}$.
				The black lines in both the left and right parts of the figure border clusters in $\tilde{\mathcal{P}}_{j}^{i-1}$. On the left illustrated the partition $\mathcal{P}_{j}^{i}=\{C_1,C_2,C_3,C_4\}$, where different clusters colored by different colors. On the right illustrated the modified partition $\mathcal{P}_{j}^{i}=\{\tilde{C}_1,\tilde{C}_2,\tilde{C}_3,\tilde{C}_4\}$. $\tilde{C}_1$ contains all the clusters in $\tilde{\mathcal{P}}_{j}^{i-1}$ intersecting $C_1$. $C_2$ contains all the clusters in $\tilde{\mathcal{P}}_{j}^{i-1}\setminus \tilde{C}_1$ intersecting $C_2$, and so on.
		}}
		\label{fig:Laminar}
	\end{figure}
	
	We argue by induction that $\tilde{\mathcal{P}}_{j}^{i}$ has diameter $\Delta_{i}(1+\epsilon)$.
	Consider $\tilde{C}_{q}\in\tilde{\mathcal{P}}_{j}^{i}$; it consists
	of $C'_{q}\subseteq C_{q}\in\mathcal{P}_{j}^{i}$ and of clusters in $\tilde{\mathcal{P}}_{j}^{i-1}$
	intersecting $C'_{q}$. 
	As the diameter of $C'_{q}$ is bounded by $\diam(C_{q})\le \Delta_i$, and by the induction hypothesis, the diameter of each cluster $C\in \tilde{\mathcal{P}}_{j}^{i-1}$ is bounded by $(1+\eps)\Delta_{i-1}$, we conclude that the diameter of $\tilde{C}_{q}$ is bounded by
	\[
	\Delta_{i}+2\cdot(1+\epsilon)\Delta_{i-1}=\Delta_{i}\left(1+\frac{2(1+\epsilon)}{4\rho}\cdot\epsilon\right)\le\Delta_{i}(1+\epsilon)~,
	\]
	since $\rho \geq 1$ and $\eps < 1$.
	
	Next we argue that 
	$ \tilde{\mathbb{P}}_i = \{\tilde{\mathcal{P}}_{1}^{i},\dots,\tilde{\mathcal{P}}_{\tau}^{i}\}$ is a $(\tau,(1+\eps)^2\rho,(1+\eps)\Delta)$-padded partition cover. Observe that it contains $\tau$ partitions, and we have shown that all the clusters have diameter at most $(1+\eps)\Delta$. Thus,  it remains  to prove the padding property.
	Consider a vertex $v\in X$, there is some index $j$ such that $B(x,\frac{\Delta_{i}}{\rho})\subseteq C_{i}\in\mathcal{P}_{j}^{i}$.
	We argue that $B(x,\frac{(1+\eps)\Delta_{i}}{(1+\eps)^{2}\rho})\subseteq\tilde{C}_{i}\in\mathcal{P}_{j}^{i}$.
	Note that if a cluster $\tilde{C}_{i-1}$ is not contained in $\tilde{C}_{i}$, then it must contain a vertex out of $C_{i}$.
	In other words, every vertex $y\notin \tilde{C}_{i}$ belongs to a cluster $\tilde{C}_{y}\in\mathcal{P}_{j}^{i-1}$ containing  vertex $z\notin C_{i}$.
	Necessarily $d_{X}(x,z)>\frac{\Delta_{i}}{\rho}$. Using the triangle inequality 
	\[
	d_{X}(x,y)\ge d_{X}(x,z)-d_{X}(z,y)>\frac{\Delta_{i}}{\rho}-(1+\eps)\Delta_{i-1}=\frac{\Delta_{i}}{\rho}\left(1-\frac{(1+\eps)\epsilon}{4}\right)\ge\frac{\Delta_{i}}{\rho(1+\epsilon)}=\frac{(1+\eps)\Delta_{i}}{(1+\eps)^{2}\rho}~.
	\]
	If follows that  $B(x,\frac{\Delta_{i}(1+\epsilon)}{\rho(1+\epsilon)^{2}})\subseteq\tilde{C}_{i}\in\mathcal{\tilde{P}}_{j}^{i}$.

	Finally, we construct an ultrametric cover.
	Fix an index $j\in [1,\tau]$; we construct an HST $U_j$ as follows. Leaves of $U_j$ bijectively correspond to points in $X$ and have label $0$. For each $i \in [0,I]$ where $I = \lceil \log_{\nicefrac{4\rho}{\epsilon}}\nicefrac\Phi c\rceil$, internal nodes at level $i$ bijectively correspond to the clusters $\tilde{\mathcal{P}}_j^i$
	(leaves of $U_j$ is at level $-1$), and have label $(1+\eps)\Delta_i$. There is an edge from each node corresponding to a cluster $\tilde{C}_{i-1}\in\tilde{\mathcal{P}}_j^{i-1}$ to the node corresponding the unique cluster  $\tilde{C}_{i}\in\tilde{\mathcal{P}}_j^{i}$ containing $\tilde{C}_{i-1}$.
	The root of $U_j$ is the unique single cluster in $\tilde{\mathcal{P}}_j^{I}$. 
	Clearly, the ultrametric cover $\{U_j\}_{j=1}^{\tau}$ is dominating.
	
	To bound the stretch, we will construct such an ultrametric cover with  $c=(1+\eps)^l$ for every $l \in [0, \lfloor\log_{1+\epsilon}\frac{4\rho}{\epsilon}\rfloor]$. The final ultrametric cover will be a union of these $O(\log_{1+\epsilon}\frac{4\rho}{\epsilon})$ ultrametric covers. Clearly, their cardinality is bounded by $\tau\cdot O(\log_{1+\epsilon}\frac{4\rho}{\epsilon})=O(\frac{\tau}{\eps}\log\frac{\rho}{\epsilon})$.
	
	Consider a pair $x,y\in X$.
	Let $l\in[0,\lfloor\log_{1+\epsilon}\frac{4\rho}{\epsilon}\rfloor]$, and $i\ge0$ be the unique indices such that $(1+\epsilon)^{l-1}(\frac{4\rho}{\epsilon})^i\leq (1+\epsilon)\rho \cdot d_{X}(x,y)\le(1+\epsilon)^{l}(\frac{4\rho}{\epsilon})^i$. 
	For the ultrametric cover constructed w.r.t. $c=(1+\epsilon)^{l}$, $\Delta_{i}=c\cdot(\frac{4\rho}{\epsilon})^{i}$, and hence there is some index $j$, and a cluster $\tilde{C}_i\in\tilde{\mathcal{P}}_j^i$ such that $B(x,\frac{\Delta_{i}(1+\epsilon)}{\rho(1+\epsilon)^{2}})\subseteq\tilde{C}_{i}\in\mathcal{\tilde{P}}_{j}^{i}$. It holds that
	\[
	d_{X}(x,y)=\frac{(1+\epsilon)\rho\cdot d_{X}(x,y)}{(1+\epsilon)\rho}\le\frac{(1+\epsilon)^{l}\cdot(\frac{4\rho}{\epsilon})^{i}}{(1+\epsilon)\rho}=\frac{c\cdot(\frac{4\rho}{\epsilon})^{i}}{(1+\epsilon)\rho}=\frac{(1+\epsilon)\Delta_{i}}{(1+\epsilon)^{2}\rho}~.
	\]
	We conclude that in $U_j$, both $x,y$ are decedents of a node with  label $(1+\epsilon)\Delta_{i}\le(1+\epsilon)^{2}\rho\cdot d_{X}(x,y)$; the stretch guarantee follows.
	
	In summary, we have constructed an $\left(O(\frac{\tau}{\epsilon}\log\frac{\rho}{\epsilon}),\rho(1+\epsilon)^2\right)$-ultrametric cover.
	The lemma follows by rescaling $\eps$.
\end{proof}

\begin{remark}
	Our construction of triangle-LSO's for doubling metrics is much simpler than for the case of Euclidean space. One may ask, why did we went through an elaborate direct proof for the  Euclidean case, instead of simply constructing padded partition cover and using \Cref{lem:PartitionCoverSchemeToUltrametric}?
	Interestingly, padded covers (as well as padded partitions \cite{Fil19padded}) in Euclidean spaces have similar tradeoff to general doubling spaces: padding parameter $t$ requires $\approx2^{O(\frac{d}{t})}$ partitions.
	The crucial structural property we use here, is that if one wishes to ensure that every pair at distance $\frac\Delta t$ is clustered together (instead of the entire ball $B(x,\frac\Delta t)$), then in Euclidean space $\approx2^{O(\frac{d}{t^2})}$ partitions are enough, a significant improvement!
\end{remark}

\section{Labeled Nearest Neighbor Search}\label{sec:LabeledNNS}
In\aidea{Our techniques should work for $k$-NNS as well} this section we construct our labeled NNS data structures for various metric spaces. See \Cref{tab:LabeledNNSFamilies} for a summary of our results.
In \Cref{subsec:NNSfromRooted} we prove a meta \Cref{thm:rootedToNNS} turning rooted LSO into labeled NNS. We conclude efficient labeled NNS for planar and fixed minor free graphs (\Cref{cor:PlanarNNS}), treewidth graphs (\Cref{cor:TreeWidthNNS}), and metric spaces with small correlation dimension (\Cref{cor:CorrelationNNS}).
Then, in \Cref{subsec:NNSfromTriangle} we prove a meta \Cref{thm:rootedToNNS} turning rooted LSO into labeled NNS.
We conclude a labeled NNS for graphs with relatively large doubling dimension (\Cref{cor:doublingNNS}), and for general metric spaces (\Cref{cor:GeneralNNSfromLSO}).
Later in \Cref{subsec:NNSGeneral} we directly construct labeled NNS for general metrics to obtain a very small query time (\Cref{thm:NNSfromRamsey}). We then focus on the $O(\log N)$-stretch regime, and construct very efficient labeled NNS structures, under various relaxations (\Cref{thm:NNSfromClan} and Corollaries \ref{cor:NNSfromRamseyAspectRatio}, \ref{cor:NNSfromClanAspectRatio}, \ref{cor:NNSfromRamseyFailureProb}, \ref{cor:NNSfromClanFailureProb}), see \Cref{tab:GeneralNNS} for a summary.
Finally, in \Cref{sec:NNSGeneralLB} we show that the classic information theoretic lower bound holds for labeled NNS as well (\Cref{thm:NNS-LB}).

\subsection{Data Structures preliminaries}\label{subsec:NNSprelims}
In our labeled NNS construction, we will be using the following data structure by Willard \cite{Wil83}:%
\begin{theorem}[Y-fast trie, \cite{Wil83}]\label{thm:Y-fast-trie}
	Given a subset $S\subseteq [N]$ of $n$ elements, there is a data structure with space $O(n)$, and such that given a query $q\in [N]$, returns in $O(\log\log N)$ time the predecessor and successor of $q$ in $S$.
	Further, the data structure supports insertion and deletion of elements in $S$ in amortized $O(\log\log N)$ time.
\end{theorem}
We will use the Y-fast trie  to obtain additional information:
\begin{observation}\label{obs:Y-fast-trie}
	When constructing the Y-fast trie, one can also answer minimum element in $S$ queries in $O(1)$ time, while all the other parameters remain unchanged.
\end{observation}
\begin{proof}
	For simplicity, we will assume that the data structure works for the elements $\{0,\dots,N\}$ instead of $\{1,\dots,N\}$, and that the element $0$ is never part of $N$. Such assumption could be ensured by building the data structure for $[N+1]$, and shifting all the elements one position to the right.
	
	When constructing the data structure in the first time we simply store the minimum element $m_S$. 
	On an insert operation of an element $x$, we check whether $x<m_S$, if this is the case we update $m_S\leftarrow x$ (otherwise do nothing).
	On a remove operation of an element $x$, we check whether $x=m_S$. If no, do nothing. Else, we query the successor of $0$, and set it to be the new $m_S$. Note that the total removal time is still $O(\log\log N)$ amortized time. 
	Clearly the other parameters (space, predecessor/successor queries) remain the same, while we can answer minimum element queries in $O(1)$ time.
\end{proof}

\subsection{Labeled NNS From Rooted LSO}\label{subsec:NNSfromRooted}

In this subsection we prove a meta theorem, turning rooted-LSO's into labeled NNS. As a corollaries we obtain labeled NNS for planar graphs, fixed minor free graphs, treewidth graphs, and for metric spaces with small correlation dimension.
\begin{theorem}[Meta theorem: rooted LSO to labeled NNS]\label{thm:rootedToNNS}
Consider a metric space $(X,d_X)$ admitting a $(\tau,\rho)$-rooted NNS.
Than there is a dynamic labeled $\rho$-NNS with label size $O(\tau)$, space $O(n\cdot \tau)$,
query time $O(\tau)$, and deletion/insertion in 
$O(\tau\cdot\log\log N)$ amortized time.
\end{theorem}
\begin{proof}
Let $\Sigma$ be the promised collection of orderings of the rooted LSO.
For every point $y\in X$, let $\Sigma_y\subseteq \Sigma$ be the subset of orderings $y$ belongs to. 
We construct a label $\ell_y$ for $y$ as follows: for every $\sigma\in\Sigma_y$ store the name of $\sigma$, the position of $y$ in the order $\sigma$, and the distance $d_X(y,x_\sigma)$ from $y$ to the first vertex $x_\sigma$ in $\sigma$. Note that the label size consist of $O(|\Sigma_y|)=O(\tau)$ words.

We will use the Y-fast trie from \Cref{thm:Y-fast-trie}, together with \Cref{obs:Y-fast-trie}.
Consider a subset $P\subseteq X$ of $n$ points.	To construct the labeled NNS, we will first store all the labels of $P$ points. Next, for every ordering $\sigma\in \Sigma$ that contains at least one point from $P$, we will construct a Y-fast trie with the inputs being the positions in $\sigma$ of $P$ points. Each point in the Y-fast trie will be linked to the relevant label. 
This finishes the data structure. Note that the total space is $O(n\cdot \tau)$. Moreover, we can add/remove a point by simply adding/removing its label, and updating the Y-fast tries of all the orderings it belongs to. Such an update will take $O(\tau\cdot\log\log N)$ amortized time.

Given a query $q\in X$ with label $\ell_q$, we go over all the orderings $\Sigma_q$ containing $q$. For each $\sigma\in \Sigma_q$ we use the Y-fast trie to obtain the minimal element $y_\sigma$ in the order $\sigma$ (among the points in $P$). Denote by $x_\sigma$ the first element in the order (globally w.r.t. $X$). Our answer will be the element $y_\sigma$ minimizing $d_X(q,x_\sigma)+d_X(x_\sigma,y_\sigma)$ over all (non empty) $\sigma\in \Sigma_q$. Note that extracting the minimal element and computing the sum of distances will take us $O(1)$ time. Thus the total query time is $O(\tau)$.

Finally, we argue that the returned point is a $\rho$ approximation to the nearest neighbor. Let $p\in P$ be the closest point to $q$. Then there is an ordering $\sigma'\in\Sigma_q$ such that $d_{X}(q,x_{\sigma'})+d_{X}(x_{\sigma'},p)\le\rho\cdot d_{X}(q,p)$. Note that $y_{\sigma'}\preceq_{\sigma'}p$, thus $d_{X}(x_{\sigma},y_{\sigma'})\le d_{X}(x_{\sigma},p)$. We conclude
\begin{align*}
	\min_{\sigma\in\Sigma_{q}}\left\{ d_{X}(q,x_{\sigma})+d_{X}(x_{\sigma},y_{\sigma})\right\}  & \le d_{X}(q,x_{\sigma'})+d_{X}(x_{\sigma'},y_{\sigma'})\\
	& \le d_{X}(q,x_{\sigma'})+d_{X}(x_{\sigma'},p)\quad\le(1+\epsilon)\cdot d_{X}(q,p)~.
\end{align*}

\end{proof}

\begin{remark}
	If one wishes to obtain a worst case update time (instead of amortized) one can replace the Y-fast tries with heaps. All the parameters will remain the same, other that the deletion/insertion that will become $O(\tau\log n)$ (worst case).
\end{remark}
By plugging in the rooted-LSO's of \cite{FL22} (see \Cref{tab:LSO}) into the \Cref{thm:rootedToNNS} we conclude:
\begin{corollary}\label{cor:PlanarNNS}
\sloppy For every $\eps\in(0.\frac12)$, planar graphs (and fixed minor free graphs) admit labeled $(1+\eps)$-NNS 
with $O(\frac{1}{\epsilon}\log^{2}N)$ label size, space
$n\cdot O(\frac{1}{\epsilon}\log^{2}N)$, query time $O(\frac{1}{\epsilon}\log^{2}N)$,
and deletion/insertion $O(\frac{1}{\epsilon}\log^{2}N\log\log N)$ amortized time. 
\end{corollary}

\begin{corollary}\label{cor:TreeWidthNNS}
	\sloppy Every treewidth $k$ graph admits an (exact) labeled $1$-NNS 
	with label size $O(k\log N)$, space $O(n\cdot k\log N)$,
	query time $O(k\log N)$, and deletion/insertion in 
	$O(k\log N\log\log N)$ amortized time.
\end{corollary}

\subsubsection{Correlation dimension}\label{subsec:correlation}
Correlation dimension was introduced by Chan and Gupta \cite{CG12} as a generalization of doubling dimension. While in doubling dimension the defining property is local (every ball of radius $r$ can be covered by at most $2^d$ balls of radius $\frac r2$), in correlation dimension this property holds only on average. In particular, a metric might have a constant correlation dimension, while containing a $\sqrt{N}$-size clique.

Given a metric space $(X,d_X)$, an $r$-net is a subset $N\subseteq X$ such that every pair of net points is at distance at least $r$: $\min_{x,y\in N}d_X(x,y)\ge r$, and every point is at distance at most $r$ from some net point: $\max_{x\in X}\min_{y\in N}d_X(x,y)\le r$.
$(X,d_X)$ is said to have \emph{correlation dimension} $k$ if for every net $N$ and $s>0$, it holds that 
$$\sum_{x\in N}\left|N\cap B_{X}(x,2s)\right|\le2^{k}\cdot\sum_{x\in N}\left|N\cap B_{X}(x,s)\right|~.$$
Chan and Gupta \cite{CG12} 
(generalizing the embedding of \cite{Talwar04}) 
showed that every $N$-point metric space $(X,d_X)$ with correlation dimension $k$ embeds into a graph with treewidth $\tilde{O}_{k,\eps}(\sqrt{N})$ and multiplicative distortion $1+\eps$. Specifically, there is a map $f$ from $X$ into the vertices of a graph $G$, where $G$ has treewidth $\tilde{O}(\sqrt{N})\cdot \alpha(k,\eps)$ ($\alpha$ being an arbitrary function of $\eps,k$), such that for every $x,y\in X$, $d_X(x,y)\le d_G(f(x),f(y)) \le (1+\eps)\cdot d_X(x,y)$.
Using this embedding, and \Cref{cor:TreeWidthNNS}, we conclude:
\begin{corollary}\label{cor:CorrelationNNS}
	\sloppy Every graph with correlation dimension $k$ admit labeled $(1+\eps)$-NNS 
	with label size $\tilde{O}_{k,\eps}(\sqrt{N})$, space $\tilde{O}_{k,\eps}(n\cdot\sqrt{N})$,
	query time $\tilde{O}_{k,\eps}(\sqrt{N})$, and deletion/insertion in 
	$\tilde{O}_{k,\eps}(\sqrt{N})$ time.
\end{corollary}

Note that by our lower bound (\Cref{thm:NNS-LB}), there is a $\sqrt{N}$-vertex graph $G$ such that every labeled $t$-NNS for $t<3$ has label size $\tilde{\Omega}(\sqrt{N})$. There is an $N$-vertex graph $H$ with constant correlation dimension, containing $G$ as a subgraph, and such that the induced shortest path metric on $G$ vertices remains the same. It follows that every $t$-labeled NNS for $t<3$ for $H$ has label size $\tilde{\Omega}(\sqrt{N})$.
It follows that \Cref{cor:CorrelationNNS} is tight up to second order terms.

\subsection{Labeled NNS From Triangle LSO}\label{subsec:NNSfromTriangle}
In this subsection we prove a meta theorem, turning triangle-LSO's into labeled NNS:
\begin{theorem}[Meta theorem: triangle LSO to labeled NNS]\label{thm:triangleToNNS}
Consider a metric space $(X,d_X)$ admitting a $(\tau,\rho)$-triangle NNS.
Then there is a dynamic labeled $2\rho$-NNS with label size $O(\tau\cdot\log N)$, space $O(n\cdot \tau\cdot\log N)$,
query time $O(\tau\cdot\log\log N)$, and deletion/insertion in 
$O(\tau\cdot\log\log N)$ amortized time.
\end{theorem}
\begin{proof}
Let $\Sigma$ be the promised collection of orderings of the triangle LSO.
For every ordering $\sigma\in \Sigma$, we form an unweighted path graph $P_{\sigma}$ with vertex set $X$ and the order of vertices along the path is $\sigma$. We construct a $2$-hop $1$ spanner $H_\sigma$ for $P^{\sigma}_{n}$ using \Cref{thm:2hopPath}. 
For every point $y\in X$, let $E_{\sigma,y}$ be the collection of edges $y$ is responsible for.
We construct a label $\ell_y$ for a point $y\in X$ as follows: for every $\sigma\in\Sigma$ store the position of $y$ in the order $\sigma$, and all the edges in $E_{\sigma,y}$ together with their weight (i.e. for $\{x,y\}\in E_{\sigma,y}$ store $(\{x,y\},d_X(x,y))$). The label size is thus $O(\tau\cdot\log N)$ words.

Consider a subset $P\subseteq X$ of $n$ points.	To construct the labeled NNS, we will first store all the labels of $P$ points. Next, for every ordering $\sigma\in \Sigma$, we will construct a Y-fast trie (\Cref{thm:Y-fast-trie})
with the inputs being the positions in $\sigma$ of $P$ points. Each point in the Y-fast trie will be linked to the relevant label. 
This finishes the data structure. Note that the total space is $O(n\cdot \tau\cdot\log N)$. Moreover, we can add/remove a point by simply adding/removing its label, and updating the Y-fast tries of all the orderings it belongs to. Such an update will take $O(\tau\cdot\log\log N)$ amortized time (we assume that just coping the label itself takes $O(1)$ times as we can store a link. But to insert in to the Y-fast tries, takes $O(\log\log N)$ per ordering).

Given a query $q\in X$ with its label $\ell_q$, we go over all the orderings. For each $\sigma\in\Sigma$ we use the Y-fast trie to obtain predecessor $y_{\sigma,p}$ and successor $y_{\sigma,s}$ of $q$ w.r.t. $\sigma$ in $O(\log\log n)$ time. 
By knowing the position of $q$ and $y_{\sigma,s}$ in $\sigma$, we can compute in $O(1)$ time a point $\hat{y}_{\sigma,s}$ such that $q\prec_{\sigma}\hat{y}_{\sigma,s}\prec_{\sigma}y_{\sigma,s}$, and $\{q,\hat{y}_{\sigma,s}\}\in E_{\sigma,y}$ , $\{\hat{y}_{\sigma,s},y_{\sigma,s}\}\in E_{\sigma,y_{\sigma,s}}$ (note that $\hat{y}_{\sigma,s}$ is not necessarily in $P$). We estimate the distance from $q$ to $y_{\sigma,s}$ by $d_X(q,\hat{y}_{\sigma,s})+d_X(\hat{y}_{\sigma,s},y_{\sigma,s})$. We estimate the distance from $q$ to $y_{\sigma,p}$ in a symmetric fashion.
Overall we estimated the distance from $q$ to at most $2\tau$ points. Our answer will be the point with smallest estimated distance among those. Clearly the query time is $O(\tau\log\log N)$.

Finally, we argue that the returned point is a $2\rho$ approximation to the nearest neighbor. Let $x\in P$ be the closest point to $q$. 
Suppose w.l.o.g. that in the special ordering $\sigma'\in\Sigma$ it holds that $q\prec_{\sigma'}x$ (the other case is symmetric). For every pair of points $q\prec_{\sigma'}a\prec_{\sigma'}b\prec_{\sigma'}x$, it holds that $d_X(a,b)\le \rho\cdot d_X(q,x)$.
Suppose w.l.o.g. that we returned the point  $y_{\tilde{\sigma},s}$. It holds that 
\begin{align*}
	d_{X}(q,y_{\tilde{\sigma},s}) & \le d_{X}(q,\hat{y}_{\tilde{\sigma},s})+d_{X}(\hat{y}_{\tilde{\sigma},s},y_{\tilde{\sigma},s})\\
	& \le d_{X}(q,\hat{y}_{\sigma',s})+d_{X}(q,\hat{y}_{\sigma',s},y_{\sigma',s})\\
	& \le\rho\cdot d_{X}(q,x)+\rho\cdot d_{X}(q,x)=2\rho\cdot d_{X}(q,x)~.
\end{align*}
The theorem now follows.
\end{proof}

By plugging in the triangle-LSO's of \Cref{thm:LSOdoublingLargeStretch}, or that of \cite{FL22} (\Cref{tab:LSO}) into the \Cref{thm:triangleToNNS} we conclude:
\begin{corollary}\label{cor:doublingNNS}
Consider an $N$-point metric space $(X,d_X)$ of doubling dimension $d$.
Then for every $t\in[\Omega(1),d]$, $X$ admits a dynamic labeled $t$-NNS with label size $O(2^{O(\nicefrac{d}{t})}\cdot d\cdot\log^{2}t\cdot\log N)$, space $O(n\cdot 2^{O(\nicefrac{d}{t})}\cdot d\cdot\log^{2}t\cdot\log N)$,
query time $O(2^{O(\nicefrac{d}{t})}\cdot d\cdot\log^{2}t\cdot\log\log N)$, and deletion/insertion in 
$O(2^{O(\nicefrac{d}{t})}\cdot d\cdot\log^{2}t\cdot\log\log N)$ amortized time.
\end{corollary}
A case of special interest is when the stretch is $d$. Here we get label size  $\tilde{O}(d)\cdot\log N$, space $\tilde{O}(d)\cdot n\cdot\log N$, query time $\tilde{O}(d)\cdot\log\log N$, and deletion/insertion in 
$\tilde{O}(d)\cdot\log\log N$ amortized time.

\begin{corollary}\label{cor:GeneralNNSfromLSO}
Consider an $N$-point metric space $(X,d_X)$.
Then for every integer $k$, and $\eps\in(0,\frac12)$, $X$ admits a dynamic labeled $4k+\eps$-NNS with label size $O(N^{\frac{1}{k}}\cdot\log^{2}N\cdot\frac{k^{2}}{\eps}\cdot\log\frac{k}{\eps})$, space $O(n\cdot \N^{\frac{1}{k}}\cdot\log^{2}N\cdot\frac{k^{2}}{\eps}\cdot\log\frac{k}{\eps})$,
query time $O(N^{\frac{1}{k}}\cdot\log N\log\log N\cdot\frac{k^{2}}{\eps}\cdot\log\frac{k}{\eps})$, and deletion/insertion in 
$O(N^{\frac{1}{k}}\cdot\log N\log\log N\cdot\frac{k^{2}}{\eps}\cdot\log\frac{k}{\eps})$ amortized time.
\end{corollary}
A case of special interest is where the have stretch $\log N$, label size $\tilde{O}(\log^4 N)$, space $n\cdot\tilde{O}(\log^4 N)$, query time $\tilde{O}(\log^3 N)$, and deletion/insertion in 
$\tilde{O}(\log^3 N)$ amortized time.
We make significant improvements on this in the next subsection.

\subsection{General Metrics Spaces}\label{subsec:NNSGeneral}

In \Cref{cor:GeneralNNSfromLSO} we obtained a labeled NNS for general metric spaces. However, for the case of stretch $\approx k$, for $k\ll\log N$, the query time is rather large: $\tilde{O}(N^{\frac1k})$.
In this section we show that by slightly increasing the stretch to $8(1+\eps)k$, we can construct labeled NNS with extremely small query time: $O(\frac{1}{\eps}\cdot\log\log N)$ (\Cref{thm:NNSfromRamsey}).
We then focus on the case of stretch $\log N$. Here \Cref{cor:GeneralNNSfromLSO} provides label size $\tilde{O}(\log^4 N)$ and query time $\tilde{O}(\log^3 N)$.
We show a host of possible trade-offs for this stretch case, starting with label size $O(\log^2 N)$ and query time $O(\log\log N)$ (\Cref{thm:NNSfromRamsey}), and culminating with label size $O(1)$ and query time $O(\log\log N)$ (\Cref{cor:NNSfromClanAspectRatio}), under the assumption that the aspect ratio is polynomial, and the label size is bounded only in expectation. See \Cref{tab:GeneralNNS} for a summary of results.

Our approach here is a reduction of the general case to that of ultrametrics (in similar manner to the distance labeling of Mendel Naor \cite{MN07}), using either the classic Ramsey trees \cite{MN07}, or the novel clan embeddings \cite{FL21}. 
We also look into the improved guarantees that we can obtain using the following relaxations: (1) allowing the data structure to fail with some small probability, (2) giving an upper bound on the label size only in expectation (3) assuming that the metric space has polynomial aspect ratio. 
In \Cref{sec:NNSUltra} we solve the problem for the important special case of ultrametrics. 
In \Cref{sec:NNSGerneralRamsey,sec:NNSGerneralClan} we reduce from general space to ultrametrics using Ramsey trees, and clan embedding respectively. 
See \Cref{tab:GeneralNNS} for a summary of our results.

\begin{table}[t]
	\begin{tabular}{|l|l|l|l|l|l|}
		\hline
		Stretch & Label & Update & Query & Ref & Notes \\ \hline
		$8(1+\eps)k$ & $O(\frac{k}{\eps}N^{\frac{1}{k}}\cdot\log N)$       & $O(\frac{k}{\eps}N^{\frac1k}\cdot\log\log N)$        & $O(\frac1\eps\cdot\log\log N)$       & \theoremref{thm:NNSfromRamsey}    &       \\ \hline
		$16k$& $O(N^{\frac1k}\cdot\log N)$      & $O(N^{\frac1k}\cdot\log\log N)$        &  $O(\log\log N)$      & \theoremref{thm:NNSfromClan}    & $\mathbb{E}$      \\ \hline
		
		$8(1+\eps)k$ & $O(\frac{k}{\eps}N^{\frac{1}{k}}\cdot\log \frac1\eps)$       & $O(\frac{k}{\eps}N^{\frac1k}\cdot\log\log N)$        & $O(\frac1\eps\cdot\log\log N)$       & \cororef{cor:NNSfromRamseyAspectRatio}    & $\Phi$      \\ \hline
		$16(1+\eps)k$& $O(N^{\frac1k}\cdot\log \frac1\eps)$      & $O(N^{\frac1k}\cdot\log\log N)$        &  $O(\log\log N)$      & \cororef{cor:NNSfromClanAspectRatio}     & $\mathbb{E},\Phi$      \\ \hline
		
		$8(1+\eps)k$ & $O(\frac{k}{\eps}N^{\frac{1}{k}}\cdot\frac{\log\frac1\delta}{\eps^2})$       & $O(\frac{k}{\eps}N^{\frac1k}\cdot \log\log N)$        & $\frac1\eps\cdot O(\log\log N+\frac{\log\frac1\delta}{\eps^2})$       & \cororef{cor:NNSfromRamseyFailureProb}    & $\delta$      \\ \hline
		$16(1+\eps)k$& $O(N^{\frac1k}\cdot\frac{\log\frac1\delta}{\eps^2})$      & $O(N^{\frac1k}\cdot \log\log N)$        &  $O(\log\log N+\frac{\log\frac1\delta}{\eps^2})$      & \cororef{cor:NNSfromClanFailureProb}    & $\mathbb{E},\delta$      \\ \hline\hline
		
		$O(\log N)$ & $O(\log^2 N)$       & $\tilde{O}(\log N)$        & $O(\log\log N)$       & \theoremref{thm:NNSfromRamsey}    &       \\ \hline
		$O(\log N)$& $O(\log N)$      & $O(\log\log N)$        &  $O(\log\log N)$      & \theoremref{thm:NNSfromClan}    & $\mathbb{E}$      \\ \hline
		$O(\log N)$ & $O(\log N)$       & $\tilde{O}(\log N)$        & $O(\log\log N)$       &     \cororef{cor:NNSfromRamseyAspectRatio} & $\Phi$      \\ \hline
		$O(\log N)$ & $O(1)$       & $\tilde{O}(\log N)$        & $O(\log\log N)$       &     \cororef{cor:NNSfromClanAspectRatio} & $\mathbb{E},\Phi$       \\ \hline
		
		$O(\log N)$ & $O(\log N\cdot \log\frac1\delta)$       & $\tilde{O}(\log N)$        & $O(\log\log N+ \log\frac1\delta)$       & \cororef{cor:NNSfromRamseyFailureProb}    & $\delta$      \\ \hline
		$O(\log N)$& $O(\log \frac1\delta)$      & $O(\log\log N)$        &  $O(\log\log N+\log \frac1\delta)$      & \cororef{cor:NNSfromClanFailureProb}    & $\mathbb{E},\delta$      \\ \hline
	\end{tabular}
	\caption{Different labeled NNS data structures for general metrics and their parameters. 
		We also illustrate all the results for the regime of stretch $O(\log N)$.	
		Here $k\in\N$, and $\eps,\delta\in(0,1)$. The space in all the cases above equals $n$ times the label size.
		$\mathbb{E}$ denotes that the labels size, space, and update time guarantees are only in expectation (the query time is always worst case).
		$\Phi$ denotes that the result holds only for metric spaces with polynomial aspect ratio.
		$\delta$ denotes that the data structure can return a wrong answer with probability $\delta$. \label{tab:GeneralNNS}
	}
\end{table}
\begin{table}[t]
	\begin{tabular}{|l|l|l|l|l|l|}
		\hline
		Stretch  & Label                                &  Update                                                                                                                                                                        & Query                                                                                                                                                                         & Ref                            & Notes                        \\ \hline
		1        & $O(\log N)$                          &  $O(\log\log N)$                                                                                                                                                               & $O(\log\log N)$                                                                                                                                                               & \Cref{lem:UltraDeterministic} &                              \\ \hline
		$1+\eps$ & $O(\log \frac1\eps)$                 & $O(\log\log N)$                                                                                                                                                               & $O(\log\log N)$                                                                                                                                                               & \Cref{lem:UltraLabeling}      & Polynomial aspect ratio      \\ \hline
		$1+\eps$ & $O(\frac{\log\frac1\delta}{\eps^2})$                 &  $O(\log\log N)$
        & $O(\log\log N+\frac{\log\frac1\delta}{\eps^2})$                                                                                                                                                               & \Cref{lem:UltraJL}      & Failure probability $\delta$\\ \hline
	\end{tabular}
	\caption{Different labeled NNS data structures for ultrametrics and their parameters. The space in all the cases above equals $n$ times the label size.\label{tab:UltraNNS}}
\end{table}

\subsubsection{Ultrametric}\label{sec:NNSUltra}
A key component for the NNS of general metric space is the case of ultrametrics. 
We refer to \Cref{sec:LSOdoubling} for the introduction and definition of ultrametric.
Let $\pi$ be a pre order of the leaves in $U$. For a vertex $x$, let $\pi_x$ be the position (index) of $x$ in $\pi$. We will treat $\pi$ an an order $\preceq_\pi$.
Let $P\subseteq U$ be a subset of the points of size $n$. We construct a Y-fast-trie of the vertices w.r.t. the indices $\pi_x$. Note that the space required by the Y-fast-trie is $O(n)$, while in $O(\log\log N)$ time we can query predecessor/successor of a point, and add/remove a point.

Consider a query point $q\in U$, we use the Y-fast-trie (in $O(\log\log N)$ time) to find the predecessor $p$, and successor $s$ of $q$ w.r.t. $\pi$. 
For every vertex $x$ such that $q\preceq_{\pi} s\preceq_{\pi}  x$, let $y$ be the lca of $q$ and $x$, it holds that $s$ is also a descendant of $y$. Thus, $d_U(q,s)\leq d_U(q,x)$.
The same argument holds for $p$, and vertices $x$ such that $x\preceq_{\pi} p\preceq_{\pi}  q$. 
We conclude that $\min\{d_U(q,s),d_U(q,p)\}=\min_{x\in P}\{d_U(q,x)\}$.%
\footnote{Alternative formulation of this fact due to \cite{FL22} is that the preorder of an ultrametric is a $(1,1)$-triangle LSO.}

Until now we only used $O(n)$ space, and the labels are the indices $\pi_x$ (a single word).
To solve the labeled NNS problem, it remains to add some information to help us distinguish between $d_U(q,s)$ and $d_U(q,p)$.
We present three different solutions to the problem of estimating distances, under different assumptions. Each yielding a labeled NNS for ultrametrics. Our results are summarized in \Cref{tab:UltraNNS}.

\paragraph{Least common ancestor labeling - exact}
Our first approach is to retrieve the exact distances  $d_U(q,s),d_U(q,p)$.
Given a rooted tree, least common ancestor (LCA) labeling is a scheme where each vertex $v$ is assigned a label $\ell_v$, and there is an algorithm that given two labels $\ell_v,\ell_u$, can determine the LCA of $v$ and $u$.
We will be using the following scheme by Peleg \cite{Peleg05}. The query time is implicit using binary search.
\begin{theorem}[\cite{Peleg05}]\label{thm:PelegLCAlabels}
	Given an $N$ vertex tree $T$ where each vertex has unique identifier (that can be stored as machine word), there is an algorithm that assign $O(\log N)$ (words) label to each vertex, such that given two labels, in $O(\log\log N)$ time the algorithm return the unique identifier of their lca.
\end{theorem}

When there are no unique identifiers, and one only required to return the assigned label, much sparser schemes exist (see \cite{AGKR04,AHL14}).
We obtained a labeled-NNS for ultrametric.
\begin{lemma}\label{lem:UltraDeterministic}
	For every ultrametric $U$, one can construct exact labeled $1$-NNS with label size $O(\log N)$, space $O(n\cdot\log N)$, and query and insertion/deletion time $O(\log\log N)$.
	Furthermore, the data structure also outputs the exact distance to the closest neighbor.
\end{lemma} 
\begin{proof}
	Consider an ultrametric $U$, where each vertex $x$ has a label $\ell_{U}(x)$. 
	The unique identifier of an internal node $x$ will be $(x,\ell_{U}(x))$, while the unique identifier of a leaf vertex $x$ will be $(x,\pi_x)$.
	Next, we use \Cref{thm:PelegLCAlabels} to construct lca-labeling scheme $\{L_U(x)\}_{x\in U}$.
	Finally, in our labeled NNS scheme, the label of each vertex $x$ will be a union of the unique identifier, and the lca-label.
	
	Given a subset $P$, in addition to the Y-fast-trie above, we will also store all the labels.
	Given a query point $q\in U$, we will proceed as above. Given the predecessor $p$, and successor $s$, we will simply compute $d_U(q,s),d_U(q,p)$ (as they are the stored labels of the lca's), and return the minimum (together with its distance).

	The query time consist of querying the Y-fast-trie, and the lca data structures, taking $O(\log\log N)$ in total.
	The space is simply storing the labels and identifiers, which is $O(\log N)$ words per vertex and $O(n\cdot\log N)$ overall.
	One can add and remove vertices from $P$ by simply updating the Y-fast-trie, and the list of vertices in $P$, which is $O(\log\log N)$ per operation.
\end{proof}

\paragraph{Distance labeling in trees - $1+\eps$ approximation}\label{subsec:UltraDistanceLabeling}
Given a metric space $(X,d)$,
\emph{distance labeling} is an algorithm that 
assigns to each point $x\in X$ a label $\ell(x)$, such that there is an algorithm $\mathcal{A}$ (oblivious to $(X,d)$) that provided labels $\ell(x),\ell(y)$ of arbitrary $x,y \in X$, can compute $d_X(u,v)$.
Specifically, a distance labeling is said to have \emph{stretch} $t\ge 1$ if
\[
\forall x,y\in X,
\qquad
d(x,y)\le \mathcal{A}\left(l(x),l(y)\right) \le t\cdot d(x,y).
\]

We refer to \cite{FGK20} for an overview of distance labeling schemes in different regimes (and comparison with metric embedding, see also \cite{Peleg00Labling,GPPR04,TZ05,EFN18}). 
Freedman, Gawrychowski, Nicholson, and Weimann  \cite{FGNW17}, (improving upon \cite{AGHP16,GKKPP01}) showed that for any $n$-vertex unweighted tree, and $\eps\in(0,1)$, one can construct an $(1+\eps)$-labeling scheme with labels of $O(\log\frac1\eps)$ words.
\begin{theorem}[\cite{FGNW17}]\label{thm:tree-label}
	For any $n$-vertex unweighted tree $T=(V,E)$, and parameter $\eps\in(0,1)$, there is a distance labeling scheme with stretch $1+\eps$, $O(\log\frac1\eps)$ label size, and constant query time.
\end{theorem}
Given a tree metric $T$ of polynomial aspect ratio, by subdividing the edges one can show that there is an unweighted tree $T'$ with $\poly(\frac n\eps)$ vertices, and a mapping $f:T\rightarrow T'$ with a constant $c$ such that $\forall x,y\in T$, $d_T(x,y)\le c\cdot d_{T'}(f(x),f(y))\le (1+\eps)\cdot d_T(x,y)$.
It follows that (assuming $\eps\ge n^{-O(1)}$) \Cref{thm:tree-label} holds also for trees with polynomial aspect ratio.

\begin{lemma}\label{lem:UltraLabeling}
	For every ultrametric $U$ with polynomial aspect ratio and $N^{-O(1)}<\eps<1$, one can construct labeled $(1+\eps)$-NNS with label size $O(\log \frac1\eps)$, space $O(n\cdot\log \frac1\eps)$, and query and insertion/deletion time $O(\log\log N)$.
\end{lemma} 
\begin{proof}
	Consider an ultrametric $U$ with polynomial aspect ratio. We treat it as a tree and use \Cref{thm:tree-label} to obtain a $1+\eps$ distance labeling.
	In our labeled NNS scheme, the label of each vertex $x$ will contain (in addition to previous information) the distance label.
	
	Given a subset $P$, in addition to the Y-fast-trie above, we will also store all the labels.
	Given a query point $q\in U$, we will proceed as above. Given the predecessor $p$, and successor $s$, we will simply estimate $d_U(q,s),d_U(q,p)$ using the distance label, and return the minimum. Clearly, we obtained a $(1+\eps)$-approximate nearest neighbor.	
	
	The query time consist of querying the Y-fast-trie, and the distance labeling, taking $O(\log\log N)$ in total.
	The space usage consist of storing the labels and identifiers, which is $O(\log \frac{1}{\eps})$ words per vertex and $O(n\cdot\log \frac{1}{\eps})$ overall.
	One can add and remove vertices from $P$ by simply updating the Y-fast-trie, and the list of vertices in $P$, which is $O(\log\log N)$ per operation.
\end{proof}

\paragraph{Embedding into $\ell_2$ and dimension reduction - success with high probability}
Here we deal with the case of unbounded aspect ratio, to cope with this situation, we will allow for failure.
It is well known fact that every ultrametric embeds isometrically into $\ell_2$ (see e.g. \cite{fkmvww14}). Specifically, there is a function $f:U\rightarrow\R^d$ (for some $d>0$) such that for every $x,y\in U$, $\|f(x)-f(y)\|_2=d_U(x,y)$.
However, the dimension $d$ could be very large.
To reduce the dimension, we will use the Johnson Lindenstrauss dimension reduction \cite{JL84,DG03}. The JL transform is a random linear map of the $\R^d$ into  $\R^{O(\frac{\log\frac1\delta}{\eps^2})}$ such that for every $\vec{x},\vec{y}\in \R^d$ it holds that
\[
	\Pr[(1-\eps)\|\vec{x}-\vec{y}\|_2\le\|g(\vec{x})-g(\vec{y})\|_2\le (1+\eps)\|\vec{x}-\vec{y}\|_2]\ge 1-\delta~.
\]%
\begin{lemma}\label{lem:UltraJL}
	For every ultrametric $U$ and parameters $\eps,\delta\in[0,1]$, one can construct $(1+\eps)$-labeled-NNS with error probability $\delta$. Specifically, given a query $q$ the data structure will return a $(1+\eps)$-approximate near neighbor with probability $1-\delta$. 
	The label size is $O(\frac{\log\frac1\delta}{\eps^2})$, the space is $O(n\cdot\frac{\log\frac1\delta}{\eps^2})$, query time $O(\log\log N+\frac{\log\frac1\delta}{\eps^2})$,
	and insertion/deletion time $O(\log\log N)$. 
\end{lemma} 
\begin{proof}
	Consider an ultrametric $U$, let $f$ be the isometric embedding of $U$ into $\R^d$, and $g$ be a random JL transform into $\R^{O(\frac{\log\frac1\delta}{\eps^2})}$.
	In our labeled NNS scheme, the label of each vertex $x$ will contain (in addition to previous information) the point  $g(f(x))$.
	
	Given a subset $P$, in addition to the Y-fast-trie above, we will also store all the images.
	Given a query point $q\in U$, we will proceed as above. Given the predecessor $p$, and successor $s$, we will simply compute $\|g(f(q))-g(f(s))\|_2,\|g(f(q))-g(f(p))\|_2$, and return $s$ or $p$ according to the achieving the minimum.
	With probability $1-2\delta$, $(1-\eps)\|g(f(q)-g(f(s)\|_2\le d_U(q,s)\le (1+\eps)\|g(f(q)-g(f(s)\|_2$ and $(1-\eps)\|g(f(q)-g(f(p)\|_2\le d_U(q,p)\le (1+\eps)\|g(f(q)-g(f(p)\|_2$. If this is indeed the case, than clearly we will obtain a $1+O(\eps)$-approximate solution.

	The query time consist of querying the Y-fast-trie, and computing the Euclidean distance: $O(\log\log N+\frac{\log\frac1\delta}{\eps^2})$ time in total.
	The space is simply storing the images and identifiers, which is $O(\frac{\log\frac1\delta}{\eps^2})$ words per vertex and $O(n\cdot\frac{\log\frac1\delta}{\eps^2})$ overall.
	One can add and remove vertices from $P$ by simply updating the Y-fast-trie, and the list of vertices in $P$, which is $O(\log\log N)$ per operation (we count here storing the image in $\R^d$ as a single operation).
\end{proof}

\subsubsection{Deterministic labeled NNS via Ramsey trees}\label{sec:NNSGerneralRamsey}
The following theorem is implicit from the distance oracle construction in Abraham \etal \cite{ACEFN20}.%
\begin{theorem}[\cite{ACEFN20}]\label{thm:Ramsey}
	Given an $N$-point metric space $(X,d_X)$ and parameters $k\in\N$, $\eps\in(0,1)$, there is a deterministic algorithm that construct a collection $\{U_i\}_{i\in\mathcal{I}}$ of $|\mathcal{I}|=\frac k\eps\cdot N^{\frac1k}$ ultrametrics over the points of $X$, such that:
	\begin{enumerate}
		\item For every $x,y\in X$ and $i\in\mathcal{I}$, $d_X(x,y)\le d_{U_i}(x,y)$.
		\item For every $x\in X$, there is a subset $\home(x)\subseteq \mathcal{I}$ of $O(\frac1\eps)$ ultrametrics, such that for every $y\in X$, $\min_{i\in\home(x)}d_{U_i}(x,y)\le8(1+\eps)k\cdot d_X(x,y)$.
	\end{enumerate}
\end{theorem}

Using \Cref{thm:Ramsey}, we can construct an labeled NNS:
\begin{theorem}\label{thm:NNSfromRamsey}
	For general metric space $(X,d_{X})$, and parameters $k\in\N$ and $\eps\in(0,1)$, there is
	a labeled $8(1+\eps)k$-NNS  with $O(\frac{k}{\eps}N^{\frac{1}{k}}\cdot\log N)$
	label size, space $O(n\cdot\frac{k}{\eps}N^{\frac1k}\cdot\log N)$,
	insertion/deletion time $O(\frac{k}{\eps}N^{\frac1k}\cdot\log\log N)$, and
	query time $O(\frac1\eps\cdot\log\log N)$.
\end{theorem}
\begin{proof}
	Using \Cref{thm:Ramsey}, we obtain the collection of ultrametrics $\{U_i\}_{i\in\mathcal{I}}$. For every ultrametric $U_i$, we will use \Cref{lem:UltraDeterministic} to construct a labeled (exact)-NNS for $U_i$ with labels $\{\ell_i(x)\}_{x\in X}$. 
	In our labeled NNS, the label of a point $x\in X$ will consist of all the labels of the ultrametrics $\{\ell_i(x)\}_{i\in\mathcal{I}}$, together with the names of all the ultrametrics in $\home(x)$.
	Thus the label size will be $O(\frac{k}{\eps}N^{\frac1k}\cdot\log N)$.
	
	Given a subset $P$ of points, our data structure will simply be the union of the data structures of all the ultrametrics. Thus a total of $O(n\cdot\frac{k}{\eps}N^{\frac1k}\cdot\log N)$ space.
	In particular, the insertion/deletion time will be $O(\frac{k}{\eps}N^{\frac1k}\cdot\log\log N)$ (we count the storage of the lca label as a single operation).
	Given a query $q\in X$, we will go over all the labeled NNS data structures of all the ultrametrics in $\home(q)$, and retrieve the nearest neighbor in each one of them (together with it's distance to $q$). We will return the retrieved point minimizing the distance to $q$. The correctness is straightforward. The query consist of just querying $O(\frac1\eps)$ data structures, and thus $ O(\frac1\eps\cdot\log\log N)$ overall.
\end{proof}
By replacing \Cref{lem:UltraDeterministic} in the construction above with either \Cref{lem:UltraLabeling} or \Cref{lem:UltraJL}, we obtain the following:
\begin{corollary}\label{cor:NNSfromRamseyAspectRatio}
	For general metric space $(X,d_{X})$ with polynomial aspect ratio, and parameters $k\in\N$ and $\eps\in(0,1)$, there is
	a labeled $8(1+\eps)k$-NNS  with $O(\frac{k}{\eps}N^{\frac{1}{k}}\cdot\log \frac1\eps)$
	label size, space $O(n\cdot\frac{k}{\eps}N^{\frac1k}\cdot\log  \frac1\eps)$,
	insertion/deletion time $O(\frac{k}{\eps}N^{\frac1k}\cdot\log\log N)$, and
	query time $O(\frac1\eps\cdot\log\log N)$.
\end{corollary}

\begin{corollary}\label{cor:NNSfromRamseyFailureProb}
	For general metric space $(X,d_{X})$, and parameters $k\in\N$ and $\eps,\delta\in(0,1)$, there is a labeled $8(1+\eps)k$-NNS  with $O(\frac{k}{\eps}N^{\frac{1}{k}}\cdot \frac{\log\frac1\delta}{\eps^2})$
	label size, space $O(n\cdot\frac{k}{\eps}N^{\frac1k}\cdot \frac{\log\frac1\delta}{\eps^2})$,
	insertion/deletion time $O(\frac{k}{\eps}N^{\frac1k}\cdot\log\log N)$, and
	query time $O(\frac1\eps\cdot(\frac{\log\frac1\delta}{\eps^2}+\log\log N))$.
	Each query is answered correctly with probability $1-\delta$ (however there are dependencies).
\end{corollary}

\subsubsection{Labeled NNS via Clan Embeddings} \label{sec:NNSGerneralClan}
An interesting case for the labeled NNS problem is for stretch $k=O(\log N)$, as here, the label size becomes logarithmic.
Specifically, from \Cref{thm:NNSfromRamsey} one can obtain a labeled $O(\log N)$-NNS with $O(\log^2N)$ label size, space $O(n\cdot \log^2N)$, insertion/deletion time $O(\log N\cdot\log\log N)$, and query time $O(\log\log N)$.

If one allows for the space bounds to be only in expectation, we can do much better. We will use the recent clan embedding of Filtser and Le \cite{FL21}.
\begin{theorem}[\cite{FL21}]\label{thm:clan}
	Given an $N$-point metric space $(X,d_X)$ and parameter $k\in\N$, there is a distribution $\mathcal{D}$ over maps $f$ into ultrametrics $U$, where each point $x$ is mapped into a non-empty subset $f(x)$ of the points, where the subsets $\{f(x)\}_{x\in X}$ are disjoint. It holds that:
	\begin{enumerate}
		\item For every $x,y\in X$, and $(f,U)\in\supp(\mathcal{D})$, $d_X(x,y)\le\min_{x'\in f(x),y'\in f(y)}d_U(x',y')$.
		\item For every point $x\in X$, and $(f,U)\in\supp(\mathcal{D})$, there is a special vertex $\chi(x)\in f(x)$ called the chief of $x$.
		For every $y\in X$, it holds that $\min_{y'\in f(y)}d_U(\chi(x),y')\le16k\cdot d_X(x,y)$.
		\item For every $x\in X$, $\mathbb{E}_{(f,U)\sim{\cal D}}[|f(x)|]=O(N^{\frac1k})$.
	\end{enumerate}
\end{theorem}

Using \Cref{thm:Ramsey}, we can construct a labeled NNS:
\begin{theorem}\label{thm:NNSfromClan}
	For general metric space $(X,d_{X})$, and parameters $k\in\N$ and $\eps\in(0,1)$, there is
	a labeled $16k$-NNS  with $O(N^{\frac1k}\cdot\log N)$ label size, space $O(n\cdot N^{\frac1k}\cdot\log N)$,
	insertion/deletion time $O(N^{\frac1k}\cdot\log\log N)$, and
	query time $O(\log\log N)$.\\
	The bounds on label size, space, and insertion/deletion time are only in expectation.
\end{theorem}
\begin{proof}
We sample an ultrametric and a map $(f,U)$ using \Cref{thm:NNSfromClan}. Next, we use \Cref{lem:UltraDeterministic} to construct a labeled (exact)-NNS for $U$.
In our labeled NNS, the label of a point $x\in X$ will consist of the labels of all the points in $f(x)\subset U$, and the identity of the chief $\chi(x)\in f(x)$. Thus the label size is $O(N^{\frac1k}\cdot\log N)$ in expectation.

Given a subset $P$ of points, our data structure will simply be the data structure from \Cref{lem:UltraDeterministic}, for all the points $\cup_{x\in P}f(x)$. Thus a total of $O(n\cdot N^{\frac1k}\cdot\log N)$ space in expectation.
In particular, the insertion/deletion time will be $O(N^{\frac1k}\cdot\log\log N)$ time in expectation.
Given a query $q\in X$, we will retrieve the identity of the chief $\chi(q)$, and find the nearest neighbor $y'$ of $\chi(q)$ in $\cup_{x\in P}f(x)$. We will return the point $y\in P$ such that $y'\in f(y)$. The query time is $O(\log\log N)$ (in worst case).
To verify correctness, note it holds that
\begin{align*}
	d_{X}(q,y) & \le d_{U}(\chi(q),y')=\min_{z'\in\cup_{z\in P}f(z)}d_{U}(\chi(q),z')\le16k\cdot\min_{z\in P}d_{X}(q,z)~.
\end{align*}
\end{proof}
By replacing \Cref{lem:UltraDeterministic} in the construction above with either \Cref{lem:UltraLabeling} or \Cref{lem:UltraJL}, we obtain the following:
\begin{corollary}\label{cor:NNSfromClanAspectRatio}
	For general metric space $(X,d_{X})$ with polynomial aspect ratio, and parameters $k\in\N$ and $\eps\in(0,1)$, there is
	a labeled $16(1+\eps)k$-NNS  with $O(N^{\frac{1}{k}}\cdot\log \frac1\eps)$
	label size, space $O(n\cdot N^{\frac1k}\cdot\log  \frac1\eps)$,
	insertion/deletion time $O(N^{\frac1k}\cdot\log\log N)$, and
	query time $O(\log\log N)$.\\
	The bounds on label size, space, and insertion/deletion time are only in expectation.
\end{corollary}

\begin{corollary}\label{cor:NNSfromClanFailureProb}
	For general metric space $(X,d_{X})$, and parameters $k\in\N$ and $\eps,\delta\in(0,1)$, there is a labeled $16(1+\eps)k$-NNS  with $O(N^{\frac{1}{k}}\cdot \frac{\log\frac1\delta}{\eps^2})$
	label size, space $O(n\cdot N^{\frac1k}\cdot \frac{\log\frac1\delta}{\eps^2})$,
	insertion/deletion time $O(N^{\frac1k}\cdot\log\log N)$, and
	query time $O(\log\log N+\frac{\log\frac1\delta}{\eps^2})$.\\
	The bounds on label size, space, and insertion/deletion time are only in expectation.
\end{corollary}

\subsection{Lower Bound for General Metrics}\label{sec:NNSGeneralLB}

In this section we show that even though labeled NNS suppose to return only an approximate nearest neighbor, and not pairwise distances, the classic information theoretic lower bound due to the girth still holds. In particular, it follows that the label size-stretch tradeoff in \Cref{cor:GeneralNNSfromLSO} is tight up to the constant $4$, which could be improved to $2$ (see \Cref{rem:LBNNStight}).

The girth of an unweighted graph $G$ is the length of the shortest cycle in $G$. The Erd\H{o}s girth conjecture \cite{Erdos64} states that for any $g$ and $N$, there exist an $N$-vertex graph with girth $g$ and $\Omega(N^{1+\frac{2}{g-2}})$ edges. The conjecture is known to holds for $g=4,6,8,12$ (see \cite{Benson66,Wenger91}). However, the best known provable lower bound for general $k$ is due to Lazebnik \etal \cite{LUW95}, who constructed graph with girth $g$ and $\Omega(n^{1+\frac43\cdot\frac{1}{g-2}})$ edges.
The following lower bound is conditioned on the Erd\H{o}s girth conjecture, however it is straightforward that using \cite{LUW95}, one can obtain the unconditional lower bound of $\Omega(n^{\frac{2}{3}\cdot\frac1k})$, on the label size.
\begin{theorem}\label{thm:NNS-LB}
	Fix a stretch parameter $t<2k+1$. Assuming the Erd\H{o}s girth conjecture, for every labeling algorithm for the labeled $t$-NNS problem and $N$, there is an $N$-vertex unweighted graph $H$, for which the labeling algorithm assigns labels of total length $\Omega(n^{1+\frac1k})$ (in particular maximal label size is $\Omega(n^{\frac1k})$, the size count here is in bits and not words).
\end{theorem}
\begin{proof}
	Fix $g=2k+2$, and let $G=(V,E)$ be the unweighted graph from the Erd\H{o}s girth conjecture with  girth $g$ and $\Omega(N^{1+\frac{2}{g-2}})=\Omega(N^{1+\frac{1}{k}})$
	 edges.
	 We construct a maximal disjoint collection of length $2$ paths in $G$: $P_1,P_2,\dots,P_s$. Such collection can be constructed greedily. Note that $G\setminus\cup P_i$ contains only isolated edges and vertices. It follows that the collection contains $s=\Omega(N^{1+\frac{1}{k}})$ paths.
	 Next, we construct a subgraph $G'$ of $G$, by taking each path $P_i=\{x,y\}\circ\{y,z\}$, and keeping only one of the edges $\{x,y\}$ or $\{y,z\}$ (we keep all the edges out of $\cup_i P_i$).
	 Note that if $\{x,y\}\in G'$, then $d_{G'}(y,x)=1$ while $d_{G'}(y,z)\ge g-1=2k+1$ (as $d_{G'}(x,y)< g-1$ will imply that there is a cycle with strictly less than $g$ edges).
	 Let $\mathcal{G}$ be all the subgraphs of $G$ that can be constructed with such a process. Note that $|\mathcal{G}|=2^s=2^{\Omega(N^{1+\frac{1}{k}})}$.

	Consider two different graphs $G_1,G_2\in \mathcal{G}$.
	Let $\{\ell_v^1\}_{v\in V}$, $\{\ell_v^2\}_{v\in V}$ be $t$-NNS labels produced by the algorithm for $G_1$ and $G_2$, respectively.
	We argue that $\{\ell_v^1\}_{v\in V}\ne \{\ell_v^2\}_{v\in V}$.
	As $G_1\ne G_2$, there is some triangle $P_i=\{x,y\}\circ\{y,z\}$ such that $\{x,y\}\in G_1$, while $\{y,z\}\in G_2$.
	Consider a $t$-NNS data structure constructed for the set $P=\{x,z\}$, and  queried on $y$. 	
	Then in $G_1$, as $d_{G_1}(y,z)\ge2k+1>t=t\cdot d_{G_1}(y,x)$ the data structure must return $x$ as the approximate nearest neighbor. On the other hand, in $G_2$, 
	$d_{G_2}(y,x)\ge2k+1>t=t\cdot d_{G_2}(y,z)$, thus the the data structure must return $z$ as the approximate nearest neighbor.
	As the two data structures return different answers, they must be different. The data structures answers based only on the labels. We conclude that $\{\ell_v^1\}_{v\in V}\ne \{\ell_v^2\}_{v\in V}$. 
	It follows that  every graph in $\mathcal{G}$ has a unique (w.r.t. $\mathcal{G}$) labelings. As $\left|\mathcal{G}\right|=2^{\Omega(n^{1+\frac1k})}$, there must be some graph $G'$, who's sum of label lengths is $\Omega(n^{1+\frac1k})$.
	In particular, the maximal label size of a vertex in $G'$ is $\Omega(n^{\frac1k})$.	
\end{proof}

\begin{remark}[Tightness of \Cref{thm:NNS-LB}]\label{rem:LBNNStight}
	Note that \Cref{thm:NNS-LB} is only an information theoretic bound on the size of NNS label, and do not take into account the query and update times. Ignoring this parameters, one can obtain a matching upper bound (up to second order terms) by using Matou{\v{s}}ek embedding \cite{Mat96} into $\ell_\infty$. Matou{\v{s}}ek showed that for every $N$-point metric space there is an embedding $f$ of $X$  into $\tilde(n^{\frac1k})$-dimensional $\ell_\infty$ such that for every $x,y\in X$, $d_X(x,y)\le\|f(x)-f(y)\|_\infty\le(2k-1)\cdot d_X(x,y)$.
	Indeed take the NNS-label of $x$ to be the vector $f(x)$, and on query $q$, return the point $x\in P$ minimizing $\|f(x)-f(q)\|_\infty$.	
\end{remark}

\section{Metric Spanners}
\subsection{Path Reporting Low Hop Spanners}
The results of this section are summarized in \Cref{tab:LowHopSpanners}.
We begin \Cref{subsub:PathReportingFromLSO} with proving three meta Theorems: \ref{thm:MetaLSOPathReportingSpanner}, \ref{thm:MetaTriangleLSOPathReportingSpanner}, and \ref{thm:MetaRootedLSOPathReportingSpanner}. Each type of LSO implies a certain quality of a path reporting $2$-hop spanner.
As a result, we conclude a host of such spanners for different metric spaces (\Cref{cor:LowHopSpannerFromLSO}).
However, the resulting path reporting time depends linearly on the number of orderings and could be somewhat large.
We then tune manually the data structures to obtain significant improvements to the path reporting time. 
In \Cref{subsub:PathReportingEuclidean} we improve the query time of the Euclidean case to $O_d(1)$ by observing that we can quickly compute a satisfying LSO instead of checking them one by one, see \Cref{cor:LowHopSpannerEuclidean}.
In \Cref{subsub:PathReportingMinorFree} we improve the query time of minor free graphs to $O(\eps^{-1}\cdot\log n)$ by identifying a small number of LSO's which will contain the right answer, see \Cref{cor:LowHopMinor}.
Similarly, in \Cref{subsub:PathReportingTreewidth} we improve the query time of treewidth $k$ graphs to $O(k)$, see \Cref{cor:LowHopTreewidth}.
Most striking is the case of planar graphs (see \Cref{subsub:PathReportingPlanar}), where we improve the query time all the way to $O(\eps^{-1})$. This is done by using the structure of shortest path separators to restrict our search for the right answer even further, see \Cref{cor:LowHopPlanar}.
Finally, in \Cref{subsub:PathReportingGeneral} we deal with general metrics, 
where we provide solutions using different techniques than LSO.
First we observe that one can use Thorup Zwick \cite{TZ05} distance oracle  to construct path reporting $2$-hop spanner with optimal stretch, but somewhat large query time (compered to \cite{KLMS22}). We then use sparse covers to construct spanner with slightly worse stretch and exponentially better query time see \Cref{thm:TZ_GeneralPathReportingSpanner,thm:GeneralSmallQueryPathReportingSpanner}.

\subsubsection{Path reporting $2$-hop spanners from LSO's}\label{subsub:PathReportingFromLSO}
In this subsection we prove three meta theorems, constructing path reporting $2$-hop spanners given LSO.
\begin{theorem}[Meta theorem: LSO implies $2$ hop path reporting spanner]\label{thm:MetaLSOPathReportingSpanner}
	Suppose that the metric space $(X,d_X)$ admits a $(\tau,\rho)$-LSO, for $\rho\in(0,\frac12)$. Then $X$ admits a path reporting $2$-hop $1+2\rho$-spanner with $O(n\tau\log n)$ edges and $O(\tau)$ query time.
\end{theorem}
\begin{proof}
	Let $\Sigma$ be a $(\tau,\rho)$-LSO assumed by the theorem.
	For every ordering $\sigma\in \Sigma$, we form an unweighted path graph $P_{\sigma}$ with vertex set $X$ and the order of vertices along the path is $\sigma$. We construct a $2$-hop $1$ spanner $H_\sigma$ for $P^{\sigma}_{n}$ using \Cref{thm:2hopPath}. Note that for every edge $\{x,y\}\in H_{\sigma}$, $w(\{x,y\})=d_X(x,y)$.
	For every point $x\in \sigma$, we will store the set $E_{x,\sigma}$ as defined in \Cref{thm:2hopPath}.
	Our spanner will be $H=\cup_{\sigma\in\Sigma}H_\sigma$.
	Clearly the number of edges is $O(n\tau\log n)$. To argue stretch and $2$-hop, consider a pair of points $x,y$, and let $\sigma \in \Sigma$ be the ordering that satisfies the ordering property for $x$ and $y$. 
	There must be $z\in X$ such that $x\preceq_{\sigma}z\preceq_{\sigma}y$ and $\{x,z\}\in E_{x,\sigma}\subseteq H_\sigma$,$\{z,y\}\in E_{y,\sigma}\subseteq H_\sigma$. 
	Suppose w.l.o.g. that $d_X(x,z)\le\rho\cdot d_X(x,y)$. By triangle inequality, $d_X(y,z)\le(1+\rho)\cdot d_X(x,y)$.
	We conclude that 
	\[
	d_{H}(x,y)\le d_X(x,z) + d_X(z,y) = \rho\cdot d_X(x,y)+(1+\rho)\cdot d_X(x,y)=(1+2\rho)\cdot d_{X}(x,y)~.
	\]
	Finally, for the query time, in each of the orderings we have a unique $2$-hop path between $x$ and $y$, we can find each such path in $O(1)$ time and compute its weight (as guaranteed by \Cref{thm:2hopPath}). We return the path of minimum weight found in $O(\tau)$ -time.
\end{proof}

Using the exact same proof, we can also conclude the following meta theorem:
\begin{theorem}[Meta theorem: Triangle LSO implies $2$ hop path reporting spanner]\label{thm:MetaTriangleLSOPathReportingSpanner}
	Suppose that the metric space $(X,d_X)$ admits a $(\tau,\rho)$-triangle LSO. Then $X$ admits a path reporting $2$-hop $2\rho$-spanner with  $O(n\tau\log n)$ edges and $O(\tau)$ query time.
\end{theorem}
The only small difference in the proof is in the stretch analysis. Given $z\in X$ such that $x\preceq_{\sigma}z\preceq_{\sigma}y$ and $\{x,z\},\{z,y\}\in H_\sigma$, we have that 
 $d_X(x,z),d_X(y,z)\le\rho\cdot d_X(x,y)$, and hence by the triangle inequality $d_{H}(x,y)\le d_X(x,z) + d_X(z,y) = 2\rho\cdot d_X(x,y)$.

\begin{theorem}[Meta theorem: rooted LSO implies $2$ hop path reporting spanner]\label{thm:MetaRootedLSOPathReportingSpanner}
	Suppose that the metric space $(X,d_X)$ admits a $(\tau,\rho)$-rooted LSO. Then $X$ admits a path reporting $2$-hop $\rho$-spanner with $O(n\tau)$ edges and $O(\tau)$ query time.
\end{theorem}
\begin{proof}
	Let $\Sigma$ be a $(\tau,\rho)$-rooted-LSO as assumed by the theorem.
	For every ordering $\sigma\in \Sigma$, we simply store the first point in the order $x_\sigma$, and add an edge from each point in $y$ in $\sigma$ to $x_\sigma$. This finished the construction of the spanner $H$. 
	Clearly the number of edges is $O(n\tau)$. To argue stretch and $2$-hop, consider a pair of points $u,v$, and let $\sigma \in \Sigma$ be the ordering where both $u,v$ participating, and $d_X(u,x_\sigma)+d_X(x_\sigma,v)\le\rho\cdot d_X(u,v)$. As $\{u,x_\sigma\},\{x_\sigma,v\}$ both belong to the spanner, we found a $2$-hop path of stretch $\rho$.
	To answer a path query, we can exhaustively search all the orderings $v$ belongs to in $O(\tau)$ total time.
\end{proof}

\begin{remark}
	The path reporting data structure in the constructions of Theorems \ref{thm:MetaLSOPathReportingSpanner}, \ref{thm:MetaTriangleLSOPathReportingSpanner} and \ref{thm:MetaRootedLSOPathReportingSpanner} can all be made distributed. Specifically, where $S$ is the number of edges in the spanner, we can give each point $x\in X$ a label of size $O(\frac{S}{n})$, such that given the $2$ labels of points $x,y$, in the same query time we can find the promised $2$-hop path in the spanner.
\end{remark}

Using the various existing LSO's (summarized in \Cref{tab:LSO}), and the meta  Theorems \ref{thm:MetaLSOPathReportingSpanner}, \ref{thm:MetaTriangleLSOPathReportingSpanner} and \ref{thm:MetaRootedLSOPathReportingSpanner}, we conclude:
\begin{corollary}\label{cor:LowHopSpannerFromLSO}
	Consider an $n$ point metric space $(X,d_X)$. Then:
	\begin{enumerate}
		\item For every integer $k$ and $\eps\in(0,\frac12)$, $X$ admits a path reporting $2$-hop $4k+\eps$-spanner with $\tilde{O}(n^{1+\frac1k}\cdot\eps^{-1})$ edges and $\tilde{O}(n^{\frac1k}\cdot\eps^{-1})$ query time.
		\item If $X$ has doubling dimension $d$, then for $\eps\in(0,\frac12)$ $X$ admits a path reporting $2$-hop $(1+\eps)$-spanner with $\eps^{-O(d)}\cdot n\log n$ edges and $\eps^{-O(d)}$ query time.
		\item If $X$ has doubling dimension $d$, then for $t=\Omega(1)$, $X$ admits a path reporting $2$-hop $t$-spanner with $2^{-O(\nicefrac dt)}\cdot \tilde{O}(n)$ edges and $2^{-O(\nicefrac dt)}\cdot d\cdot \log^2 t$ query time.
		\item If $X$ is taken from $d$ dimensional Euclidean space, then for $\eps\in(0,\frac12)$, $X$ admits a path reporting $2$-hop $(1+\eps)$-spanner with $O_d(\eps^{-d})\log\frac1\eps\cdot n\log n$ edges and $O_d(\eps^{-d})\cdot\log\frac1\eps$ query time.
		\item If $X$ are taken from $d$ dimensional Euclidean space, then for $t\in [8,\sqrt{d}]$ and $\eps\in(0,\frac12)$,  $X$ admits a path reporting $2$-hop $(1+\eps)t$-spanner with $\tilde{O}(\frac{d^{1.5}}{\eps\cdot t})\cdot e^{\frac{2d}{t^{2}}\cdot(1+\frac{8}{t^{2}})}\cdot n\log n$ edges and $\tilde{O}(\frac{d^{1.5}}{\eps\cdot t})\cdot e^{\frac{2d}{t^{2}}\cdot(1+\frac{8}{t^{2}})}$ query time.
		
		\item If $X$ are taken from $d$ dimensional $\ell_p$ space for $p\in[1,2]$, then for $t\ge 10$,  $X$ admits a path reporting $2$-hop $t$-spanner with $\tilde{O}(d)\cdot e^{O(\frac{d}{t^{p}})}\cdot n\log n$ edges and $\tilde{O}(d)\cdot e^{O(\frac{d}{t^{p}})}$ query time.
		
		\item If $X$ are taken from $d$ dimensional $\ell_p$ space for $p\in[2,\infty]$, then $X$ admits a path reporting $2$-hop $2\cdot d^{1-\frac1p}$-spanner with $\tilde{O}(d)\cdot n\log n$ edges and $\tilde{O}(d)$ query time.
		
		\item If $(X,d_X)$ is the shortest path metric of a graph with treewidth $k$ (in particular a tree for $k=1$), then  $X$ admits a path reporting $2$-hop $1$-spanner with $O(nk\log n)$ edges and $O(k\log n)$ query time.
		\item If $(X,d_X)$ is the shortest path metric of a fixed minor free graph (in particular planar graph), then for $\eps\in(0,\frac12)$,  $X$ admits a path reporting $2$-hop $(1+\eps)$-spanner with $O(\frac{n}{\eps}\log^2 n)$ edges and $O(\eps^{-1}\log^2 n)$ query time.
	\end{enumerate}
\end{corollary}

\subsubsection{Euclidean Metrics}\label{subsub:PathReportingEuclidean}
Consider the Euclidean space and its $(O_d(\eps^{-d})\log\frac1\eps,\eps)$-LSO \cite{CHJ20}.
This LSO is defined for the entire space, regardless of the specific $n$ points at hand. In other words, it is ``hard coded''. 
As such, given a pair of points $x,y\in\R^d$, we can find the ordering $\sigma$ in the collection satisfying them in $O_d(1)$ time - independent of $\eps$ and $n$. As a result, in the path reporting $2$-hop spanner constructed by \Cref{thm:MetaLSOPathReportingSpanner}, we can efficiently return a path.

We provide here a very brief description of the LSO, and refer to \cite{CHJ20} for a precise description. We will describe here only the ordering of the box $[0,1)^d$ (the rest of the space is similar).
Chan \etal divide the box $[0,1)^d$ into $O(\eps)^{-d}$ boxes with side length $\eps$. Next they argue that there are  $O(\eps)^{-d}$ orderings on this smaller boxes, such that any two small boxes are consecutive it one of them. Assume by induction that we have an LSO with $O(\eps)^{-d}$ orderings for each of the small box (they are all actually symmetric). They simply implement these orderings in the global  $O(\eps)^{-d}$ orders of the boxes. As a result we obtain $O(\eps)^{-d}$ orderings. Finally, they take this entire construction, and make it with small $O_d(1)$ shifts of the form $(\frac{i}{O_d(1)},\frac{i}{O_d(1)}\dots,\frac{i}{O_d(1)})$.
Chan \etal showed that for every pair of points $x,y$, there is a shift such that $x,y$ contained together in a box of diameter $O_d(1)\cdot\|x-y\|_2$. Now, when we partition this box, both points will belong to small boxes of diameter  $O_d(\eps)\cdot\|x-y\|_2$. The ordering where these two boxes will be consecutive, will satisfy the pair $x,y$.

Our main observation about this construction, is first that the right shift could be found efficiently (a simple computation). Next, given the shift we can compute the small boxes containing $x,y$, and pre-compute the right ordering for them (there are only $O_d(\eps^{-2d})$ potential queries, we can store the solutions to all).
Once we know what is the satisfying ordering for $x,y$, we can use the data structure from \Cref{thm:MetaLSOPathReportingSpanner} to compute a $2$-hop path in $O(1)$ time. We conclude:
\begin{corollary}\label{cor:LowHopSpannerEuclidean}
	Consider a subset $X$ of $n$ point in $d$ dimensional Euclidean space, for $\eps\in(0,\frac12)$, $X$ admits a path reporting $2$-hop $(1+\eps)$-spanner with $m=O_d(\eps^{-d})\log\frac1\eps\cdot n\log n$ edges, 
	\\space $m+O_d(\eps^{-2d})\log\frac1\eps$ and $O_d(1)$ query time.
\end{corollary}

\subsubsection{Fixed Minor free graphs}\label{subsub:PathReportingMinorFree}
Our construction closely follows the construction of rooted LSO's from \cite{FL22}. For clarity, we will use the notion of shortest path decomposition (abbreviated \SPD) introduced by Abraham \etal \cite{AFGN22} (see also \cite{Fil20}).

\begin{definition}[\SPD depth]\label{def:SPD-K} A graph has an \SPD of depth $1$ if and only if it is a (weighted) path. A graph $G$ has an \SPD of depth $k \geq 2$ if there exists a \emph{shortest path} $P$, such that deleting $P$ from the graph $G$	results in a graph whose connected components all have \SPD of depth at most $k-1$. 
\end{definition}

It is known that (see~\cite{AFGN22}) that fixed minor-free graphs have \SPD of depth $k = O(\log n)$. However, the family of bounded \SPDdepth graphs is much richer and contains dense graphs with $K_r$ as a minor, for arbitrarily large $r$.
The rest of the section is devoted to constructing a path reporting $2$-hop $1+\eps$ spanner for graph with \SPDdepth $k$. The number of edges will be $O(n\cdot k\cdot\eps^{-1}\cdot\log n)$, and the query time $O(k\cdot\eps^{-1})$. As a corollary, we will obtain the path reporting $2$-hop $(1+\eps)$-spanner for minor free graphs.

Our construction will rely on the following lemma from \cite{Kle02} (see also \cite{Thorup04,BFN19Ramsey}).
\begin{lemma}\label{lem:landmarks}
	Consider a weighted graph $G=(V,E,w)$ with parameter $\eps\in(0,1)$, and let $P$ be a shortest path in $G$. Then one can find for each vertex $v\in V$ a set of vertices, called  \emph{landmarks},  $L_v$ on $P$ of size $|L_v|=O(\frac1\eps)$, such that for every $v,u\in V$, if the shortest $u,v$ path intersects $P$, then there exists $x\in L_v$ and $y\in L_u$ satisfying
	$d_G(v,x)+d_P(x,y)+d_G(y,u)\le(1+\epsilon)\cdot d_G(v,u)$.
	Furthermore, $x$ and $y$ are consecutive w.r.t. $P$. That is, there are no additional vertices from $L_u\cup L_v$ on $P$ between $x$ and $y$.
\end{lemma}

We will construct the spanner recursively. Let $P$ be a shortest path of $G$ such that every component of $G\setminus P$ has an \SPD of depth $k-1$.
Let $\{L_v\}_{v\in V}$ be the set of landmarks provided by \Cref{lem:landmarks} w.r.t. $P$. 
For every vertex $v\in V$, denote $l(v)=|L_v|$. 
We construct an auxiliary tree $T_P$ that initially contains the shortest path $P$ only. For every vertex $v\in V$, let $L_v=\{x_1,\dots,x_{l(v)}\}\subseteq P$. We will add the vertices $\{v_i\}_{i=1}^{l(v)}$, which are $l(v)$ copies of $v$, to $T_P$, and for every $i$ connect $v_i$ to $x_i\in L_v$ with an edge of weight $d_G(v,x_i)$.	Note that $T_P$ has  $\sum_{v\in V}l(v)=O(\frac{n}{\eps})$ vertices. The distances in $T_P$ dominates the distances in the original graph. That is for every $v_i,u_j$ copies of $v,u$ respectively, it holds that $d_{T_P}(v_i,u_j)\ge d_G(v,u)$. Additionally, by \Cref{lem:landmarks}, for every pair of vertices $u,v$, such that  the shortest path between them intersects $P$, it holds that $d_{T_P}(v_i,u_j)\le(1+\epsilon)\cdot d_G(v,u)$, for some pair of ``consecutive'' $v_i,u_j$.\footnote{Interesting to note, the construction of rooted LSO for graphs with \SPD $k$ in \cite{FL22} works by constructing these trees $T_P$ where each vertex belong to $k$ trees (with $O(\frac1\eps)$ copies per tree), and then constructing $(O(\log n),1)$-rooted LSO for each tree as a black box.} 
For the tree $T_P$ we will construct a path reporting $2$-hop $1$-spanner $H_P$ using \cite{KLMS22}:
\begin{theorem}[\cite{KLMS22}]\label{thm:LowHopSpannerTreeKLMS22}
	Consider an $n$ point tree metric $(X,d_X)$, then $X$ admits a path reporting $2$-hop $1$-spanner with $O(n\log n)$ edges and $O(1)$ query time.
\end{theorem}
The number of edges in $H_P$ is thus bounded by $O(\frac{n}{\eps}\cdot\log(\frac{n}{\eps}))=O(\frac{n}{\eps}\cdot\log n)$. \footnote{We can assume that $\eps>\frac{1}{n^2}$, as otherwise we can simply return the clique as a spanner (with $1$-hop, stretch $1$, and $O(1)$ query time).}
Note that in the actual spanner $H_P$ all the copies of each vertex will be identified. However, for the purpose of looking up the path, we will keep the distinctions in the data structure. 
 
Let $\{C_1,\ldots, C_s\}$ be the set of connected components of $G\setminus P$. For each $i$, let $G_i=G[C_i]$ be the graph induced by $C_i$, and let $H_i$ be a path reporting $2$-hop $1+\eps$ spanner for $G_i$ constructed recursively. Our final spanner will be $H_P$ union $\cup_iH_i$.
The total number of edges is bounded by  $O(n\cdot k\cdot\eps^{-1}\cdot\log n)$, as the depth of the recursion is $k$, and each vertex contributes $O(\eps^{-1}\cdot\log n)$ edges per level.

Next we describe our query algorithm.
The \SPD defines a hierarchical partition of the vertices.
This partition can be represented by a tree $\mathcal{T}$, where the leaves are real vertices, and each internal node in $\mathcal{T}$ represents a cluster.
Each such cluster $C$, represents some recursive step, and is associated with a shortest path $P(C)$ in $G[C]$.
Consider a query pair $u,v$. Starting from the root of $\mathcal{T}$, for every internal node $C$ containing both $u,v$ we will do the following:
consider the removed path $P(C)$, and let $L_v=\{x_1,\dots,x_{l(v)}\}\subseteq P(C)$ and $L_u=\{y_1,\dots,y_{l(v)}\}\subseteq P(C)$ be the respective landmarks of $v$ and $u$.
We go over the landmarks, and for every pair of consecutive landmarks $x_i,y_j$, we query the data structure of the spanner $H_{P(C)}$ to find a $2$-hop path between $v_i$ and $u_j$ (recall that in the tree $T_{P(C)}$, $v_i$ and $u_j$ are leaves attached to $x_i,y_j$ respectively).
Note that the returned path is associated with a $2$-hop path in the real spanner $H$ between $u$ and $v$ of weight $d_{T_{P(C)}}(v_i,u_j)=d_G(v,x_i)+d_G(x_i,y_j)+d_G(y_j,u)$.
We keep only the shortest observed $2$-hop path. Note that there are at most $O(\eps^{-1})$ consecutive pairs, and as each query of $H_{P(C)}$ takes $O(1)$, the total time spent on $C$ is only $O(\eps^{-1})$.
In total, we go over all clusters $C$ in $\mathcal{T}$ containing both $u,v$ and check $O(\eps^{-1})$ pairs for each such cluster. We return the shortest observed $2$-hop path. The total query time is $O(k\cdot\eps^{-1})$.

Finally, we argue that the returned $2$-hop path has weight at most $(1+\eps)\cdot d_G(u,v)$. It will be enough to show  that a path of such weight was considered.
Let $Q_{u,v}$ be the shortest $u$-$v$ path in $G$.
Let $C$ be the first cluster in $\mathcal{T}$ (i.e. closest to the root) where the removed path $P(C)$ and $Q_{u,v}$ intersect (note that $C$ is not necessarily the last cluster containing $u,v$, i.e. the lca of $u$ and $v$ in $\mathcal{T}$).
As $Q_{u,v}\subseteq C$, it follows that the distance between $u$ and $v$ in the induced graph $G[C]$ is $d_{G[C]}(u,v)\le w(Q_{u,v})=d_{G}(u,v)$.
By \Cref{lem:landmarks}, there are consecutive landmarks $x_i\in L_v,y_j\in L_u$ such that 
\begin{align*}
	d_{T_{P(C)}}(v_{i},u_{j}) & =d_{G[C]}(v,x_{i})+d_{G[C]}(x_{i},y_{j})+d_{G[C]}(y_{j},u)\\
	& \le(1+\epsilon)\cdot d_{G[C]}(v,u)\\
	& =(1+\epsilon)\cdot d_{G}(v,u)~.
\end{align*}
In particular, we queried the pair  $(v_{i},u_{j})$. As a result, one of the $2$-hop paths we checked had weight bounded by $(1+\epsilon)\cdot d_{G}(v,u)$. The theorem follows:
\begin{theorem}\label{thm:LowHopSpannerSPD}
	Consider an $n$ vertex graph $G=(V,E,w)$ with \SPDdepth $k$. Then for every $\eps\in(0,\frac12)$, the shortest path metric $(V,d_G)$ admits a path reporting $2$-hop $(1+\eps)$-spanner with $O(n\cdot k\cdot\eps^{-1}\cdot\log n)$ edges and $O(k\cdot\eps^{-1})$ query time.
\end{theorem}

We conclude:
\begin{corollary}\label{cor:LowHopMinor}
	Consider an $n$ point graph $G$ which is a fixed minor free, then for every $\eps\in(0,\frac12)$, the shortest path metric $(V,d_G)$ admits a path reporting $2$-hop $(1+\eps)$-spanner with $O(n\cdot\eps^{-1}\cdot\log^2 n)$ edges and $O(\eps^{-1}\cdot\log n)$ query time.
\end{corollary}

\subsubsection{Planar graphs}\label{subsub:PathReportingPlanar}
The spanner construction for planar graphs will be (almost) the exact same one as for fixed minor free case, using the shortest path decomposition. 
However, the \SPD in planar graphs has additional structural properties that can be exploited to obtain $O(\eps^{-1})$ query time (in a similar manner to the fact that the distance oracles for planar graph \cite{Kle02,Thorup04} have $O(\eps^{-1})$ query time, while the very similar minor free counter part distance oracles \cite{AG06} has $O(\eps^{-1}\log n)$ query time). 
We will go briefly over the details.

The \SPD of planar graphs comes from cycle separators \cite{LT79}. Consider an $n$-vertex planar graph $G$ with a specific drawing on the plane. Fix an arbitrary root vertex $r$, and let $T_r$ be the shortest path tree  rooted in $r$. One can find a pair of neighboring vertices $a,b$ in $G$ such that the cycle $C_{a,b}$ consisting of the edge $\{a,b\}$ and the two unique paths $P_a$,$P_b$ from $r$ to $a$ and $b$ in $T_r$ is a separator. Specifically the cycle $C_{a,b}$ partitions the plane into interior and exterior parts, where at most $\frac23n$ vertices lay in the interior and the exterior of the cycle.
The crux is that all the paths in the \SPD are taken from cycle separators created from a single shortest path tree $T_r$. In particular, \textbf{all the removed paths are shortest paths in the original graph, and not only in the induced graphs}. We refer to \cite{Kle02} for more details \footnote{See also \cite{CFKL20} for a detailed and careful exposition, alas in \cite{CFKL20} the planar graph contains a vortex.}.

There is a very small change in the construction of the spanner from \Cref{thm:LowHopSpannerSPD}: when considering a cluster $C$ and a removed path $P(C)$, we will create the landmarks (and define the weights of the auxiliary edges to the landmarks $\{v_i,x_i\}$ in $T_{P(C)}$) w.r.t. to distance in the original graph $G$ and not the induced graph $G[C]$. This is possible as $P$ is an actual shortest path in $G$.
As a result, in the query algorithm for the pair $u,v$, instead of checking all the clusters $C$ in $\mathcal{T}$ containing both $u,v$ (as we don't know which is the last cluster containing the entire shortest path $Q_{u,v}$ from $u$ to $v$), it is enough to check the last two paths constituting the cycle separator $C_{a,b}=P_a\cup P_b$ that separated $u$ from $v$. Indeed, the shortest path $Q_{u,v}$ from $u$ to $v$ in $G$ necessarily intersects either $P_a$ or $P_b$. 
Thus applying \Cref{lem:landmarks} on one of $P_a,P_b$ will be enough, regardless of the intersection of $Q_{u,v}$ with previously deleted paths.
We conclude that in the query algorithm it is enough to check at most two clusters (instead of $k=O(\log n)$). These two clusters can be found efficiently in $O(1)$ time in $\mathcal{T}$ using an lca query \cite{HT84}.
We conclude:
\begin{corollary}\label{cor:LowHopPlanar}
	Consider an $n$ point planar graph $G=(V,E,w)$, then for every $\eps\in(0,\frac12)$, then the shortest path metric $(V,d_G)$ admits a path reporting $2$-hop $(1+\eps)$-spanner with $O(n\cdot\eps^{-1}\cdot\log^2 n)$ edges and $O(\eps^{-1})$ query time.
\end{corollary}

\subsubsection{Treewidth graphs}\label{subsub:PathReportingTreewidth}
We begin with recalling the definition of treewidth:
\begin{definition}[Tree decomposition]
	A tree decomposition of $G(V,E)$, denoted by $\mathcal{T}$, is a tree satisfying the following conditions: 
	\begin{enumerate} [noitemsep,nolistsep]
		\item Each node $i \in V(\mathcal{T})$ corresponds to a subset of vertices $X_i$ of $V$  (called bags), such that $\cup_{i \in V(\mathcal{T})}X_i = V$.
		\item For each edge $uv \in E$, there is a bag $X_i$ containing both $u,v$.
		\item For a vertex $v \in V$, all the bags containing $v$ make up a subtree of $\mathcal{T}$.  
	\end{enumerate}
	The \emph{width} of a tree decomposition $\mathcal{T}$ is $\max_{ i \in V(\mathcal{T})}|X_i| -1$ and the treewidth of $G$, is the minimum width among all possible tree decompositions of $G$.
\end{definition}

Consider a graph $G=(V,E,w)$ with tree decomposition $\mathcal{T}$ of width $k$.
Similarly to the case of planar and minor free graphs, the rooted LSO for treewidth graph is also constructed using separators.
Initially the number of bags is at most $n$. There is a bag $X_1$ such that if we remove it, the number of bags in each connected component of $\mathcal{T}\setminus X_1$ will be at most $\frac n2$. Thus by the removal of at most $k+1$ vertices we obtain that each connected component has a tree decomposition of width $k$, and at most $\frac n2$ bags. Continue this process for $\log n$ steps and we deleted every vertex.
In particular, an $n$-vertex graph $G$ with treewidth $k$ has \SPDdepth $O(k\cdot\log n)$, where all the removed shortest path are simply single vertices.

The rooted LSO is constructed w.r.t. this ``\SPD''. That is, initially we delete a bag $X_1$ and have $|X_1|\le k+1$ orderings w.r.t. distance from each vertex $x\in X_1$ (in the spanner we respectively added edges from all the points to $X_1$).
Then we continue recursively with each connected component in the tree decomposition $\mathcal{T}\setminus X_1$.
An important detail is that when we delete a bag $X_i$ at some step, we create ordering for every vertex initially in $X_i$, even if it was deleted at some previous step.\footnote{If a vertex was deleted at a previous step it implies that we already created an ordering for it, and the new ordering is redundant. However, for the sake of simplicity we will create a new one.}
Note that we look on components of the tree decomposition and not actual graph. 
Note also that each vertex could correspond to only a single connected component of the remaining tree decomposition (if it belongs to two components, it necessary belong to a deleted bag, a thus not correspond to anything).
This process also defines a laminar partition (exactly like in the previous cases).
Given a pair of points $x,y$, consider the last cluster $C$ in this laminar partition containing both $x,y$ ($C$ can be found in $O(1)$ time in $\mathcal{T}$ using an lca query \cite{HT84}).
Denote by $X_C$ the bag deleted from $C$.
As $C$ is the last cluster containing both $x,y$, every path between $x$ and $y$ in $G$ goes through a vertex (originally) in  $X_C$. In particular $d_G(x,y)=\min_{z\in X_C}d_G(x,z)+d_G(z,y)$.
We conclude that in order to answer the distance query $x,y$, it is enough to check only the orderings created for the cluster $C$ w.r.t. the bag $X_C$. Thus we check at most $k+1$ orderings.
\begin{corollary}\label{cor:LowHopTreewidth}
	Consider an $n$ point graph $G=(V,E,w)$ with treewidth $k$, then the shortest path metric $(V,d_G)$ admits a path reporting $2$-hop $1$-spanner with $O(n\cdot k\cdot\log n)$ edges and $O(k)$ query time.
\end{corollary}

\subsubsection{General Metrics}\label{subsub:PathReportingGeneral}
A distance oracle is a data structure that given a distance query $x,y$, returns an estimate $\est(x,y)$ of the distance between $x$ and $y$. 
The distance oracle has stretch $t$ if for every pair $x,y$ it holds that $d_X(x,y)\le \est(x,y)\le t\cdot d_X(x,y)$.
Thorup and Zwick \cite{TZ05} in a celebrated result, for every integer $k$, constructed a distance oracle with stretch $2k-1$, space $O(n^{1+\frac1k}\cdot k)$ and query time $O(k)$.
Later, Chechik \cite{Chechik15} improved their result and constructed a distance oracle with stretch $2k-1$, space $O(n^{1+\frac1k})$, and query time $O(1)$.

The path reporting $2$-hop spanner for general graphs constructed in \Cref{cor:LowHopSpannerFromLSO} has sub-optimal stretch, and very large query time.
Our first observation is that \cite{TZ05} implicitly defines a $2$-hop (metric) spanner with $O(n^{1+\frac1k}\cdot k)$ edges. In particular, we can use \cite{TZ05} query algorithm to find such a path in $O(k)$ time.
Note that this spanner is optimal up to second order terms (assuming the Erd\H{o}s girth conjecture \cite{Erdos64}).
See \Cref{thm:TZ_GeneralPathReportingSpanner} for details.

In the interesting regime of stretch $O(\log n)$, the resulting $2$-hop spanner has query time $O(\log n)$. One might hope to reduce it to a smaller query time.
Indeed, Kahalon \etal \cite{KLMS22} got $O(1)$ query time using Ramsey trees \cite{BLMN05b,MN07,ACEFN20,Fil21}. Here one have a collection of trees $\mathcal{T}$, and each point $x$ has a home tree $T_x\in{\cal T}$ which approximates the shortest path tree of $x$ (see also \Cref{thm:Ramsey}).
In particular, to find a $2$-hop path from $x$ to $y$ it is enough to look for it in $T_x$.
While this method is indeed very efficient, the down side is that the stretch is inherently larger.
The constant in the stretch is important, as it governs the exponent in the number of edges.

We bridge between the two extremes, and improve the $O(k)$ query time resulting from \cite{TZ05} exponentially to $O(\eps^{-1}\cdot\log 2k)$, while inquiring additional factor of $2$ to the stretch (getting stretch $(4k-2)(1+\eps)$).
Our approach is to use Awerbuch, Peleg \cite{AP90} sparse cover (similar to padded partition cover \Cref{def:PaddedPartitionCover}), which is a collection of clusters of radius $(2k-1)\Delta$, such that every ball of radius $\Delta$ is contained in one of the clusters. To find the distance from $x$ to $y$, we will be looking for a minimal diameter cluster containing them both, but due to the padding property, the number of clusters we will need to check will be very small.

\begin{theorem}[\cite{TZ05} based $2$ hop path reporting spanner]\label{thm:TZ_GeneralPathReportingSpanner}
	Consider an $n$ point metric space and integer parameter $k\ge1$. Then $X$ admits a path reporting $2$-hop $(2k-1)$-spanner with $O(n^{1+\frac{1}{k}}\cdot k)$ edges and $O(k)$ query time.
\end{theorem}
\begin{proof}[Proof sketch.]	
	We begin with a brief overview of the distance oracle of Thorup and Zwick \cite{TZ05}.
	Recall that in the distance oracle construction of \cite{TZ05}, a sequence of sets $X=A_0\supseteq A_1\supseteq\dots\supseteq A_t=\emptyset$ is sampled randomly, by choosing each element of $A_{i-1}$ to be in $A_i$ with probability $n^{-\frac1k}$. 
	For each $v\in X$ and $0\le i\le k-1$, define the $i$-th pivot $p_i(v)$ as the nearest point to $v$ in $A_i$, and $B_i(v)=\{w\in A_i~:~ d(v,w) < d(v,A_{i+1})\}$.\footnote{We assume that $d(v,\emptyset)=\infty$ (this is needed as $A_k=\emptyset$).} The {\em bunch} of $v$ is defined as $B(v)=\bigcup_{0\le i\le k-1}B_i(v)$. The distance oracle stores in a hash table, for each $v\in X$, all the distances to points in $B(v)$, and also to the pivot $p_i(v)$ vertices.
	The expect size of $B(v)$ is $O(k\cdot n^{\frac1k})$, and thus  $O(k\cdot n^{1+\frac1k})$ overall.
	We will construct a spanner $H$ that will accompany the distance oracle: for every point $x$, add edges to all the points $y$ in its bunch, and to it's pivots. The overall number of edges in the spanner is $O(k\cdot n^{\frac1k})$. \footnote{Even though apriori the guarantee on the size is only in expectation, we can repeat the construction several times until we obtain the desired size.}
	
	\begin{algorithm}[t]
		\caption{\cite{TZ05} query: $\texttt{Dist}(v,u,i)$}\label{alg:TZquery}
		\DontPrintSemicolon
		
		$w\leftarrow v$\;
		\While{$w\notin B(u)$}{
			$i\leftarrow i+1$\;
			$(u,v)\leftarrow (v,u)$\;
			$w\leftarrow p_i(v)$\;
		}
		\Return $d_X(w,u)+d_X(w,v)$\;
	\end{algorithm}

	The query algorithm for the distance between $u,v$ in \cite{TZ05} is described in \Cref{alg:TZquery}, and takes $O(k)$ time.
	Crucially, the algorithm returns $d_X(w,u)+d_X(w,v)$, where $w$ is a pivot, or in the bunch of both $u$ and $v$.
	In particular, both edges $\{w,u\},\{w,v\}$ belong to our spanner $H$.
	It follows that we can run \Cref{alg:TZquery}, and return the path consisting of the two return edges. The $2k-1$ stretch is guaranteed by \cite{TZ05}.
	We refer to \cite{TZ05} for proofs and analysis, the theorem now follows.
\end{proof}

\begin{theorem}[Sparse cover based hop path reporting spanner]\label{thm:GeneralSmallQueryPathReportingSpanner}
	Consider an $n$ point metric space with aspect ratio $\Phi$. Then for every integer parameter $k\ge1$ and $\eps\in(0,\frac12)$, $X$ admits a path reporting $2$-hop $(1+\eps)(4k-2)$-spanner with $O(n^{1+\frac{1}{k}}\cdot\eps^{-1}\cdot k\cdot\log\Phi)$ edges with $O(\frac{\log2k}{\eps})$ query time.
\end{theorem}
\begin{proof}
	We will use the following classic sparse covers due to Awerbuch and Peleg \cite{AP90}.
	\begin{lemma}[Sparse cover, \cite{AP90}]\label{lem:SparseCoverAP90}
		Consider an $n$-point metric space $(X,d_X)$, and a parameter $\Delta>0$. Then $X$ admits a collection of clusters $\mathcal{C}$ of $X$ points such that:
		\begin{enumerate}
			\item Every point $x\in X$ belongs to at most $O(k\cdot n^{\frac1k})$ clusters in $\mathcal{C}$.
			\item For every cluster $C\in\mathcal{C}$, there is some point $x_C$, called the center of $C$, such that $C\subseteq B_X(x_c,(2k-1)\cdot\Delta)$. In particular, $C$ has diameter at most $(4k-2)\Delta$.
			\item For every $x$ there come cluster $C\in\mathcal{C}$ such that $B_X(x,\Delta)\subseteq C$.
		\end{enumerate}
	\end{lemma}
	Assume w.l.o.g. that the minimum distance in $X$ is $1$ (otherwise scale). In particular the maximum distance is $\Phi$. 
	For every integer $i\in[1,\left\lceil \log_{1+\eps}\Phi\right\rceil]$, we apply \Cref{lem:SparseCoverAP90} with parameter $\Delta_i=(1+\eps)^i$ to obtain a collection of clusters (sparse cover) $\mathcal{C}^{(i)}$.
	For every cluster $C\in \mathcal{C}^{(i)}$ let $x_C$ be the cluster center.
	We create a spanner $H$ by adding an edge from every point $y$ to the cluster centers $x_C$, for all the clusters $y$ belongs to in all the distance scales. 
	Note that we added a total of $O(n\cdot\log_{1+\eps}\Phi\cdot n^{\frac{1}{k}}\cdot k)=O(n^{1+\frac{1}{k}}\cdot\eps^{-1}\cdot k\cdot\log\Phi)$
	edges.
	In our data structure we will simply store the spanner, all the created clusters with their centers.
	In addition for every point $x$ and scale $i$, we will store the name of the cluster $C$ of the $i$'th scale where $B_X(x,\Delta_i)\subseteq C$.
	Finally, we will also create a $2k-1$-distance oracle with $O(1)$ query time and $O(n^{1+\frac1k})$ space \cite{Chechik15}.
	The total space is $O(n^{1+\frac{1}{k}}\cdot\eps^{-1}\cdot k\cdot\log\Phi)$.
	
	Given a query $u,v$ we first use the distance oracle to get an estimation ${\rm est}(u,v)$ of $d_X(u,v)$. Next, we go over the indices $\left\lfloor \log_{1+\epsilon}\frac{{\rm est}(u,v)}{2k-1}\right\rfloor ,\dots,\left\lceil \log_{1+\epsilon}{\rm est}(u,v)\right\rceil+1$.
	For every index $i$, let $C^{(i)}\in{\cal C}^{(i)}$ be the cluster containing $B_X(u,(1+\eps)^i)$. If $C^{(i)}$ also contains $v$, then  $C^{(i)}$ will provide us with the $2$-hop path: $(u,x_{C^{(i)}},v)$. We will return the $2$-hop path of minimum weight (among the observed paths). Clearly, the query time is $O(\log_{1+\eps}2k)=O(\frac{\log2k}{\eps})$.
	For correctness of estimate, let $i$ be the unique index such that $(1+\eps)^{i-1}\le d_X(u,v)< (1+\eps)^{i}$. Then 
	$(1+\eps)^{i-1}\le {\rm est}(u,v)\le (2k-1)\cdot (1+\eps)^{i}$. In particular, $\left\lfloor \log_{1+\epsilon}\frac{{\rm est}(u,v)}{2k-1}\right\rfloor\le i\le\left\lceil \log_{1+\epsilon}{\rm est}(u,v)\right\rceil+1$, thus $i$ is one of the indices we checked. 
	As $v\in B_X(u,(1+\eps)^i)\subseteq C^{(i)}$, one of the $2$-hop paths we checked been of weight $d_X(u,x_{C^{(i)}})+d_X(v,x_{C^{(i)}})\le 2\cdot (2k-1)\cdot (1+\eps)^{i}\le(1+\eps)(4k-2)\cdot d_X(u,v)$.
	The theorem now follows.	
\end{proof}

\subsection{Fault tolerant spanners}

In this section we present an extremely simple construction of fault tolerant spanners based on LSO's, which obtains the optimal linear dependence on the fault parameter $f$.
We prove meta theorem Theorems \ref{thm:MetaLSOFaultTolerantSpanner}, \ref{thm:MetaTriangleLSOFaultTolerantSpanner}, \ref{thm:MetaTriangleLSOFaultTolerantSpanner}, from which we derive all our results in \Cref{cor:FaultTolerantSpannerFromLSO}. See \Cref{tab:FaultTolerantSpanners} for a summary.

We begin with an $f$-fault tolerant $1$-spanner construction for the path graph. The construction is very similar to the classic \Cref{thm:2hopPath}.

\begin{theorem}\label{thm:2hopfaultTolerantPath}
	The path graph $P_n=(v_1,v_2,\dots,v_n)$ admits a $2$-hop $f$-fault tolerant $1$-spanner $H$ with $O(n\cdot f\cdot \log n)$ edges. 
\end{theorem}
\begin{proof}
	For simplicity we will assume that $n=2^\delta$ for some integer $\delta$. Later we can simply construct the spanner for a power of $2$ in $[n,2n)$, and remove the redundant edges. For simplicity we will also assume that $f$ is even (otherwise construct an $f+1$ fault tolerant spanner).
	Assume by induction that for every $\delta'<\delta$, a $2$-hop $f$-fault tolerant $1$-spanner with $(2^{\delta'}\cdot\delta'+1)\cdot (f+1)$ edges
	exist for the path graph $P_{n'}=P_{2^{\delta'}}$.
	The base case with $2$ vertices and $\delta'=1$ clearly holds.
	Next consider $n=2^\delta$. Add edges from all the vertices to the set of $f+1$ vertices $v_{2^{\delta-1}-\frac{f}{2}},v_{2^{\delta-1}-\frac{f}{2}+1},\dots,v_{2^{\delta-1}},\dots,v_{2^{\delta-1}+\frac{f}{2}}$.
	Denote this edges by $H_0$.
	Next, we use the induction hypothesis on the graph induced by $(v_1,v_2,\dots,v_{2^{\delta-1}})$, and the graph induced by $(v_{2^{\delta-1}+1},\dots,v_{2^{\delta}})$. Denote the returned spanners by $H_1$ and $H_2$ respectively.
	 By the induction hypothesis, the total number of edges is bounded by 
	\begin{align*}
		(2^{\delta}-1)\cdot(f+1)+2\cdot\left(2^{\delta-1}(\delta-1)+1\right)\cdot(f+1) & =\left(2^{\delta}-1+2^{\delta}(\delta-1)+2\right)\cdot(f+1)\\
		& =\left(2^{\delta}\cdot\delta+1\right)\cdot(f+1)
	\end{align*}
	Consider a faulty set $F$ of at most $f$ vertices, and let $v_i,v_j\notin F$ where $i<j$.
	If $i\le 2^{\delta-1}<j$, there there is some index $l\in\left\{ 2^{\delta-1}-\frac{f}{2},\dots,2^{\delta-1},\dots,2^{\delta-1}+\frac{f}{2}\right\} $
	such that $i\le l\le j$ and $v_l\notin F$.
	In particular, $\{v_i,v_l\},\{v_l,v_j\}$ are edges in $H$, and thus $H\setminus F$ contains a $2$ hop shortest path  between $v_i$ and $v_j$.
	Otherwise, if either $j\le2^{\delta-1}$ or $i>2^{\delta-1}$, then by the induction hypotesis one of the spanners $H_1\setminus F$ or $H_2\setminus F$ contains a shortest path between $v_i$ and $v_j$.
\end{proof}

We now turn to our meta theorems.
\begin{theorem}[Meta theorem: LSO to $2$ hop fault tolerant spanner]\label{thm:MetaLSOFaultTolerantSpanner}
	Suppose that a metric space $(X,d_X)$ admits a $(\tau,\rho)$-LSO for $\rho\in(0,\frac12)$. Then for every integer $f\ge1$, $X$ admits a $2$-hop $f$-fault tolerant $1+2\rho$-spanner with $O(n\cdot f\cdot \tau\cdot \log n)$ edges.
\end{theorem}
\begin{proof}
	Let $\Sigma$ be the $(\tau,\rho)$-LSO as assumed by the theorem.
	For every ordering $\sigma\in \Sigma$, we form an unweighted path graph $P_{\sigma}$ with vertex set $X$ and the order of vertices along the path is $\sigma$. We construct a $2$-hop $f$-fault tolerant $1$-spanner $H_\sigma$ for $P^{\sigma}_{n}$ using \Cref{thm:2hopfaultTolerantPath}. 
	Our spanner will be $H=\cup_{\sigma\in\Sigma}H_\sigma$ (where each edge will get the weight equal to the actual metric distance). 
	Clearly the number of edges is $O(n\cdot f\cdot \tau\cdot \log n)$. To argue stretch and $2$-hop, consider a set $F$ of at most $f$ faulty points, and let $x,y\notin F$. Let $\sigma \in \Sigma$ be the ordering that satisfies the LSO property for $x$ and $y$. 
	There must be $z\notin F$ such that $x\preceq_{\sigma}z\preceq_{\sigma}y$ and $\{x,z\},\{z,y\}\in H_\sigma$. 
	Suppose w.l.o.g. that $d_X(x,z)\le\rho\cdot d_X(x,y)$. By triangle inequality, $d_X(y,z)\le(1+\rho)\cdot d_X(x,y)$.
	We conclude that 
	\[
	d_{H\setminus F}(x,y)\le d_X(x,z) + d_X(z,y) = \rho\cdot d_X(x,y)+(1+\rho)\cdot d_X(x,y)=(1+2\rho)\cdot d_{X}(x,y)~.
	\]
\end{proof}

Using the exact same proof, we can also conclude the following meta theorem:
\begin{theorem}[Meta theorem: Triangle LSO  to $2$ hop fault tolerant spanner]\label{thm:MetaTriangleLSOFaultTolerantSpanner}
	Suppose that a metric space $(X,d_X)$ admits a $(\tau,\rho)$-triangle LSO for $\rho\ge 1$. Then for every integer $f\ge1$, $X$ admits a $2$-hop $f$-fault tolerant $2\rho$-spanner with $O(n\cdot f\cdot \tau\cdot \log n)$ edges.
\end{theorem}
The only small difference in the proof is in the stretch analysis. Given $z\in X$ such that $x\preceq_{\sigma}z\preceq_{\sigma}y$ and $\{x,z\},\{z,y\}\in H_\sigma$, we have that $d_X(x,z),d_X(y,z)\le\rho\cdot d_X(x,y)$, and hence $d_{H}(x,y)\le d_X(x,z) + d_X(z,y) = 2\rho\cdot d_X(x,y)$.

Next we deal with rooted LSO's. 
Consider the simple star graph (which actually has $(1,1)$-rooted LSO). Then every $1$-fault tolerant spanner with less ${n \choose 2}$ edges will have stretch $\ge 2$. Indeed, if the edge $\{x,y\}$ is missing, and we deleted the middle vertex, then every remaining path will have stretch at least $2$.
This is the reason that in the meta theorem bellow, the stretch parameter is $2\rho$.%
\begin{theorem}[Meta theorem: rooted LSO  to $2$ hop fault tolerant spanner]\label{thm:MetaRootedLSOFaultTolerantSpanner}
	Suppose that a metric space $(X,d_X)$ admits a $(\tau,\rho)$-rooted LSO for $\rho\ge 1$. Then for every integer $f\ge1$, $X$ admits a $2$-hop $f$-fault tolerant $2\rho$-spanner with $O(n\cdot f\cdot \tau)$ edges.
\end{theorem}
\begin{proof}
	Let $\Sigma$ be a $(\tau,\rho)$-rooted-LSO as assumed by the theorem.
	For every ordering $\sigma\in \Sigma$, we construct a spanner $H_\sigma$ by adding edges from the first $f+1$ points in $\sigma$, to all the other points. The final spanner $H$ is simply the union $H=\cup_{\sigma\in\Sigma}H_\sigma$.
	Note that as each vertex participates in at most $\tau$ orderings, the number of edges is bounded by $O(n\cdot f\cdot \tau)$. 
	
	Consider a set $F$ of at most $f$ faulty points, and let $x,y\notin F$. Let $\sigma \in \Sigma$ be the ordering that satisfies the ordering property for $x$ and $y$. 
	Let $z\notin F$ be the first not deleted point among the first $f+1$  points in the order $\sigma$ (note that only $f$ points belong to $F$). 
	Recall that for the first point $x_\sigma$ in $\sigma$ it holds that  $d_X(u,x_\sigma)+d_X(x_\sigma,v)\le\rho\cdot d_X(u,v)$. 
	Furthermore, the points are ordered w.r.t. distance from $x_\sigma$ and thus $d_X(z,x_\sigma)\le \min\{d_X(u,x_\sigma),d_X(x_\sigma,v)\}$. As $\{u,z\},\{z,v\}$ both belong to the spanner, we conclude
	\begin{align*}
		d_{H\setminus F}(x,y) & \le d_{X}(x,z)+d_{X}(z,y)\\
		& \le d_{X}(x,x_{\sigma})+d_{X}(x_{\sigma},z)+d_{X}(z,x_{\sigma})+d_{X}(x_{\sigma},y)\\
		& \le2\cdot\left(d_{X}(x,x_{\sigma})+d_{X}(x_{\sigma},y)\right)\le2\rho\cdot d_{X}(x,y)~.
	\end{align*}
\end{proof}

Using the various existing LSO's (summarized in \Cref{tab:LSO}), and the meta  Theorems \ref{thm:MetaLSOFaultTolerantSpanner}, \ref{thm:MetaTriangleLSOFaultTolerantSpanner} and \ref{thm:MetaRootedLSOFaultTolerantSpanner}, we conclude:
\begin{corollary}\label{cor:FaultTolerantSpannerFromLSO}
	Consider an $n$ point metric space $(X,d_X)$. Then:
	\begin{enumerate}
		\item If $X$ has doubling dimension $d$, then for $\eps\in(0,\frac12)$ $X$ admits a $2$-hop $f$-fault tolerant $(1+\eps)$-spanner with $\eps^{-O(d)}\cdot f\cdot n\cdot\log n$ edges.
		\item If $X$ has doubling dimension $d$, then for $t=\Omega(1)$, $X$ admits a $2$-hop  $f$-fault tolerant $t$-spanner with $2^{-O(\nicefrac dt)}\cdot f\cdot \tilde{O}(n)$ edges.
		\item For every integer $k$ and $\eps\in(0,\frac12)$, $X$ admits a $2$-hop  $f$-fault tolerant $4k+\eps$-spanner with $\tilde{O}(n^{1+\frac1k}\cdot f\cdot\eps^{-1})$ edges.				
		\item If $X$ is taken from $d$ dimensional Euclidean space, then for $\eps\in(0,\frac12)$, $X$ admits a $2$-hop  $f$-fault tolerant $(1+\eps)$-spanner with $O_d(\eps^{-d})\cdot \log\frac1\eps\cdot f\cdot n\log n$ edges.
		
		\item If $X$ are taken from $d$ dimensional Euclidean space, then for $t\in[8,\sqrt{d}]$ and $\eps\in(0,\frac12)$,  $X$ admits a $2$-hop $f$-fault tolerant $(1+\eps)t$-spanner with $e^{\frac{2d}{t^{2}}\cdot(1+\frac{8}{t^{2}})}\cdot\tilde{O}(\frac{d^{1.5}}{\eps\cdot t})\cdot f \cdot n\cdot \log n$ edges.
		
		\item If $X$ are taken from $d$ dimensional $\ell_p$ space for $p\in[1,2]$, then for $t\in[10,d^{\frac1p}]$, $X$ admits a $2$-hop $f$-fault tolerant $t$-spanner with $e^{O(\frac{d}{t^{p}})}\cdot\tilde{O}(d)\cdot f \cdot n\cdot \log n$ edges.
		
		\item If $X$ are taken from $d$ dimensional $\ell_p$ space for $p\in[2,\infty]$, then  $X$ admits a $2$-hop $f$-fault tolerant $2\cdot d^{1-\frac1p}$-spanner with $\tilde{O}(d)\cdot f \cdot n\cdot \log n$ edges.
		
		\item If $(X,d_X)$ is the shortest path metric of a graph with treewidth $k$ (in particular a tree for $k=1$), then  $X$ admits a $2$-hop $f$-fault tolerant $2$-spanner with $O(n\cdot k \cdot f\cdot\log n)$ edges.
		\item If $(X,d_X)$ is the shortest path metric of a fixed minor free graph (in particular planar graph), then for $\eps\in(0,\frac12)$,  $X$ admits a $2$-hop $f$-fault tolerant $2+\eps$-spanner with $O(\frac{n}{\eps}\cdot f\cdot \log^2 n)$ edges.
	\end{enumerate}
\end{corollary}

\aidea{If time permits we can look into it, I believe that we can as well get $\log n$ hop diameter with reasonable max degree using the expander approach of \cite{BHO19}. See some comments bellow (though not super relevant to the todo).}

\subsection{Reliable Spanners}\label{subsec:ReliableSpanner}
We begin with a formal definition of reliable spanner.
\begin{definition}[Reliable spanner]\label{def:reliableSpanner}
	A weighted graph $H$ over point set $X$ is a deterministic $\nu$-reliable $t$-spanner
	of a metric space $(X,d_{X})$ if $d_{H}$ dominates 
	\footnote{Metric space $(X,d_H)$ dominates metric space $(X,d_X)$ if  $\forall u,v\in X$, $d_X(u,v)\le d_H(u,v)$.\label{foot:dominating}} 
	$d_{X}$, and for every
	set $B\subseteq X$ of points, called an \emph{attack set}, there is a set $B^{+}\supseteq B$, called a \emph{faulty extension} of $B$, 
	such that:
	\begin{enumerate}
		\item $|B^{+}|\le(1+\nu)|B|$.
		\item For every $x,y\notin B^{+}$, $d_{H[X\setminus B]}(x,y)\le t\cdot d_{X}(x,y)$.
	\end{enumerate}	
	An oblivious $\nu$-reliable $t$-spanner is a distribution $\mathcal{D}$ over dominating graphs $H$, such that for every attack set $B\subseteq X$ and $H\in\supp(\mathcal{D})$,
	there exist a superset $B^{+}$ of $B$ such that, for
	every $x,y\notin B^{+}$, $d_{H[X\setminus B]}(x,y)\le t\cdot d_{X}(x,y)$,
	and  $\mathbb{E}_{H\sim\mathcal{D}}\left[|B^{+}|\right]\le(1+\nu)|B|$. We say that the oblivious spanner $\mathcal{D}$ has $m$ edges if every graph $H\in\supp(\mathcal{D})$ has at most $m$ edges.
\end{definition}

The distribution $\mathcal{D}$ in \Cref{def:reliableSpanner} is called an oblivious $\nu$-reliable $t$-spanner because the adversary is oblivious to the specific spanner produced by the distribution (it may be aware to the distribution itself).

Buchin \etal \cite{BHO19,BHO20} in two consecutive papers (implicitly) proved that given a $(\tau,\eps)$-LSO, one can construct a deterministic/oblivious sparse $\nu$-reliable $(1+O(\eps))$-spanner.
Later, Filtser and Hung \cite{FL22} proved additional meta theorems creating oblivious reliable spanners from triangle (and left sided) LSO's. We will be using the following:
\begin{theorem}[\cite{FL22} Meta theorem: triangle LSO implies reliable spanners]\label{Thm:TriangleLSOtoSpanner} Suppose that a metric space $(X,d_{X})$ admits a $(\tau,\rho)$-Triangle-LSO. Then for every $\nu\in(0,1)$, $X$ admits an \emph{oblivious} $\nu$-reliable, $(2\rho)$-spanner with $n\tau\cdot O\left(\log^{2}n+\nu^{-1}\tau\log n\cdot\log\log n\right)$ edges.
\end{theorem}
Plugging our newly constructed LSO's (see \Cref{tab:LSO}), we conclude:
\begin{corollary}\label{cor:RelaibleSpanner}
	Consider an $n$ point metric space $(X,d_X)$. Then:
	\begin{enumerate}
		\item If $X$ has doubling dimension $d$, then for parameters  $t\in[\Omega(1),d]$, $\nu\in(0,1)$, $X$ admits an oblivious $\nu$-reliable, $2t$-spanner with $\tilde{O}(n)\cdot\nu^{-1}\cdot2^{O(\nicefrac{d}{t})}$ edges.
		
		\item  If $X$ are taken from $d$ dimensional Euclidean space, then for parameters $t\in[8,\sqrt{d}]$, $\delta\in(0,\frac12)$ and $\nu\in(0,1)$,  $X$ admits a an oblivious $\nu$-reliable, $(1+\delta)\cdot t$-spanner with\\ $\nu^{-1}\cdot e^{\frac{4d}{t^{2}}\cdot(1+\frac{8}{t^{2}})}\cdot\tilde{O}(n\cdot \frac{d^{3}}{\eps^{2}\cdot t^{2}})$
		 edges.
		
		\item If $X$ are taken from $d$ dimensional $\ell_p$ space for $p\in[1,2]$, 
		 then for parameters $t\in[10,d^{\frac1p}]$ and $\nu\in(0,1)$, $X$ admits an oblivious $\nu$-reliable, $t$-spanner with $\nu^{-1}\cdot e^{O(\frac{d}{t^{p}})}\cdot\tilde{O}(n\cdot d^{2})$ edges.
		
		\item If $X$ are taken from $d$ dimensional $\ell_p$ space for $p\in[2,\infty]$, then for every $\nu\in(0,1)$,  $X$ admits an oblivious $\nu$-reliable, $2d^{1-\frac1p}$-spanner with $\nu^{-1}\cdot\tilde{O}\left(n\cdot d^{2}\right)$ edges.		
	\end{enumerate}
\end{corollary}

See a summary of all state of the art results on reliable spanner in \Cref{tab:reliableSpanners}.

\subsection{Light Spanners}\label{subsec:light}
Recently, Le and Solomon \cite{LS21Unified2} constructed a general framework for constructing light spanners from spanner oracles (originally introduced by \cite{Le20}). 
Here we show that spanner oracles can be constructed using LSO's, and thus derive new results on light spanners. 
See \Cref{tab:resultsLightSpanners} for a summary.
We begin with the required definitions, and meta-theorems from \cite{LS21Unified2}.

\begin{definition}[Spanner Oracle \cite{LS21Unified2}]\label{def:oracle} Let $(X,d_X)$ be a metric space and let $t > 1$ be a stretch parameter. A $t$-spanner oracle for $X$, given a subset of terminals  $T\subseteq X$ and a distance parameter $L > 0$, outputs a metric spanner $S$ over $X$ such that for every pair of vertices $x,y \in T, x\not= y$ with $L \leq d_G(x,y) < 2L$:
	\begin{equation*}
		d_{S}(x,y)\leq t\cdot d_X(x,y).
	\end{equation*}	
	We denote a $t$-spanner oracle for $X$ by  $\mathcal{O}_{X,t}$, and its output subgraph is denoted by $\mathcal{O}_{X,t}(T,L)$, given two parameters $T\subseteq X$ and $L >0$.
	
	The weak sparsity of a $t$-spanner oracle $\mathcal{O}_{X,t}$, denoted by $\wsp_{\mathcal{O}_{X,t}}$ 
	as follows:
	\begin{equation*}
		\wsp_{\mathcal{O}_{X,t}} = \sup_{T\subseteq X, L \in \real^+}\frac{w\left(\mathcal{O}_{X,t}(T,L)\right)}{|T|L}
	\end{equation*}	
\end{definition}

Le and Solomon proved two reduction from spanner oracles to light spanners. The first (\Cref{thm:LS21Unified2-general-stretch-2}) is for stretch $t>2$, while the second (\Cref{thm:LS21Unified2-general-stretch-1eps}) is for stretch $1+\eps$.
\begin{theorem}[\cite{LS21Unified2}]\label{thm:LS21Unified2-general-stretch-2} Let $(X,d_X)$ be a metric space with a $t$-spanner oracle $\mathcal{O}$ of weak sparsity $\wsp_{\mathcal{O}_{X,t}}$ for $t\geq 2$. Then for any $\eps > 0$, there exists a $t(1+\epsilon)$-spanner $H$ for $X$ with lightness $\wsp_{\mathcal{O}_{X,t}}\cdot\tilde{O}\left(\epsilon^{-1}\right)$.
\end{theorem}
\begin{theorem}[\cite{LS21Unified2}]\label{thm:LS21Unified2-general-stretch-1eps} Let $(X,d_X)$ be a metric space  with a $(1+\epsilon)$-spanner oracle $\mathcal{O}$ of weak sparsity $\wsp_{\mathcal{O}_{X,1+\epsilon}}$   for any $\eps > 0$. Then there exists an $(1+O(\epsilon))$-spanner $S$ for $X$ with lightness $\wsp_{\mathcal{O}_{X,t}}\cdot\tilde{O}\left(\epsilon^{-1}\right)+\tilde{O}\left(\epsilon^{-2}\right)$.
\end{theorem}

We next show that LSO and triangle LSO imply spanner oracles.
\begin{theorem}\label{thm:LSOToSparsityOracle}
	Consider a metric space $(X,d_X)$ admitting a $(\tau,\rho)$-LSO for $\rho<\frac14$, then $X$ admits a $1+8\rho$ spanner oracle $\mathcal{O}_{X,1+8\rho}$ with weak sparsity $\wsp_{\mathcal{O}_{X,1+8\rho}}\le2\tau$.
\end{theorem}
\begin{proof}
	Let $T\subseteq X$ be a subset of terminals, and $L>0$ a parameter.
	For a pair of terminals $u,v\in T$ and ordering  $\sigma\in\Sigma$, we say that $u,v$ are $\sigma$-close if $d_X(u,v)\le 2L$, and there is no terminal $z\in T$ between $u$ and $v$ in $\sigma$.
	Let $H$ be a spanner constructed by adding the edge $\{u,v\}$ for every pair $u,v$ which are $\sigma$-close w.r.t. some ordering $\sigma\in\Sigma$.
	
	First we bound the weight of $H$. Note that for every ordering $\sigma$, there are at most $|T|-1$ $\sigma$-close pairs. Further, the weight added to $H$ for each $\sigma$-close pair is at most $2L$, as there are only $\tau$ orderings, it follows that $w(H)\le (|T|-1)\cdot\tau\cdot 2L$. In particular, the weak sparsity of our spanner oracle is bounded by $2\tau$.
	
	Next we bound the stretch of our spanner by $1+8\rho$. 
	Our proof follow similar lines to the spanner construction from LSO in \cite{CHJ20}.
	The proof is by induction on the distance between two terminals. Consider a pair of terminals $u,v\in T$ such that $d_X(u,v)\le 2L$, and assume that for every pair of terminals $u',v'\in T$ such that  $d_X(u',v')<  d_X(u,v)$ it holds that  $d_H(u',v')\le (1+8\rho)\cdot d_H(u',v')$. Let $\sigma\in\Sigma$ be an ordering such that (w.l.o.g.) $u\preceq_{\sigma}v$. The points between $u$ and $v$ w.r.t. $\sigma$ could be partitioned into two consecutive intervals $I_{u},I_{v}$ where $I_{u}\subseteq B_{X}(u,\rho\cdot d_{X}(u,v))$ and $I_{v}\subseteq B_{X}(v,\rho\cdot d_{X}(u,v))$.
	If $u,v$ are $\sigma$-close, then $\{u,v\}\in H$ and we are done. Otherwise, let $u'$ (resp. $v'$) be rightmost (resp. leftmost) terminal in $I_u$ (resp. $I_v$). Note that it is possible that either $u'=u$ or $v'=v$. \\
	\begin{center}
		\includegraphics[width=0.7\textwidth]{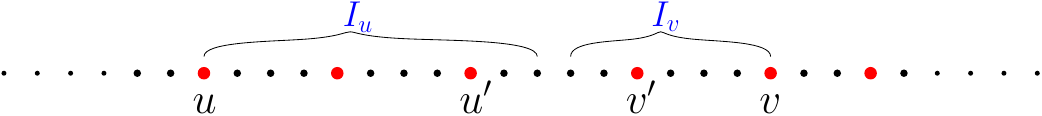}\\
	\end{center}
	Necessarily $u',v'$ are $\sigma$-close, thus  $d_H(u',v')=d_X(u',v')$. Furthermore, $d_X(u,u')\le \rho \cdot d_X(u,v)$, and hence by the induction hypothesis, $d_H(u,u')\le (1+8\rho)\cdot d_X(u,u')$. Similarly $d_H(v,v')\le (1+8\rho)\cdot d_X(v,v')$.
	By triangle inequality $d_{X}(u',v')\le d_{X}(u',u)+d_{X}(u,v)+d_{X}(v,v')\le(1+2\rho)\cdot d_{X}(u,v)$.
	We conclude
	\begin{align*}
		d_{H}(u,v) & \le d_{H}(u,u')+d_{H}(u',v')+d_{H}(v',v)\\
		& \le(1+8\rho)\cdot d_{X}(u,u')+d_{X}(u',v')+(1+8\rho)\cdot d_{X}(v',v)\\
		& \le2\rho(1+8\rho)\cdot d_{X}(u,v)+(1+2\rho)\cdot d_{X}(u,v)\\
		& \le(1+8\rho)\cdot d_{X}(u,v)~.
	\end{align*}
	where in the last inequality we used the assumption $\rho\le\frac{1}{4}$.
\end{proof}

\begin{theorem}\label{thm:TriangleLSOToSparsityOracle}
	Consider a weighted $n$-vertex graph $G=(V,E,w)$ that admits a $(\tau,\rho)$-triangle-LSO, then $G$ admits:
	\begin{OneLiners}
		\item A $2\rho$ spanner oracle  $\mathcal{O}_{G,2\rho}$ with weak sparsity $\wsp_{\mathcal{O}_{G,2\rho}}\le O(\rho\cdot\tau\cdot\log n)$.
		\item A $3\rho$ spanner oracle  $\mathcal{O}_{G,3\rho}$ with weak sparsity $\wsp_{\mathcal{O}_{G,3\rho}}\le O(\rho\cdot\tau\cdot\log\log n)$.
		\item A $4\rho$ spanner oracle  $\mathcal{O}_{G,4\rho}$ with weak sparsity $\wsp_{\mathcal{O}_{G,4\rho}}\le O(\rho\cdot\tau\cdot\log^*n)$.
	\end{OneLiners} 
\end{theorem}
\begin{proof}
	We begin by proving the first assertion of the theorem. The two others  will follows similar lines.
	Let $\Sigma$ be the $(\tau,\rho)$-LSO assumed by the theorem.
	For every ordering $\sigma\in \Sigma$, we form an unweighted path graph $P_{\sigma,T}$ with vertex set $T$ and the order of vertices along the path is $\sigma$ (delete all the vertices out of $T$). We construct a $2$-hop $f$-fault tolerant $1$-spanner $H_\sigma$ for $P_{\sigma,T}$ using \Cref{thm:2hopfaultTolerantPath}. 
	Let $H_\sigma^{\le2\rho L}$ be a spanner of $G$, where we keep from $H_\sigma$ only edges of weight at most $\rho \cdot2L$.
	Our spanner will be $H=\cup_{\sigma\in\Sigma}H_\sigma^{\le2\rho L}$.
	
	To bound the weight, note that by \Cref{thm:2hopfaultTolerantPath},  every $H_\sigma$ consist of $O(|T|\log|T|)$ edges, and thus $H_\sigma^{\le2\rho L}$ have weight bounded by $O(|T|\log|T|\cdot 2\rho L)$. It follows that the weight of $H$ is bounded by $O(|T|\log|T|\cdot \rho L\cdot \tau)$. We conclude:
	\begin{equation}\label{eq:WeightOracle}
		\wsp_{\mathcal{O}_{G,t}}=\sup_{T\subseteq V,L\in\real^{+}}\frac{w\left(\mathcal{O}_{G,t}(T,L)\right)}{|T|L}\le\sup_{T\subseteq V,L\in\real^{+}}\frac{O(|T|\log|T|\cdot\rho L\cdot\tau)}{|T|L}=O(\rho\cdot\tau\cdot\log n)~.		
	\end{equation}

	Next we argue that $H$ has stretch at most $2\rho$ w.r.t. $T$. Consider a pair of terminals $u,v\in T$ such that $d_X(u,v)\le 2L$.
	Let $\sigma\in\Sigma$ be an ordering such that (w.l.o.g.) $u\preceq_{\sigma}v$ and for every $a,b$ such that $u\preceq_{\sigma}a\preceq_{\sigma}b\preceq_{\sigma}v$ it holds that $d_X(a,b)\le \rho\cdot d_X(u,v)$.
	There is a terminal $x\in T$ such that  $u\preceq_{\sigma}x\preceq_{\sigma}v$ and $\{u,x\},\{x,v\}\in H_\sigma$. 
	As $d_X(u,x),d_X(v,x)\le \rho\cdot d_X(u,v)\le \rho\cdot 2L$,  $\{u,x\},\{x,v\}\in H_\sigma^{\le2\rho L}$.
	We conclude 
	\[
	d_{H}(u,v)\le d_{H_{\sigma}^{\le2\rho L}}(u,x)+d_{H_{\sigma}^{\le2\rho L}}(x,v)=d_{X}(u,x)+d_{X}(x,v)\le2\rho\cdot d_{X}(u,v)~,
	\]
	thus the first assertion of the theorem holds.
	
	Similarly to \Cref{thm:2hopPath}, one can construct $3$-hop $1$-spanner for the path with $O(n\log\log n)$ edges, and $4$-hop $1$-spanner for the path with $O(n\log^*n)$ edges (see \cite{Solomon13}).
	Taking the construction above, and replacing \Cref{thm:2hopPath} with these two construction immediately yields the other two assertions.
\end{proof}

\begin{remark}
	In fact, there is extremely slow growing function $\alpha_h(n)$, such that one can construct $h$-hop $1$-spanner for the path with $O(n\cdot\alpha_h(n))$ edges. We used here $\alpha_2(n)=\lceil\log n \rceil$, $\alpha_3(n)=\lceil\log\log n \rceil$, and $\alpha_4(n)=\log^*n$. However, there is a general definition, and interestingly for $h=\alpha(n)$, $\alpha_{\alpha(n)}(n)=O(1)$, where $\alpha(n)$ is essentially the inverse Ackermann function. We refer to \cite{Solomon13} (see also \cite{KLMS22}) for definitions. As a result, given a $(\tau,\rho)$-triangle-LSO, one can construct a $h\rho$ spanner oracle  $\mathcal{O}_{G,h\rho}$ with weak sparsity $\wsp_{\mathcal{O}_{G,h\rho}}\le O(\rho\cdot\tau\cdot\alpha_h(n))$, or a $\alpha(n)\cdot \rho$ spanner oracle  $\mathcal{O}_{G,h\rho}$ with weak sparsity $\wsp_{\mathcal{O}_{G,h\rho}}\le O(\rho\cdot\tau)$.
\end{remark}

Using the various LSO's (summarized in \Cref{tab:LSO}), and the meta
\Cref{thm:LSOToSparsityOracle,thm:TriangleLSOToSparsityOracle}  we obtain weak sparsity oracles.
Now, plugging in these weak sparsity oracles into the meta \Cref{thm:LS21Unified2-general-stretch-1eps,thm:LS21Unified2-general-stretch-2} of \cite{LS21Unified2} we conclude:

\begin{corollary}\label{cor:LightSpannerDoubling}
	Consider an $n$ point metric space $(X,d_X)$ with doubling dimension $d$. Then $\forall\eps\in(0,\frac12)$, $X$ admits $(1+\eps)$-spanner oracle with weak sparsity $\eps^{-O(d)}$. It follows that $X$ admits $(1+\eps)$-spanner with  lightness $\eps^{-O(d)}$.\\
	Furthermore, for $t=\Omega(1)$, $X$ admits a $t$-spanner oracle with weak sparsity $2^{O(\nicefrac{d}{t})}\cdot d\cdot\log^{2}t\cdot\log^*n$.
	It follows that $X$ admits $O(t)$-spanner with lightness $2^{O(\nicefrac{d}{t})}\cdot d\cdot\log^{2}t\cdot\log^*n$.	
\end{corollary}

Using our LSO for $d$-dimensional Euclidean space, we conclude, 
\begin{corollary}\label{cor:highDimEuclidean}
	For every parameters $t\in[4,2\sqrt{d}]$ and $\eps\in(0,1)$, every set of $n$-point in $\R^d$, admits:
	\begin{OneLiners}
		\item $(1+\eps)$-spanner oracle with weak sparsity $O_d(\epsilon^{-d})\cdot \log\frac{1}{\epsilon}$, and thus $(1+\eps)$-spanner with lightness $O_d(\epsilon^{-d-1})\cdot \polylog\frac{1}{\epsilon}$.
		\item  $(1+\eps)2t$-spanner oracle with weak sparsity $e^{\frac{d}{2t^{2}}\cdot(1+\frac{2}{t^{2}})}\cdot\tilde{O}(\frac{d^{1.5}}{\eps})\cdot\log n$, and thus $(1+\eps)4t$-spanner with lightness $e^{\frac{d}{2t^{2}}\cdot(1+\frac{2}{t^{2}})}\cdot\tilde{O}(\frac{d^{1.5}}{\eps^{2}})\cdot\log n$.
		\item $(1+\eps)3t$-spanner oracle with weak sparsity $e^{\frac{d}{2t^{2}}\cdot(1+\frac{2}{t^{2}})}\cdot\tilde{O}(\frac{d^{1.5}}{\eps})\cdot\log\log n$, and thus $(1+\eps)6t$-spanner with lightness  $e^{\frac{d}{2t^{2}}\cdot(1+\frac{2}{t^{2}})}\cdot\tilde{O}(\frac{d^{1.5}}{\eps^{2}})\cdot\log\log n$.
		\item $(1+\eps)4t$-spanner oracle with weak sparsity $e^{\frac{d}{2t^{2}}\cdot(1+\frac{2}{t^{2}})}\cdot\tilde{O}(\frac{d^{1.5}}{\eps})\cdot\log^* n$-spanner with lightness $e^{\frac{d}{2t^{2}}\cdot(1+\frac{2}{t^{2}})}\cdot\tilde{O}(\frac{d^{1.5}}{\eps^{2}})\cdot\log^* n$.		
	\end{OneLiners}
\end{corollary}

Using our LSO for $\ell_p$ spaces we conclude:
\begin{corollary}\label{cor:LightSpannerLp12}
	For every parameters $t\ge 21$, every set of $n$-point in $d$ dimensional $\ell_p$ space for $p\in[1,2]$ admits
	$t$-spanner oracle with weak sparsity $e^{O(\frac{d}{t^{p}})}\cdot\tilde{O}(d\cdot t)\cdot\log^{*}n$, and thus $t$-spanner with lightness $e^{O(\frac{d}{t^{p}})}\cdot\tilde{O}(d\cdot t)\cdot\log^{*}n$.
\end{corollary}
\begin{corollary}\label{cor:LightSpannerLp2infty}
	Every set of $n$-point in $d$ dimensional $\ell_p$ space for $p\in[2,\infty]$ admits
	$4\cdot d^{1-\frac1p}$-spanner oracle with weak sparsity $\tilde{O}(d^{2-\frac{1}{p}})\cdot\log^{*}n$, and thus $5\cdot d^{1-\frac1p}$-spanner with lightness $\tilde{O}(d^{2-\frac{1}{p}})\cdot\log^{*}n$.
\end{corollary}

{\small 
	\bibliographystyle{alphaurlinit}
	\bibliography{LSObib}}

\appendix

\section{Missing Proofs from \Cref{sec:LSOForEuclidean}}\label{appendix:Missing}
This appendix contains the missing proofs from \Cref{sec:LSOForEuclidean}, restated for convenience.
\HypersphericalCup*
\begin{proof}
	By scaling, it is enough to prove the claim for $r=1$. To calculate
	the volume of the cap we will simply integrate over the values of
	the first coordinate, which take values in $[\alpha,1]$. Once the
	first coordinate is fixed to be $x$, the rest of the coordinate constitute a
	$d-1$ dimensional ball of radius $\sqrt{1-x^{2}}$. We have that
	
	\begin{align*}
		\text{Vol}\left(\text{cup at distance }\alpha\right) & =\int_{\alpha}^{1}V_{d-1}\cdot\sqrt{1-x^{2}}^{d-1}dx\\
		& =V_{d-1}\cdot\int_{0}^{1-\alpha^{2}}\frac{y^{\frac{d-1}{2}}}{2\sqrt{1-y}}dy\\
		& \ge V_{d-1}\cdot\int_{1-(\frac{3}{2}\alpha)^{2}}^{1-\alpha^{2}}\frac{y^{\frac{d-1}{2}}}{2\cdot\frac{3}{2}\alpha}dy\\
		& =V_{d}\cdot\frac{V_{d-1}}{V_{d}}\cdot\frac{2}{d+1}\cdot\frac{1}{3\alpha}\cdot y^{\frac{d+1}{2}}\mid_{1-(\frac{3}{2}\alpha)^{2}}^{1-\alpha^{2}}\\
		& =\Omega(1)\cdot V_{d}\cdot\frac{1}{\alpha\cdot\sqrt{d}}\cdot(1-\alpha^{2})^{\frac{d+1}{2}}\cdot\left(1-\left(\frac{1-(\frac{3}{2}\alpha)^{2}}{1-\alpha^{2}}\right)^{\frac{d+1}{2}}\right)~,
	\end{align*}
	where in the second equality follows by substitution $y=1-x^{2}$,
	and the last equality holds by the fact $\frac{V_{d-1}}{V_{d}}=\Omega(\sqrt{d})$.
	The claim holds as $\alpha\ge\frac{1}{2\sqrt{d}}$, and
	thus $1-\left(\frac{1-(\frac{3}{2}\alpha)^{2}}{1-\alpha^{2}}\right)^{\frac{d+1}{2}}=1-\left(1-\frac{\frac{5}{4}\cdot\alpha^{2}}{1-\alpha^{2}}\right)^{\frac{d+1}{2}}\ge1-\left(1-\frac{\frac{5}{4}\cdot\frac{1}{4d}}{1-\frac{1}{4d}}\right)^{\frac{d+1}{2}}=\Omega(1)$.
\end{proof}

\BallIntersectionLpLarge*
\begin{proof}
	The following is a know fact ($\ell_{p}$ is $p$-uniformly smooth, \footnote{For $p\in(1,2]$, \cref{eq:p12-uniformly smoth} holds for $\ell_{p}$ space due to uniform smoothness. On the other hand, one can easily verify that \cref{eq:p12-uniformly smoth} holds for absolute value, and hence also for $\ell_1$.}
	see e.g. \cite{BCL02}): 
	\begin{equation}
		\text{For }p\in[1,2]\text{ and }x,y\in\ell_{p}:\qquad\frac{\|x\|_{p}^{p}+\|y\|_{p}^{p}}{2}\le\|\frac{x+y}{2}\|_{p}^{p}+\|\frac{x-y}{2}\|_{p}^{p}~.\label{eq:p12-uniformly smoth}
	\end{equation}
	By symmetry we can assume that $w_i=1$.
	\begin{claim}\label{clm:BallLp12Intersectin}
		For $p\in[1,2]$, and $\|x-y\|_{p}\le\frac{1}{t}<\frac{1}{3}$, it
		holds that\\\hfill$B_{p}\left(\frac{x+y}{2},\left(1-\frac{3}{(2t)^{p}}\right)^{\frac{1}{p}}\right)\subseteq\left(B_{p}\left(x,1\right)\cap B_{p}\left(y,1\right)\right)\cup B_{p}\left(x,\left(1-\frac{4}{(2t)^{p}}\right)^{\frac{1}{p}}\right)\cup B_{p}\left(y,\left(1-\frac{4}{(2t)^{p}}\right)^{\frac{1}{p}}\right)$.
	\end{claim}
	
	\begin{proof}\sloppy
		Suppose for contradiction otherwise, then there is  point $z\in B_{p}\left(\frac{x+y}{2},\left(1-\frac{3}{(2t)^{p}}\right)^{\frac{1}{p}}\right)\setminus\left(B_{p}\left(x,1\right)\cap B_{p}\left(y,1\right)\right)\cup B_{p}\left(x,\left(1-\frac{4}{(2t)^{p}}\right)^{\frac{1}{p}}\right)\cup B_{p}\left(y,\left(1-\frac{4}{(2t)^{p}}\right)^{\frac{1}{p}}\right)$.
		W.l.o.g. $\|x-z\|_{p}\le\|y-z\|_{p}$. Hence
		\begin{align*}
			\|y-z\|_{p} & >1\\
			\|x-z\|_{p} & >\left(1-\frac{4}{(2t)^{p}}\right)^{\frac{1}{p}}\\
			\|\frac{x+y}{2}-z\|_{p} & \le\left(1-\frac{3}{(2t)^{p}}\right)^{\frac{1}{p}}~,
		\end{align*}
		applying \cref{eq:p12-uniformly smoth} to $x-z$ and $y-z$
		we have 
		\begin{align*}
			\|\frac{x-y}{2}\|_{p}^{p} & \ge\frac{\|x-z\|_{p}^{p}+\|y-z\|_{p}^{p}}{2}-\|\frac{x+y}{2}-z\|_{p}^{p}\\
			& >\frac{1+\left(1-\frac{4}{(2t)^{p}}\right)}{2}-\left(1-\frac{3}{(2t)^{p}}\right)=\frac{1}{(2t)^{p}}\,,
		\end{align*}
		implying that $\|x-y\|_{p}>\frac{1}{t}$, a contradiction.
	\end{proof}
	
	Denote by $V_{p}$ the volume of the $\ell_{p}$ unit ball. Note that
	the volume of a ball of radius $t$ equals $t^{d}\cdot V_{p}$. By \Cref{clm:BallLp12Intersectin} it
	follows that 
	\begin{align*}
		{\rm \frac{{\rm Vol}\left(B_{p}\left(x,1\right)\cap B_{p}\left(y,1\right)\right)}{V_{p}}} & \ge\frac{1}{V_{p}}\cdot{\rm Vol}\left(B_{p}\left(\frac{x+y}{2},\left(1-\frac{3}{(2t)^{p}}\right)^{\frac{1}{p}}\right)\right)-\frac{2}{V_{p}}\cdot{\rm Vol}\left(B_{p}\left(x,\left(1-\frac{4}{(2t)^{p}}\right)^{\frac{1}{p}}\right)\right)\\
		& =\left(1-\frac{3}{(2t)^{p}}\right)^{\frac{d}{p}}-2\cdot\left(1-\frac{4}{(2t)^{p}}\right)^{\frac{d}{p}}\\
		& =\left(1-\frac{3}{(2t)^{p}}\right)^{\frac{d}{p}}\cdot\left(1-2\cdot\left(1-\frac{1}{(2t)^{p}-3}\right)^{\frac{d}{p}}\right)=\Omega(1)\cdot e^{-\frac{6}{2^{p}\cdot p}\cdot\frac{d}{t^{p}}}~,
	\end{align*}
	where the last equality follows by the fact $1-x\ge e^{-2x}$, and
	as $t\ge3$ and thus for large enough $d$, $2\cdot\left(1-\frac{1}{(2t)^{p}-3}\right)^{\frac{d}{p}}\le\frac{1}{2}$.
	As ${\rm \frac{{\rm Vol}\left(B_{p}\left(x,1\right)\cup B_{p}\left(x,1\right)\right)}{V_{p}}}\le 2\cdot V_p$, the lemma follows.
\end{proof}

\BallIntersectionLpSmall* 
\begin{proof}
	By triangle inequality, for every $z\in\mathbb{R}^{d}$, $\|z-y\|_{p}\le\|z-x\|_{p}+\|x-y\|_{p}$,
	and hence $B_{p}\left(x,w_{i}-\|x-y\|_{p}\right)\subseteq B_{p}\left(x,w_{i}\right)\cap B_{p}\left(y,w_{i}\right)$.
	In particular,
	\begin{align*}
		\text{Vol}_{d}\left(B_{p}(x,w_{i})\cap B_{p}(y,w_{i})\right) & \ge\text{Vol}_{d}\left(B_{p}(x,w_{i}-\|x-y\|_{p})\right)\\
		& =\left(1-\frac{\|x-y\|_{p}}{w_{i}}\right)^{d}\cdot\text{Vol}_{d}\left(B_{p}(x,w_{i})\right)\ge\left(1-2d\cdot\frac{\|x-y\|_{p}}{w_{i}}\right)\cdot\text{Vol}_{d}\left(B_{p}(x,w_{i})\right)~,
	\end{align*}
	where the last inequality holds as $1-x\le e^{-x}\le1-\frac{x}{2}$.
	Denote $\text{Vol}_{d}\left(B_{p}(x,w_{i})\right)=Q$, and $\alpha\cdot Q=\text{Vol}_{d}\left(B_{p}(x,w_{i})\cap B_{p}(y,w_{i})\right)$.
	Set $\beta=1-\alpha$, then $\beta\le2d\cdot\frac{\|x-y\|_{p}}{w_{i}}$.
	It follows that 
	\[
	\frac{\text{Vol}_{d}\left(B_{p}(x,w_{i})\cap B_{p}(y,w_{i})\right)}{\text{Vol}_{d}\left(B_{p}(x,w_{i})\cup B_{p}(y,w_{i})\right)}=\frac{\alpha\cdot Q}{2Q-\alpha\cdot Q}=\frac{1-\beta}{1+\beta}\ge1-2\beta\ge1-4d\cdot\frac{\|x-y\|_{p}}{w_{i}}~.
	\]
\end{proof}

\end{document}